\documentclass[reprint,superscriptaddress,aps,prx,twocolumn,longbibliography,nofootinbib]{revtex4-2}
\usepackage{graphicx}
\usepackage{mathrsfs}
\usepackage{graphicx}
\usepackage{braket}
\usepackage{amsfonts}
\usepackage{amssymb}
\usepackage{amsthm}
\usepackage{amsmath}
\usepackage{mleftright}
\usepackage{dcolumn}
\usepackage{color}
\usepackage{xcolor}
\colorlet{RED}{red}
\colorlet{BLACK}{black}
\usepackage{verbatim}
\usepackage{cancel}
\usepackage[normalem]{ulem}
\usepackage[breaklinks,
            colorlinks,
            urlcolor=blue,
            linkcolor=blue,
            anchorcolor=blue,
            citecolor=blue]{hyperref}

\usepackage{soul}
\usepackage{apptools}  
\usepackage{appendix}

\newcommand\norm[1]{\left\lVert#1\right\rVert}
\graphicspath{{./images/},{./imagesAppendix/}}

\newcommand{\change}[1]{{\color{black} #1}}

\newif\ifshowjp
 \showjptrue       %

\newif\ifshowdmremark
 \showdmremarktrue       %

\theoremstyle{definition}
\newtheorem{theorem}{Theorem}
\newtheorem{lemma}{Lemma}
\newtheorem{proposition}{Proposition}
\newtheorem{corollary}{Corollary}
\newtheorem{definition}{Definition}

\renewcommand{\epsilon}{\varepsilon}

\newcommand{\Eset}[1]{\underset{#1}{\mathbb{E}}}

\newcommand{\parens}[1]{\left(#1\right)}
\newcommand{\sparens}[1]{\left[ #1 \right]}
\newcommand{\bparens}[1]{\left\{ #1 \right\}}
\newcommand{\abs}[1]{\left\lvert #1 \right\rvert}
\DeclareMathOperator{\E}{\mathbb{E}}

\newcommand{\bigO}[1]{\mathcal{O}\mleft(#1\mright)}
\newcommand{\bigOtilde}[1]{\tilde{\mathcal{O}}\mleft(#1\mright)}

\DeclareMathOperator*{\argmin}{arg\,min}

\DeclareMathOperator{\Tr}{Tr}

\definecolor{THc}{rgb}{0.9,0.3,0.2}

\newcommand{\idg}[1]{{\bfseries #1)}}
\newcommand{\subfigimg}[3][,]{%
	\setbox1=\hbox{\includegraphics[#1]{#3}}%
	\leavevmode\rlap{\usebox1}%
	\rlap{\hspace*{2pt}\raisebox{\dimexpr\ht1-0.5\baselineskip}{{\bfseries \large\textsf{#2}}}}%
	\phantom{\usebox1}%
}

\newcommand{\SM}{Appendix}

\makeatletter
\DeclareRobustCommand\ket[1]{%
  \@ifnextchar\bra{\k@t{#1}\!}{\k@t{#1}}%
}
\newcommand\k@t[1]{{|{#1}\rangle}}
\makeatother

\begin{document}

\title{Nature is stingy: Universality of Scrooge ensembles in quantum many-body systems}

\author{Wai-Keong Mok}
\affiliation{Institute for Quantum Information and Matter, California Institute of Technology, Pasadena, CA 91125, USA}
\author{Tobias Haug}
\affiliation{Quantum Research Centre, Technology Innovation Institute, Abu Dhabi, UAE}
\author{Wen Wei Ho}
\affiliation{Department of Physics, National University of Singapore, Singapore 117551}
\affiliation{Centre for Quantum Technologies, National University of Singapore, Singapore 117543}
\author{John Preskill}
\affiliation{Institute for Quantum Information and Matter, California Institute of Technology, Pasadena, CA 91125, USA}
\affiliation{AWS Center for Quantum Computing, Pasadena CA 91125}

\begin{abstract}

Recent advances in quantum simulators allow direct experimental access to ensembles of pure states generated by measuring part of an isolated quantum many-body system. These projected ensembles encode fine-grained information beyond thermal expectation values and provide a new window into quantum thermalization. In chaotic dynamics, projected ensembles exhibit universal statistics governed by maximum-entropy principles, known as deep thermalization. At infinite temperature this universality is characterized by Haar-random ensembles. More generally, physical constraints such as finite temperature or conservation laws lead to Scrooge ensembles, which are maximally entropic distributions of pure states consistent with these constraints. Here we introduce Scrooge $k$-designs, which approximate Scrooge ensembles, and use this framework to sharpen the conditions under which Scrooge-like behavior emerges. We first show that global Scrooge designs arise from long-time chaotic unitary dynamics alone, without measurements. Second, we show that measuring a complementary subsystem of a scrambled global state drawn from a global Scrooge $2k$-design induces a local Scrooge $k$-design. Third, we show that a local Scrooge $k$-design arises from an arbitrary entangled state when the complementary system is measured in a scrambled basis induced by a unitary drawn from a Haar $2k$-design. These results show that the resources required to generate approximate Scrooge ensembles scale only with the desired degree of approximation, enabling efficient implementations. Complementing our analytical results, numerical simulations identify coherence, entanglement, non-stabilizerness, and information scrambling as essential ingredients for the emergence of \change{local} Scrooge-like behavior. Together, our findings advance theoretical explanations for maximally entropic, information-stingy randomness in quantum many-body systems.

 \end{abstract}
\maketitle

\let\oldaddcontentsline\addcontentsline%
\renewcommand{\addcontentsline}[3]{}%

\section{Introduction}

Understanding universal behavior in complex systems is a central goal of physics. In closed quantum many-body systems, generic unitary dynamics is expected to drive local subsystems toward equilibrium, so that local observables are well described by generalized Gibbs ensembles determined by conserved quantities~\cite{nandkishor2015many,abanin2019many}. Explaining how irreversibility emerges from unitarity has led to powerful concepts, such as the eigenstate thermalization hypothesis ~\cite{deutsch1991quantum,srednicki1994chaos,rigol2008thermalization} and the maximum-entropy principle of statistical mechanics~\cite{alessio2016quantum,borgonovi2016quantum,mori2018thermalization,ueda2020quantum,gogolin2016equilibration,jaynes1957information,grandy1980principle,martyushev2006maximum,banavar2010applications,presse2013principles}.

\begin{figure*}
    \centering
    \includegraphics[width=0.85\linewidth]{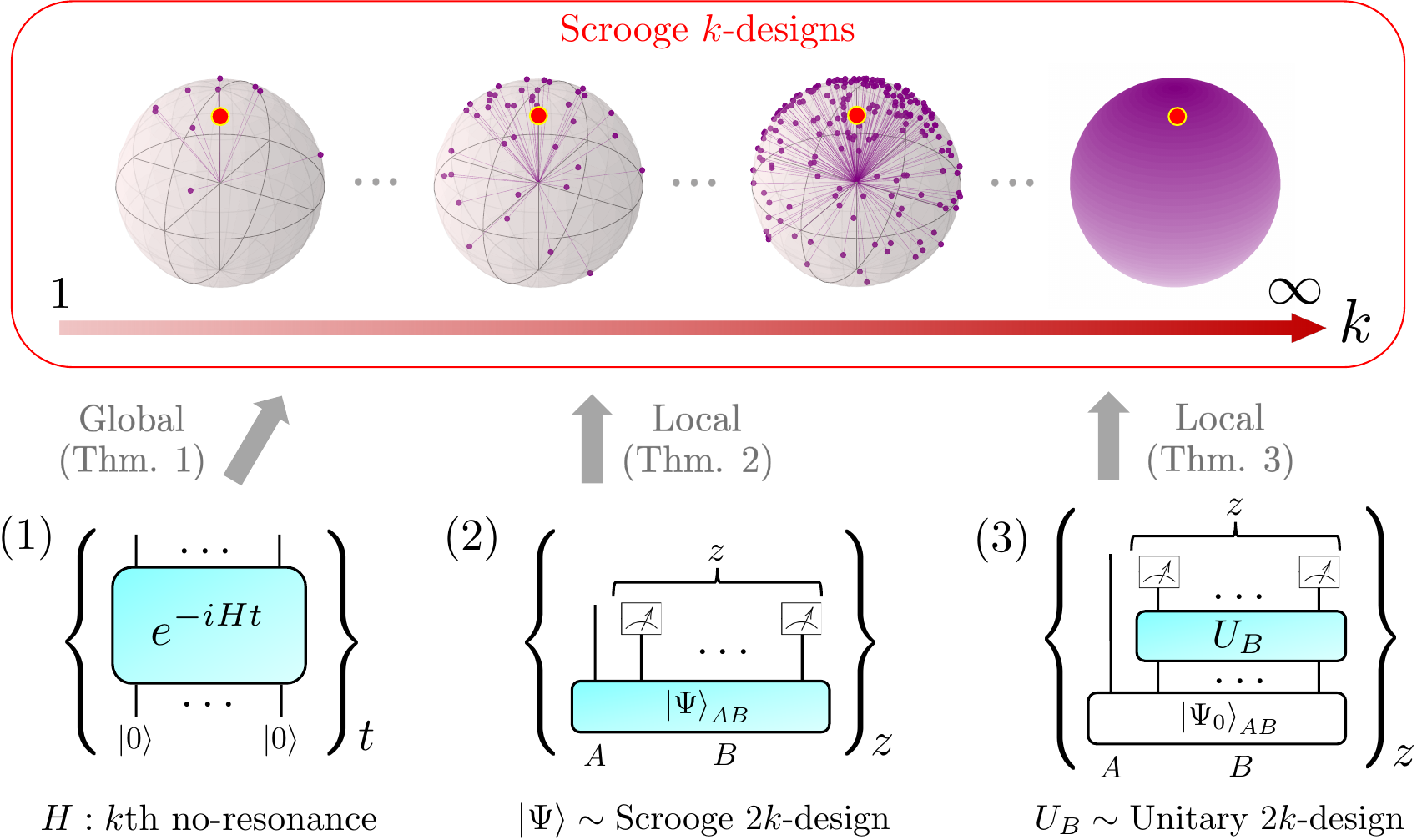}
    \caption{\textbf{Scrooge $k$-designs in temporal and projected ensembles.} 
    (Top panel) Scrooge ensembles are distributions of random pure states that reveal minimal classical information, amongst all ensembles that realize a fixed density operator $\sigma$. 
    In this work, we introduce Scrooge $k$-designs, which approximate Scrooge ensembles up to the $k$th moment. For the illustrative case of  a single qubit shown here, this is depicted by purple points over the Bloch sphere; in all cases illustrated, the  density operator $\sigma$ is common and represented by a red circle. Higher Scrooge $k$-designs usually require more states and complexity in the ensemble, converging to the continuous Scrooge ensemble in the $k \to \infty$ limit. 
    (Bottom panel):
    We rigorously show that Scrooge $k$-designs emerge in the following physical settings, commonly considered in the literature in the context of Hilbert-space ergodicity and deep thermalization: (1) Temporal ensemble $\{e^{-iHt}\ket{0}\}_{t > 0}$ obtained by evolving an initial reference state $\ket{0}$ with a Hamiltonian $H$ satisfying the $k$th no-resonance condition, for late times $t$. (2) Projected ensemble generated by a state $\ket{\Psi}_{AB}$ drawn from a global Scrooge $2k$-design, with subsystem $B$ measured in an arbitrary fixed basis. (3) Projected ensemble generated by an arbitrary bipartite entangled state $\ket{\Psi_0}_{AB}$, with subsystem $B$ scrambled by a unitary $U_B$ drawn from a Haar $2k$-design, prior to measurements in an arbitrary fixed basis. We prove rigorously that the temporal ensemble (1) forms a global Scrooge $k$-design [Theorem~\ref{thm:global_scrooge}], the projected ensemble in (2) forms a probabilistic mixture of local Scrooge $k$-designs [Theorem~\ref{thm:2kgenerator}], and the projected ensemble in (3) forms a local Scrooge $k$-design [Theorem~\ref{thm:ScroogeByMeasBasis}].} 
    \label{fig:schematic}
\end{figure*}

Recently developed quantum simulators provide access to fine-grained information going beyond thermal expectation values~\cite{parsons2016site,choi2023preparing,schauss2012observation,fossfeig2024progress,arute2019quantum,yan2025characterizing}. In particular, one can probe experimentally the features of \textit{projected ensembles}~\cite{cotler2023emergent,choi2023preparing}: given a many-body pure quantum state on the composite system $AB$, one measures the complementary system $B$ in a fixed local basis, hence sampling from an ensemble of conditional pure states on $A$. Averaging over these pure states, weighted by the corresponding measurement outcome probabilities, yields the marginal density operator $\sigma_A$. However, the structure of the projected ensemble encodes additional properties of the composite state that are not captured by $\sigma_A$ alone.

In chaotic quantum systems with no relevant conservation laws, or at very high temperatures, 
analytical, numerical, and experimental studies (including some rigorous mathematical proofs) have established that the projected ensemble on $A$ becomes indistinguishable from the Haar distribution at late times~\cite{cotler2023emergent,choi2023preparing,ippoliti2022solvable,ho2022exact,wilming2022high,claeys2022emergent,shrotriya2023nonlocality,liu2024deep,bhore2023deep,ippoliti2022solvable,lucas2023generalized,mark2024maximum,chan2024projected,chang2025deep,varikuti2024unraveling,mok2025optimal,liu2025coherence,zhang2025holographic,yan2025characterizing,chakraborty2025fast,bejan2025matchgate,manna2025projected,vairogs2025extracting,loio2025quantum,goldstein2006distribution,goldstein2016universal,varikuti2025deep}.
The Haar ensemble is
a collection of states distributed uniformly randomly over the Hilbert space. Thus, not only is $\sigma_A$ maximally mixed, the projected ensemble is maximally random as well.

In the presence of constraints such as energy or charge conservation, the situation is much richer, as uniform Haar randomness is no longer attainable. 
Instead, pioneering work by Goldstein et al.~\cite{goldstein2006distribution,goldstein2016universal}, 
based on canonical typicality~\cite{popescu2006entanglement,goldstein2006canonical}, argued that the resulting limiting projected ensemble is an analog of the Haar distribution, distorted to account for the constraints, called the Scrooge ensemble.\footnote{In Refs.~\cite{goldstein2006distribution,goldstein2016universal} this ensemble is called the ``Gaussian Adjusted Projected (GAP)'' ensemble.} Among all ensembles realizing a fixed density operator $\sigma_A$, the Scrooge ensemble is the most stingy in that
it minimizes the classical information that can be accessed by measuring $A$~\cite{josza1994lower,scrooge}. Recent work by Mark et al.~\cite{mark2024maximum} further argued that Scrooge ensembles emerge universally in broad physical settings such as in late-time chaotic quantum dynamics, and that their emergence adheres to a generalized maximum entropy principle.

The appearance of such maximally random projected ensembles has been observed across a host of physically distinct systems like quantum spin systems~\cite{cotler2023emergent,manna2025projected,mark2024maximum,chang2025deep}, fermionic systems~\cite{lucas2023generalized}, and bosonic continuous-variable systems~\cite{liu2024deep}, constituting a strong form of universality that has come to be known as \textit{deep thermalization}, a burgeoning and active area of research~\cite{cotler2023emergent,choi2023preparing,ippoliti2022solvable,ho2022exact,wilming2022high,claeys2022emergent,shrotriya2023nonlocality,liu2024deep,bhore2023deep,ippoliti2022solvable,lucas2023generalized,mark2024maximum,chan2024projected,chang2025deep,varikuti2024unraveling,mok2025optimal,liu2025coherence,zhang2025holographic,yan2025characterizing,chakraborty2025fast,bejan2025matchgate,manna2025projected,vairogs2025extracting,loio2025quantum,goldstein2006distribution,goldstein2016universal,varikuti2025deep}. %

Given the apparent ubiquity of Scrooge behavior in nature, 
a compelling question is not simply whether such behavior can arise in principle, but how robustly it arises, and what resources are required in practice for the emergence of ensembles that accurately approximate Scrooge ensembles in low-complexity, experimentally accessible regimes.
Existing works, however, rely on idealized assumptions, including the preparation of highly complex many-body states or measurement bases~\cite{goldstein2016universal}, or chaotic evolution over infinitely-long times~\cite{mark2024maximum}. These assumptions limit their applicability to realistic experiments, leaving open the question of whether Scrooge behavior can emerge robustly with finite resources.

In this paper, we expand on these insightful earlier contributions by developing a unified and rigorous framework based on the notion of \textit{approximate Scrooge $k$-designs}.
State $k$-designs, much studied in quantum information theory~\cite{dankert2009exact,gross2007evenly}, are distributions on pure states whose first $k$ moments match those of the Haar distribution; hence a $k$-design cannot be distinguished from the Haar ensemble by any measurement acting on at most $k$ copies of the state. Likewise, a Scrooge $k$-design matches the first $k$ moments of the Scrooge ensemble, and an approximate Scrooge $k$-design is close to an exact one in an appropriate sense. This concept provides a quantitative tool for studying the emergence of Scrooge behavior for systems that have finite size and that evolve for a finite time. Within this framework, we prove three main theorems that provide sufficient conditions for the emergence of Scrooge behavior in several different physical settings. These results are summarized in Fig.~\ref{fig:schematic}.

Our first theorem concerns the \textit{temporal ensemble}, which describes the states reached when a closed quantum system governed by a specified Hamiltonian evolves for a random amount of time.
Assuming the absence of spectral resonances, Mark et al.~\cite{mark2024maximum} showed that averaging over very long evolution times is equivalent to averaging over the  \textit{random phase ensemble}~\cite{nakata2012phase,goldstein2010long,linden2009quantum}, consisting of superpositions of energy eigenstates with random phases, a phenomenon called \textit{Hilbert-space ergodicity}~\cite{pilatowskycameo2024hilbert,pilatowskycameo2023complete,pilatowskycameo2025critically,shaw2025experimental,liu2025observation}. Moreover, Ref.~\cite{mark2024maximum} derived limiting expressions for the statistical moments of this ensemble, and showed that they assume a form similar to, though not exactly equal to, those of the Scrooge ensemble. 
Our Theorem~\ref{thm:global_scrooge} rigorously formalizes this closeness in the language of Scrooge designs, with corresponding quantum information-theoretic error bounds:

\begin{itemize}
\item Theorem~\ref{thm:global_scrooge} (\textit{Global Scrooge from chaotic dynamics, informal}).
If the temporal ensemble generated by a generic chaotic quantum system 
has low purity ($k^2 \norm{\sigma}_2 \ll 1$, where $\sigma$ is the density operator realized by the ensemble), then it forms an approximate Scrooge $k$-design. This result establishes that Scrooge-like universality can already arise naturally from unitary quantum dynamics, beyond the context of projected ensembles.
\end{itemize}

We next consider projected ensembles in the context of deep thermalization.  Goldstein et al.~\cite{goldstein2016universal} studied the projected ensemble in a bipartite system $AB$ in two idealized settings. First they considered global states drawn from an \textit{exact} Scrooge ensemble, and showed that, when system $B$ is measured in a fixed basis, the projected ensemble on $A$ converges weakly to the Scrooge ensemble in the thermodynamic limit $|B| \to \infty$. Second, they demonstrated that a Scrooge ensemble on $A$ emerges from an {\it arbitrary} bipartite entangled state on $AB$ if a unitary drawn from the \textit{exact} Haar ensemble is applied to $B$ prior to measuring.

In practice, however, exact Scrooge  ensembles do not arise in finite systems or after finite-time evolution, nor is exact Haar randomness available in the laboratory.
Furthermore, experimentally validating an  exact Scrooge ensemble would require collective measurements on   arbitrarily many copies of the system.  
Our framework of Scrooge $k$-designs allows us to extend the settings considered by Goldstein et al.~to an experimentally relevant regime. Adopting modern quantum information-theoretic notions of convergence, we build a {\it quantitative} theory of the resources needed for deep thermalization that applies under realistic conditions. Our results are
captured by the following Theorems. 

\begin{itemize}
\item  Theorem~\ref{thm:2kgenerator} (\textit{Local Scrooge design from global Scrooge design, informal}).
Consider a single global state drawn from {\it any} approximate Scrooge $2k$-design with small relative error. When $B$ is measured in a fixed basis, the resulting projected ensemble on $A$ is a probabilistic mixture of approximate Scrooge $k$-designs with high probability. Because the design order $k$ may be regarded as a proxy for quantum  complexity~\cite{brandao2016local,mcginley2025scrooge}, this result indicates that Scrooge-like behavior of the projected ensemble is expected when the global state is sufficiently scrambled~\cite{mcginley2025scrooge}. In the special case where the global ensemble approximates the Haar distribution, Theorem~\ref{thm:2kgenerator} shows that a global $AB$ state drawn from an approximate Haar $2k$-design yields a projected ensemble on $A$ that is an approximate Haar $k$-design, strengthening the results in Refs.~\cite{cotler2023emergent,ghosh2025design}. Theorem~\ref{thm:2kgenerator} reduces to Goldstein et al.'s result \cite{goldstein2016universal} in the limit $k\to\infty$, where the complexity of the global state may become exponential in the size of $AB$.
\end{itemize}

\begin{itemize}
\item Theorem~\ref{thm:ScroogeByMeasBasis} (\textit{Local Scrooge design from scrambled measurements, informal}).
Consider an {\it arbitrary} entangled state of a bipartite quantum system $AB$ and  suppose that a unitary transformation drawn from an approximate unitary $2k$-design is applied to $B$ prior to measuring in a fixed basis. Then the resulting projected ensemble on $A$ forms an approximate Scrooge $k$-design with high probability. Complementing Theorem~\ref{thm:2kgenerator}, Theorem~\ref{thm:ScroogeByMeasBasis} shifts the focus from the randomness of the global state to the randomness of the measurement basis. 
Combined with recent efficient constructions of unitary designs~\cite{brandao2016local,schuster2025extremely}, this result provides a practical scheme for generating Scrooge designs using quantum circuits whose depth scales logarithmically with system size. Theorem~\ref{thm:ScroogeByMeasBasis} reduces to Goldstein et al.'s result \cite{goldstein2016universal} in the limit $k\to\infty$, where the quantum complexity of the scrambling unitary becomes exponential in the size of $B$.
\end{itemize}

Our three theorems delineate, in a unified and systematic way, several physically distinct many-body settings in which hierarchies of Scrooge behavior arise.
In particular, Theorems~\ref{thm:global_scrooge} and~\ref{thm:2kgenerator} together suggest that projected ensembles formed from late-time temporal states generically exhibit Scrooge-like behavior. 
This physically natural conclusion was anticipated by Mark et al.~\cite{mark2024maximum} through an analysis of the distribution of unnormalized post-measurement states, but had not previously been established at the level of normalized projected ensembles. Using the techniques developed in this work, we rigorously confirm this expectation, though strictly speaking not directly from the theorems.

Finally, complementing our analytical theorems, we further clarify the physical resources needed for Scrooge behavior to emerge in projected ensembles. We  argue that coherence, entanglement, non-stabilizerness (magic), and information scrambling are all necessary ingredients: removing any one of them can obstruct Scrooge-like universality. Conversely, when these resources are present, local Scrooge behavior can arise even in systems that do not thermalize dynamically. 
Extensive numerical simulations in several quantum many-body settings, including commuting circuits, doped Clifford circuits, and Hamiltonian ground states, are performed to illustrate these points.  In particular, we demonstrate that ground states of one-dimensional integrable Hamiltonians --- whose area laws and symmetries preclude thermalization --- can nevertheless generate emergent Scrooge designs when measured in a random stabilizer basis. These results indicate that Scrooge universality extends beyond conventional dynamical settings such as deep thermalization or Hilbert-space ergodicity.

Together, our results  provide a comprehensive and rigorous theoretical framework to analyze the emergence of maximally entropic, information-stingy randomness in quantum many-body systems, formulated in terms of Scrooge $k$-designs, and strongly sharpen and generalize previously known results in the literature. Beyond their conceptual significance, Scrooge designs generalize the notion of quantum randomness to physically constrained regimes.
Since Haar $k$-designs underlie applications such as randomized benchmarking and classical shadow tomography~\cite{dankert2009exact,choi2023preparing,huang2020predicting,mcginley2023shadow,tran2023measuring,mok2025optimal}, our results, by replacing Haar randomness with Scrooge designs, extend comparable performance guarantees to realistic quantum simulators operating at finite temperature or under symmetry and entanglement constraints.

This paper is organized as follows. In Sec.~\ref{sec:preliminaries}, we introduce the basic notions of quantum state ensembles, settings for their emergence in natural quantum many-body systems, and describe the key ensemble of interest --- the Scrooge ensemble --- studied in this work. In Sec.~\ref{sec:emergent_scrooge}, we present our main theorems, which identify general scenarios under which the Scrooge ensemble emerges. In Sec.~\ref{sec:physical_ingredients}, we elucidate the key physical ingredients underlying emergent Scrooge behavior \change{in projected ensembles}, and discuss the physical insights provided by our theoretical results. In Sec.~\ref{sec:numerics}, we report numerical investigations of emergent Scrooge behavior in a variety of quantum many-body systems. Finally, we conclude in Sec.~\ref{sec:outlook} and provide an outlook on future directions.

\section{Preliminaries of quantum state ensembles:  projected and temporal ensembles, Scrooge ensembles and Scrooge designs
}
\label{sec:preliminaries}
We begin with a general but brief introduction to ensembles of pure quantum states and explain how to characterize them  statistically and through their information-theoretic properties. We then review two physical settings where such  ensembles naturally appear. The first setting is deep thermalization, which concerns projected ensembles of local post-measurement states. The second setting is Hilbert-space ergodicity, which concerns temporal ensembles of global states generated by quantum chaotic dynamics.
We then elaborate upon the Scrooge ensemble, the  ensemble of interest in our paper. The informed reader may skip to our main results in  Sec.~\ref{sec:emergent_scrooge}.

\subsection{Quantum state ensembles and their information content}
 \label{sec:quantum_state_ensembles}

In this work, we are interested in  ensembles of pure quantum states
\begin{align}
\mathcal{E} = \{ p_i, |\psi_i\rangle \}_i.
\label{eqn:state_ensemble}
\end{align}
Here, $|\psi_i\rangle$ is a quantum state supported on a $D$-dimensional Hilbert space $\mathcal{H}$, and $p_i$ is that state's associated a priori probability. For simplicity we assume that $i$ is a discrete label; more generally, it can be a continuous label, in which case $p_i$  is replaced by a probability measure.
Such ensembles could arise in a multitude of physical contexts, for example in quantum communication wherein classical information is encoded in quantum messages, or in the quantum many-body dynamical phenomena of deep thermalization and Hilbert-space ergodicity (reviewed below).

For now, we keep our discussions purely formal  and explain common tools used to characterize quantum state ensembles. 
A standard diagnostic is to study the statistical moments of the distribution that an ensemble $\mathcal{E}$ encodes over the Hilbert space $\mathcal{H}$. To wit, for moment $k \in \mathbb{N}$, such information is encoded by the moment operator 
\begin{align}
\rho_{\mathcal{E}}^{(k)} & := \Eset{\psi \sim \mathcal{E}}(|\psi\rangle\langle \psi|)^{\otimes k} = \sum_{i} p_i \parens{\ket{\psi_i}\bra{\psi_i}}^{\otimes k} 
\label{eqn:k-moment}
\end{align}
supported on the symmetric subspace of $\mathcal{H}^{\otimes k}$ (note that $\rho_{\mathcal{E}}^{(k)}$ constitutes a valid density matrix on this space). We can compare the statistical similarity of one ensemble $\mathcal{E}$ with another ensemble $\mathcal{E}'$ at each moment $k$ via the trace distance
\begin{align}
\Delta^{(k)} := \frac{1}{2}  \norm{\rho_{\mathcal{E}}^{(k)} - \rho_{\mathcal{E'}}^{(k)}}_1
\end{align}
where $\| \cdot \|_p$ is the Schatten $p$-norm (thus $p=1$ above) and $\rho_{\mathcal{E'}}^{(k)}$ is the $k$th-moment operator of $\mathcal{E'}$, constructed analogously to Eq.~\eqref{eqn:k-moment}. 

A particularly useful and general lens with which one can understand the distribution that a quantum state ensemble describes in Hilbert space is through their information content, captured by the so-called accessible information $\mathcal{I}_\text{acc}(\mathcal{E})$ utilized in quantum information theory~\cite{preskill1998lecture,nielsen2011quantum}. Given an ensemble $\mathcal{E}$, the accessible information quantifies the maximum amount of classical information about the label $i$ that can be extracted by an observer through an optimal measurement on a single copy of a state drawn from the ensemble:
\begin{align}
\label{eqn:acc_info}
\mathcal{I}_\text{acc}(\mathcal{E}) := \sup_{M \in \text{POVM}}I(\mathcal{E}:M). 
\end{align}
Here $M$ is a positive operator-valued measure (POVM) and $I(\mathcal{E}:M)$ is the classical mutual information between the probability distribution $\{p_i\}$ and the distribution of measurement outcomes; see \SM{}~\ref{app:preliminaries} for details. 
 For an ensemble $\mathcal{E}$ with  average state  $\sigma = \sum_i p_i |\psi_i\rangle \langle \psi_i|  \equiv \rho_{\mathcal{E}}^{(1)}$, i.e., the density matrix, there are bounds on how small and large accessible information can be:
\begin{align}
Q(\sigma) \leq \mathcal{I}_\text{acc}(\mathcal{E}) \leq S(\sigma),
\end{align}
where $Q(\sigma) = -\sum_j [\lambda_j \ln \lambda_j \prod_{k\neq j} \lambda_k/(\lambda_k - \lambda_j)]$ is the so-called subentropy~\cite{josza1994lower} and $S(\sigma) = -\sum_j \lambda_j \log \lambda_j$ is the von Neumann entropy, where $\{\lambda_j\}_j$ are the eigenvalues of $\sigma$. The latter is famously known as the Holevo bound~\cite{holevo1973bounds} (here written for pure states). 

The accessible information $\mathcal{I}_\text{acc}(\mathcal{E})$ quantifies the ``ergodicity'' of the ensemble $\mathcal{E}$; that is, how uniformly the ensemble fills the Hilbert space, where a more ergodic ensemble has smaller accessible information and a less ergodic ensemble has larger accessible information. Indeed, the ensemble that maximizes the accessible information of a given density matrix $\sigma$ is an ensemble of mutually orthogonal states, which are perfectly distinguishable by the optimal measurement. This ensemble is clearly far from ergodic --- it is a discrete collection of well-separated states in Hilbert space.

The ensemble that minimizes the accessible information for a given $\sigma$ is less trivial and more interesting. 
It is the so-called Scrooge ensemble, denoted $\text{Scrooge}(\sigma)$, a continuous collection of overlapping quantum states so named because of its information-stingy property~\cite{josza1994lower, scrooge}. We give its precise definition in Sec.~\ref{sec:Scrooge}. Among all ensembles whose average state is constrained to be $\sigma$, Scrooge$(\sigma)$ is the most spread out over Hilbert space, maximizing the difficulty of acquiring information about signal state $|\psi_i\rangle$. Indeed, in the special case where $\sigma = I/D$ is the maximally mixed state, Scrooge$(\sigma)$ reduces to the uniformly distributed Haar ensemble such that every pure state in Hilbert space is equally likely.

\subsection{Projected ensembles and temporal ensembles}

Thus far, we have introduced quantum state ensembles as purely formal mathematical objects. Recently, two  physical settings where quantum state ensembles arise have received increasing attention in the study of quantum many-body systems. One is the projected ensemble composed of local post-measurement states, and the other is the temporal ensemble of global states generated by time-evolution. Both ensembles have been found to exhibit universal features in limiting cases, embodied by the physical phenomena of ``deep thermalization'' and ``Hilbert-space ergodicity,'' respectively. Here, we quickly review these topics. 

\subsubsection{Projected ensemble and deep thermalization}

Deep thermalization concerns the %
ensemble of states on a local region of a quantum many-body system conditioned upon measurements of the complementary region; this is called the {\it projected ensemble} and defined as follows~\cite{cotler2023emergent,choi2023preparing}. %

 Let $\ket{\Psi}_{AB}$ be a global ``generator'' $D$-dimensional quantum state supported on a bipartite Hilbert space $AB$, with subregions $A(B)$ of dimensions $D_A(D_B)$ respectively, so that $D = D_A D_B$.  
While not strictly necessary for this work, we will regard $\ket{\Psi}_{AB}$ as an $N$-qubit state with subsystems $A$ and $B$ comprising $N_A$ and $N_B$ qubits, respectively. In such a case, $D_A = 2^{N_A}$ and $D_B = 2^{N_B}$. Consider next a projective measurement performed on the ``bath'' subsystem $B$, described by a set of orthogonal rank-1 projectors $\{\Pi_z\}_{z=1}^{D_B}$ satisfying the normalization $\sum_z \Pi_z = I_B$. Measurement outcome $z$ occurs with probability $p_z = \braket{\Psi|(I_A \otimes \Pi_z)|\Psi}$, upon which the associated post-measurement state on the unmeasured subsystem $A$ is  $\ket{\psi_z}_A \propto (I_A \otimes \Pi_z)\ket{\Psi}$.
The set of such (normalized) post-measurement projected states, weighted by their probabilities, constitutes the projected ensemble %
\begin{equation}
    \mathcal{E}(\Psi) := \{p_z, \ket{\psi_z}_A\}_{z=1}^{D_B}\,
\end{equation}
(depicted in settings \textit{(2)} and \textit{(3)} of Fig.~\ref{fig:schematic}).
By capturing correlations with classical information extracted from the bath $B$, the projected ensemble provides a more refined description of the local subsystem $A$ than the reduced density matrix $\sigma_A$.

The term deep thermalization, a burgeoning research topic in recent years~\cite{cotler2023emergent,choi2023preparing,ippoliti2022solvable,ho2022exact,wilming2022high,claeys2022emergent,shrotriya2023nonlocality,liu2024deep,bhore2023deep,lucas2023generalized,mark2024maximum,chan2024projected,chang2025deep,varikuti2024unraveling,mok2025optimal,liu2025coherence,zhang2025holographic,yan2025characterizing,chakraborty2025fast,bejan2025matchgate,manna2025projected,vairogs2025extracting,loio2025quantum,goldstein2006distribution,goldstein2016universal,varikuti2025deep}, refers to the observation that, for generator states $|\Psi\rangle_{AB}$ arising from late-time dynamics of chaotic quantum many-body systems and in the thermodynamic limit $D_B \to \infty$ with $D_A$ fixed, the projected ensemble generically approaches Scrooge($\sigma_A$) (and related variants), which has been recognized to have maximal entropy in the sense of harboring minimal accessible information, given the reduced density matrix $\sigma_A$ \cite{mark2024maximum}. In particular, when there are no conservation laws constraining the dynamics such that we expect $\sigma_A$ to be maximally mixed, the corresponding Scrooge ensemble is the Haar ensemble. If the dynamics conserves energy and charge such that we expect quantum thermalization to drive $\sigma_A$ to a Gibbs state with a specified temperature and chemical potential~\cite{borgonovi2016quantum,mori2018thermalization,ueda2020quantum,gogolin2016equilibration}, the corresponding Scrooge ensemble is then the most stingy `unraveling' of the Gibbs state in terms of constituent pure states.

 \subsubsection{Temporal ensemble and Hilbert-space ergodicity}
 \label{sec:temporal_ensemble}

 Quantum state ensembles also arise naturally when we consider evolution of a global state for a random time. Suppose a $D$-dimensional quantum system evolves unitarily under some Hamiltonian $H(t)$, which in general could be time-dependent, and suppose the evolution time $t$ is sampled uniformly from the interval $[0,T]$. The resulting state ensemble 
 \begin{align}
 \label{eq:temporal_ensemble}
 \mathcal{E}_\text{Temporal}:=\{ dt/T, |\Psi_t\rangle\}_{t \in [0,T]}
 \end{align}
 (depicted in setting \textit{(1)} of Fig.~\ref{fig:schematic}) is called the \textit{temporal ensemble}~\cite{linden2009quantum,nakata2012phase,mark2024maximum,goldstein2010long}. 
 Here $t$ plays the role of the classical label $i$ in Eq.~\eqref{eqn:state_ensemble}. 
If energy is not conserved and $H(t)$ has generic time-dependence, one expects the temporal ensemble to be uniformly distributed for any choice of the initial state, because there is no preferred direction in Hilbert space, \change{allowing the system to explore the entire Hilbert space without constraints}. Indeed, the emergence of the Haar ensemble at late times has been demonstrated rigorously for various classes of time-dependent $H(t)$ \change{(such as aperiodic and quasiperiodic drives consisting of sequences of kicks generated randomly or even deterministically according to the Fibonacci and Thue--Morse words)}, a phenomenon termed \change{complete} Hilbert-space ergodicity~\cite{pilatowskycameo2024hilbert,pilatowskycameo2023complete,pilatowskycameo2025critically,shaw2025experimental,liu2025observation}.

When $H(t) = H$ is time-independent, energy conservation ensures that the populations $|\braket{E_j|\Psi_t}|^2$ in the energy basis $\{\ket{E_j}\}_j$ remain invariant in time, precluding the emergence of Haar in general. However, this does not impose constraints on the relative phases between energy eigenstates. At late observation times $T$, it is natural to expect the temporal ensemble to be well-described by the so-called random phase ensemble~\cite{nakata2012phase,mark2024maximum,mao2025random}

\begin{equation}\label{eq:randomphase}
    \mathcal{E}_{\text{Random Phase}} := \bparens{\frac{d^D\varphi}{(2\pi)^D},\sum_{j=1}^{D} |\braket{E_j|\Psi_0}|e^{i\varphi_j} \ket{E_j}},
\end{equation}
where $|\Psi_0\rangle = \sum_{j=1}^D \langle E_j|\Psi_0\rangle |E_j\rangle$ is the initial state and $\varphi = (\varphi_1, \varphi_2, \cdots, \varphi_D)$ are angles uniformly distributed on the $D$-dimensional torus $[0,2\pi)^D$. 
Indeed, the density matrix associated with the temporal ensemble is the diagonal ensemble~\cite{deutsch1991quantum,srednicki1994chaos}
\begin{equation}
    \sigma_\text{diag} = \sum_{j=1}^{D} |\braket{E_j|\Psi_0}|^2 \ket{E_j}\bra{E_j},
\label{eq:diagonal_ensemble}
\end{equation}
a well-known equilibrium state, defined by the initial state dephased in the energy basis~\cite{popescu2006entanglement}. %
The random phase ensemble may be thought of as the most `ergodic' ensemble subject to the constraint of conservation of populations on energy eigenstates (in the sense that the relative phases $e^{i \varphi_j}$, which are the remaining degrees of freedom, are distributed with no preferred location on the torus), and thus its emergence can be understood as Hilbert-space ergodicity in the case of energy conservation. 

Technically, Mark et al.~proved that the $k$th moments of the random phase ensemble describe the limiting form of the temporal ensemble if the energy spectrum of $H$ satisfies the assumption of a so-called ``$k$-th no-resonance condition''~\cite{mark2024maximum}: any two subsets of energy levels $\{E_1,\ldots,E_k\}$ and $\{E_1^\prime,\ldots,E_k^\prime\}$ satisfy $\sum_{i=1}^k E_i = \sum_{i=1}^k E_i^\prime$ if and only if the subsets are equivalent, up to a reordering of energy levels. 
This assumption describes the scenario where energy levels are generic, 
a physical property which is 
intuitively expected to hold true for quantum chaotic systems, whose energy spectra are known to be well governed by random matrix theory~\cite{alessio2016quantum}. In particular, for $k = 1$, this condition is equivalent to \change{having no degeneracies in the energy spectrum, a mild condition that holds generically for both interacting integrable and quantum chaotic systems alike}.

\subsection{Scrooge ensemble and Scrooge designs}
\label{sec:Scrooge}
We have seen how quantum state ensembles arise in many physical settings, and in certain appropriate limits (large system sizes or late times), appear to take various simple universal limiting forms, like the Haar ensemble, the Scrooge ensemble, and the random phase ensemble.
As mentioned in the introduction, our key interest in this paper will be firming up general but precise conditions when the information-stingy Scrooge ensemble can be {\it rigorously} proven to appear, including but going even beyond the setting of deep thermalization, and in the context of Scrooge designs.
To that end, we now formally define the Scrooge ensemble and Scrooge designs.

\begin{definition}[Scrooge ensemble~\cite{josza1994lower}]
\label{defn:scrooge}
The Scrooge ensemble with density matrix $\sigma$ supported on a $D$-dimensional Hilbert space, denoted Scrooge($\sigma$), is the unique ensemble $\mathcal{E}$ of pure states satisfying
\begin{equation}
\begin{aligned}
    \mathcal{E} = \argmin_{\mathcal{E}^\prime}& \quad \mathcal{I}_{\text{acc}}(\mathcal{E}^\prime), \\
   \text{subject to} \quad& \Eset{\psi \sim \mathcal{E}^\prime} \ket{\psi}\bra{\psi}= \sigma.
\end{aligned}
\end{equation}
\end{definition}
\noindent An explicit construction of  
Scrooge($\sigma$) is given by  
\begin{equation}
    \text{Scrooge}(\sigma) = \bparens{D\braket{\phi|\sigma|\phi} \text{d}\phi, \frac{\sqrt{\sigma}\ket{\phi}}{\norm{\sqrt{\sigma}\ket{\phi}}}},
\end{equation}
with $\text{d}\phi$ the Haar measure on the unit sphere in $\mathbb{C}^{D}$~\cite{josza1994lower}. From the expression, one sees that the Scrooge ensemble can be thought of as a deformed version of the Haar ensemble, to account for the constraint that its mean is $\sigma$. 

We will also be interested in ensembles that {\it approximate} the Scrooge ensemble. 
Analogously to state $k$-designs routinely used in quantum information theory~\cite{dankert2009exact,gross2007evenly},
which provide low-complexity approximations that capture the statistical properties of the Haar ensemble up to the $k$th moment, we likewise introduce the notion of $k$-designs for the Scrooge ensemble: 

\begin{definition}[Approximate Scrooge $k$-designs]\label{defn:scrooge_designs} Let $\mathcal{E}$ be an ensemble of pure states, with $k$th moment
\begin{equation}
    \rho_{\mathcal{E}}^{(k)} = \Eset{\psi \sim \mathcal{E}} \parens{\ket{\psi}\bra{\psi}}^{\otimes k}.
\end{equation} 
We say $\mathcal{E}$ is a  Scrooge$(\sigma)$ $k$-design with additive error $\epsilon$ if
\begin{equation}
    \frac{1}{2} \norm{\rho_{\mathcal{E}}^{(k)} - \rho_{\text{Scrooge}}^{(k)}(\sigma)}_1 \leq \epsilon\,,
\end{equation}
and a Scrooge$(\sigma)$ $k$-design with relative error $\epsilon$ if\footnote{The relative error is a strictly stronger notion of approximation: A relative error $\epsilon$ implies an additive error $\epsilon$, but the converse is not true.} 
\begin{equation}
(1-\epsilon)\rho_{\text{Scrooge}}^{(k)}(\sigma) \preceq \rho_{\mathcal{E}}^{(k)} \preceq (1+\epsilon)\rho_{\text{Scrooge}}^{(k)}(\sigma)
\end{equation}
\change{in Loewner order, where $A \preceq B$ if $B - A$ is positive semidefinite.} The $k$th moment of Scrooge$(\sigma)$ is given by
\begin{equation}
    \rho_{\text{Scrooge}}^{(k)}(\sigma) = D \Eset{\phi \sim \text{Haar}(D)}\sparens{\frac{\parens{\sqrt{\sigma}\ket{\phi}\bra{\phi}\sqrt{\sigma}}^{\otimes k}}{\braket{\phi|\sigma|\phi}^{k-1}}}.
\label{eq:scrooge_kth_mom_exact}
\end{equation}
\end{definition}

We note that evaluating the $k$th moment $\rho_{\text{Scrooge}}^{(k)}(\sigma)$ in Eq.~\eqref{eq:scrooge_kth_mom_exact} involves integrating a rational function over the Haar ensemble, and is thus not amenable to simple closed-form expressions using standard Weingarten calculus~\cite{collins2017weingarten,collins2022weingarten,kostenberger2021weingarten}. While Mark et al.~derived an explicit expression for it~\cite{mark2024maximum}, its form is complicated and unwieldy. They therefore also considered the closely related {\it unnormalized} Scrooge ensemble $\tilde{\mathcal{E}}$ composed of unnormalized states $\tilde{\mathcal{E}} = \{\sqrt{D\sigma} \ket{\phi}\}$, where $\ket{\phi} \sim \text{Haar}(D)$, with $\sigma$ the $D$-dimensional density matrix defining the Scrooge ensemble in question. For this ensemble, the $k$th moment takes a simple and analytically tractable form: $\tilde{\rho}^{(k)}_{\text{Scrooge}}(\sigma) = (D\sigma)^{\otimes k} \rho_{\text{Haar}}^{(k)}$, to which many of their results pertain. 
In our work, the unnormalized Scrooge also plays a useful albeit intermediary role, as captured in the following technical lemma 
we derive and whose proof we present in Appendix~\ref{app:scrooge_approximation}:
\begin{lemma}[Scrooge approximation] Consider the ensemble of unnormalized states $\tilde{\mathcal{E}} = \{\sqrt{D\sigma} \ket{\phi}\}$, where $\ket{\phi} \sim \text{Haar}(D)$, and $\sigma$ is an arbitrary density matrix with dimension $D$. 
For $k^2 \norm{\sigma}_2\ll 1$, $\tilde{\mathcal{E}}$ forms a Scrooge$(\sigma)$ $k$-design with additive error $\bigO{k \norm{\sigma}_2}$ and relative error $\bigO{4^k k \norm{\sigma}_2}$. 
\label{lemma:scrooge_approx}
\end{lemma}
Lemma~\ref{lemma:scrooge_approx} is the technical backbone of this paper: 
the $k$th moment of the {\it normalized} Scrooge ensemble (the actual object of interest) can be approximated by the simpler expression up to an error controlled by the purity of $\sigma$ (see Ref.~\cite{mcginley2025scrooge} for an analogous relative-error bound, controlled by $\norm{\sigma}_\infty$ instead of $\norm{\sigma}_2$).\footnote{In quantum many-body systems, $\norm{\sigma}_2 = \sqrt{\Tr\parens{\sigma^2}}$ is often exponentially small in the system size (or, equivalently, inverse power in $D$). Therefore, we expect the low-purity condition to be valid up to exponentially high moments $k$.} This result generalizes that of Ref.~\cite{mark2024maximum} for $k = 2$, and 
will be key in allowing us to make rigorous statements on the limiting form of the temporal and projected ensembles by analyzing their unnormalized counterparts.

\section{Emergent Scrooge designs in quantum many-body systems}
\label{sec:emergent_scrooge}

We are now ready to state our main results, which provide sufficient conditions for the rigorous emergence of Scrooge designs in quantum many-body systems.

\subsection{Scrooge designs emerge naturally in dynamics}

Our first result pertains to the temporal ensemble for energy-conserving Hamiltonian dynamics, which as explained in Sec.~\ref{sec:temporal_ensemble}, generically matches the random phase ensemble~\eqref{eq:randomphase} in the limit of long observation time. 
Ref.~\cite{mark2024maximum} had derived that the $k$th moment of the random phase ensemble (which equals the temporal ensemble under the assumption of $k$th no-resonance condition on the spectrum) is well-approximated by that of the unnormalized Scrooge ensemble, in the limit of low purity. Our Lemma~\ref{lemma:scrooge_approx} allows us to go further, relating it directly to the $k$th moment of the normalized Scrooge ensemble, resulting in Theorem~\ref{thm:global_scrooge}.

\begin{theorem}[Scrooge designs from late-time chaotic dynamics] The ensemble of states sampled at late times in quantum chaotic dynamics, generated by a Hamiltonian obeying the $k$th no-resonance condition,
forms an approximate Scrooge$(\sigma_{\text{diag}})$ $k$-design for $k^2 \norm{\sigma_{\text{diag}}}_2\ll 1$, with an additive error     \begin{equation}
    \epsilon = \change{\bigO{k \norm{\sigma_{\text{diag}}}_2}}.
    \label{eqn:global_scrooge_error}
\end{equation}
 
\label{thm:global_scrooge}
\end{theorem}
The proof is presented in \SM{}~\ref{app:global_scrooge}. \change{Our proof works by showing that the $k$th moment of the random phase ensemble is close to that of the unnormalized ensemble $\{\sqrt{D \sigma_{\text{diag}}} \ket{\phi}\}$ where $\ket{\phi} \sim \text{Haar}(D)$, with a trace distance of $\bigO{k^2 \norm{\sigma_\text{diag}}_2^2}$.} Lemma~\ref{lemma:scrooge_approx} then guarantees that the resulting ensemble forms a Scrooge$(\sigma_{\text{diag}})$ $k$-design, incurring an extra additive error of $\bigO{k \norm{\sigma_{\text{diag}}}_2}$. 

Theorem~\ref{thm:global_scrooge} offers substantial improvements compared to those arising from results of Ref.~\cite{mark2024maximum}. 
Extracting an error bound from the expressions derived in Ref.~\cite{mark2024maximum} and invoking our Lemma~\ref{lemma:scrooge_approx} yields the much weaker $\epsilon = \bigO{k! \exp(k) \norm{\sigma_{\text{diag}}}_2^2 + k \norm{\sigma_{\text{diag}}}_2}$ \change{which contains an additional term that grows very quickly with $k$}, though this was not explicitly stated there.\footnote{For brevity, we suppress constants in the definition of $\exp(k)$, writing $\exp(k) = e^{c_1 k^{c_2}}$ for some constants $c_1, c_2 > 0$.} \change{Note that $\norm{\sigma_{\text{diag}}}_2$ is often exponentially small in $N$, for typical initial product states that are supported on exponentially many energy eigenstates, which holds even for integrable many-body Hamiltonians, thereby satisfying the low-purity condition. On the other hand, the $k$th no-resonance condition is generically expected for interacting many-body Hamiltonians, including both integrable and chaotic systems~\cite{mark2024maximum}. Integrable systems that are effectively non-interacting such as the 1D transverse-field Ising model may violate the no-resonance conditions due to algebraic relations between the energy levels.} \change{While beyond the scope of this paper, it would be interesting to study the robustness of Theorem~\ref{thm:global_scrooge} when the no-resonance conditions are weakly violated~\cite{linden2009quantum,riddell2024no-resonance}.}

Conceptually, Theorem~\ref{thm:global_scrooge} rigorizes the expectation that Scrooge-like behavior {\it naturally} appears already in late-time dynamics; in contrast to the case of deep thermalization, no measurements are needed. \change{Here, ``late-time" formally refers to the long-time limit $T \to \infty$. In Ref.~\cite{mcginley2025scrooge}, it was argued that the temporal ensemble for a fixed time-independent Hamiltonian typically requires $T \gtrsim \exp(k N)$ to form an approximate Scrooge $k$-design. This can be physically thought of as the timescale required to resolve exponentially small energy gaps.} Hence, not only is the late-time temporal ensemble maximally ergodic in the sense that the relative phases of energy eigenstates are uniformly distributed, in addition (for the case of low purity), its accessible information is close to minimal. This means that similar maximum entropy principles apply to both deep thermalization and Hilbert-space ergodicity, conceptually unifying these two phenomena.

\subsection{Scrooge designs emerge in the projected ensemble of a sufficiently complex initial state}

Our next theorem returns to the projected ensemble. Motivated by the insight from the previous result, Theorem~\ref{thm:global_scrooge}, --- that {\it global} Scrooge designs emerge in late-time chaotic dynamics of a quantum many-body system --- we inquire if a typical global state drawn at late times, and more generally, from a global Scrooge ensemble, can itself produce Scrooge behavior {\it locally} in the setting of deep thermalization. 
As mentioned in the introduction, Goldstein et al.~\cite{goldstein2016universal} had considered the setting where the global generator state is drawn from an {\it exact} Scrooge ensemble; here, we generalize this to the case of Scrooge designs.

\begin{theorem}[Emergent Scrooge $k$-design from Scrooge $2k$-design generator]\label{thm:2kgenerator}
Fix $k \in \mathbb{N}$. Let $\ket{\Psi}_{AB}$ be sampled from a Scrooge$(\sigma)$ $2k$-design, with relative error $\epsilon$. Denote the reduced density matrices of $\sigma$ on $A$ and $B$ by $\sigma_A$ and $\sigma_B$, respectively. Let $\mathcal{E}(\Psi)$ be the projected ensemble generated from $\ket{\Psi}_{AB}$ by measuring $B$ in an arbitrary orthonormal basis $\{\ket{z}\}_{z=1}^{D_B}$ and denote the normalized density matrix $\hat{\sigma}_{A|z} = (I_A \otimes \bra{z})\sigma(I_A \otimes \ket{z})/\braket{z|\sigma_B|z}$. Then for $\norm{\hat{\sigma}_{A|z}}_2 \ll 1$ and $1 \ll D_A \leq D_B$,
\begin{equation}
\label{eq:avg_td_2kgenerator_scrooge}
\begin{aligned}
    &\E_{\Psi}\norm{\rho_\mathcal{E}^{(k)}(\Psi) - \sum_{z=1}^{D_B}\braket{z|\sigma_B|z}\rho_{\text{Scrooge}}^{(k)}(\hat{\sigma}_{A|z})}_1 \\
    &\leq \bigO{\sqrt{{D_A}^k \parens{\epsilon + \norm{\sigma}_2 +\change{D_A \norm{\sigma}_2^2}}}}.
\end{aligned}
\end{equation}
\end{theorem}
This result is proven in \SM{}~\ref{app:projens_2kgenerator}, with Lemma~\ref{lemma:scrooge_approx} playing a key role in the proof. Theorem~\ref{thm:2kgenerator} implies that, with high probability, the projected ensemble for $\epsilon + \norm{\sigma}_2 \ll {D_A}^{-k}$ approximates a mixture of Scrooge$(\hat{\sigma}_{A|z})$, i.e., a so-called ``generalized Scrooge ensemble''~\cite{mark2024maximum}. Note that throughout this manuscript, we will suppress error terms that vanish in the asymptotic limit $D_A,D_B\to \infty$ for simplicity. Intuitively, Theorem~\ref{thm:2kgenerator} says that we can interpret each projected state $\ket{\psi_z}$ as a realization of Scrooge$(\hat{\sigma}_{A|z})$, weighted by the average measurement probability $\braket{z|\sigma_B|z}$. This is nontrivial since each projected state $\parens{\ket{\psi_z}\bra{\psi_z}}^{\otimes k}$ is, in general, far from $\rho_{\text{Scrooge}}^{(k)}(\hat{\sigma}_{A|z})$.

While Theorem~\ref{thm:2kgenerator} states that the projected ensemble is formally a mixture of exponentially many (in $N_B$) distinct Scrooge ensembles, in certain physical settings it is possible to obtain a simpler, `coarse-grained' mixture with a much smaller number of distinct Scrooge ensembles. For instance, as argued for in  Ref.~\cite{chang2025deep}, if the generator state $\ket{\Psi}$ is obtained under $\text{U}(1)$-symmetric scrambling dynamics, then the $2^{N_B}$ possible measurement outcomes can be grouped into at most $N_B + 1$ equivalence classes. Consequently, the projected ensemble is a mixture of only a polynomial number of distinct Scrooge ensembles. 
In the extreme case where the measurement basis is `information non-revealing', that is, measurements do not reveal information about the local charge~\cite{mark2024maximum, chang2025deep} so that $\hat{\sigma}_{A|z} \approx \sigma_A$, the different Scrooge$(\hat{\sigma}_{A|z})$ ensembles collapse to a single ensemble Scrooge$({\sigma}_A)$. This yields $\rho_{\mathcal{E}}^{(k)}(\Psi) \approx \rho_{\text{Scrooge}}^{(k)}(\sigma_A)$, i.e., a vanilla single Scrooge ensemble, where $\sigma_A$ is the reduced density matrix of $\ket{\Psi}$ on $A$. Understanding if a similar coarse-grained description of the generalized Scrooge ensemble exists for generator states obtained under different conservation laws, like energy, would be an interesting future direction.

In the following, we discuss various examples applying Theorem~\ref{thm:2kgenerator}.

\subsubsection{Example: Infinite-temperature states}

A special case of Theorem~\ref{thm:2kgenerator} is when the density matrix in the definition of the Scrooge generator is maximally mixed, $\sigma = I/D$. We then immediately have that local {\it Haar} $k$-designs emerge from a typical global {\it Haar} $2k$-design generator state in the context of deep thermalization.\footnote{This is consistent with the result of a  recent work by Ghosh et al.~\cite{ghosh2025design}, though there they prove this for the weaker additive error and here we prove this for the stronger relative error.}.
In fact, for this special case, we are able to improve upon the error bound stated in Theorem~\ref{thm:2kgenerator}, which we present as the following corollary. 
\begin{corollary}[Emergent $k$-design from $2k$-design generator]\label{cor:2ktokdesign}
Let $\ket{\Psi}_{AB}$ be sampled from a $2k$-design, with relative error $\epsilon$. $\mathcal{E}(\Psi)$ is the projected ensemble generated from $\ket{\Psi}_{AB}$ by measuring $B$ in an arbitrary orthonormal basis. For $k^2 \ll D_A$ and $D_A, D_B \gg 1$,
\begin{equation}
    \E_{\Psi}\norm{\rho_\mathcal{E}^{(k)}(\Psi) - \rho_{\text{Haar},A}^{(k)}}_1 \leq \sqrt{\frac{{D_A}^k}{k!}\parens{\frac{1}{D_B} + \bigO{\epsilon}}}.
\label{eq:avg_td_2kgenerator_haar}
\end{equation}
\end{corollary}
The proof is given in \SM{}~\ref{app:projens_2kgenerator}.
The main improvement here in Eq.~\eqref{eq:avg_td_2kgenerator_haar} compared to Eq.~\eqref{eq:avg_td_2kgenerator_scrooge} is that the latter is valid only for any fixed $k$ independent of system size, while $k$ can grow up to exponentially large in $N_A$ in the former.

Corollary~\ref{cor:2ktokdesign} is moreover a significant improvement over existing results in the literature~\cite{cotler2023emergent}, which require $\ket{\Psi}_{AB}$ to be drawn from a $k^\prime$-design with $k^\prime = \bigO{k^2 \log D_A}$ in order to guarantee that $\mathcal{E}(\Psi)$ forms an approximate $k$-design on $A$. In contrast, we only require $k^\prime = 2k$ independent of system size. 
For the simplest non-trivial case of $k=2$, Corollary~\ref{cor:2ktokdesign} states that drawing $\ket{\Psi}_{AB}$ from \textit{any} $4$-design (with sufficiently small $\epsilon$) suffices for $\mathcal{E}(\Psi)$ to form an approximate $2$-design. This is, in fact, optimal: multi-qubit stabilizer states form exact $3$-designs with $\epsilon = 0$~\cite{webb2015clifford,PhysRevA.96.062336, zhu2016clifford}, but the projected ensemble generated by a stabilizer state forms only a $1$-design, which is far from a $2$-design~\cite{vairogs2025extracting}.\footnote{This can be derived by a counting argument: Measuring subsystem $B$ of a stabilizer state $\ket{\Psi}_{AB}$ in a stabilizer basis will only generate at most $2^\ell$ distinct stabilizer states on $A$, where $\ell \leq N_A$ is the number of ebits shared between $A$ and $B$. This is insufficient to form approximate $2$-designs, which require $\sim 4^{N_A}$ states.}
This observation highlights the essential role of nonstabilizerness (also known as magic) in deep thermalization, complementing existing works~\cite{vairogs2025extracting,loio2025quantum}, a point we discuss further in Sec.~\ref{sec:physical_ingredients} and explore numerically in Sec.~\ref{sec:numerics}.

The $2k$-design with relative error $\epsilon$ used in Corollary~\ref{cor:2ktokdesign} can be prepared by applying a global unitary drawn from a unitary $2k$-design (with relative error $\epsilon$) on a fixed reference state $\ket{0}_{AB}$. Such unitaries can be constructed efficiently, for example, using local random circuits with depth $\bigOtilde{k \log(N/\epsilon)}$~\cite{schuster2025extremely} or doped Clifford circuits with $\bigOtilde{N k + \log(1/\epsilon)}$ non-Clifford gates~\cite{haferkamp2022efficient,leone2025non}. We also remark that while we use relative error in the theorem for convenience, the emergence of state designs, as implied by Corollary~\ref{cor:2ktokdesign}, also holds true if $\ket{\Psi}$ is sampled from a $2k$-design with additive error $\epsilon \ll k!/(D_A D_B^2)^k$.\footnote{Note that in order to guarantee that the projected ensemble approximates a $2$-design, the additive error must satisfy $\epsilon \ll 1/(D_A D_B)$, since $N$-qubit stabilizer states form approximate 4-designs with additive error $\epsilon = \Theta(2^{-N})$~\cite{damanik2018optimality,bittel2025complete}, yet their projected ensembles only form a 1-design.} 

\subsubsection{
(Almost) example: Quantum chaotic Hamiltonian dynamics}

Returning to the motivation of Theorem \ref{thm:2kgenerator}, we may ask whether a typical late-time state evolved under a chaotic global quantum many-body Hamiltonian is guaranteed to generate a local projected ensemble that is approximately Scrooge. If so, this conclusion would provide a satisfying and broadly applicable explanation for deep thermalization. It is tempting to declare that the desired result follows immediately from our theorems; Theorem \ref{thm:global_scrooge} supports the claim that chaotic time evolution produces a global state that is sufficiently scrambled, and then the Scrooge-like behavior of the projected ensemble seems to follow from Theorem \ref{thm:2kgenerator}. Unfortunately, though, this reasoning is flawed for a subtle technical reason: Theorem \ref{thm:2kgenerator} requires a generator state drawn from an approximate Scrooge design with small {\it relative} error,  while Theorem \ref{thm:global_scrooge} only guarantees the weaker {\it additive} error. 

Fortunately, the statement ``dynamically generated global states yield local Scrooge behavior'' can be justified nevertheless {\it without} invoking Theorem \ref{thm:2kgenerator}:

\begin{proposition}[Emergent Scrooge $k$-design from late-time chaotic dynamics, informal]\label{prop:scrooge_from_temporal}
Let $\ket{\Psi_t} = e^{-iHt}\ket{\Psi_0}$ be drawn from the temporal ensemble~\eqref{eq:temporal_ensemble} defined for long times $T$, and suppose $H$ is an arbitrary Hamiltonian satisfying the $k$th no-resonance condition. In the low-purity limit, the projected ensemble $\mathcal{E}(\Psi)$ forms a generalized Scrooge ensemble with high probability.
\end{proposition}
\noindent The formal statement and proof are provided in Appendix~\ref{app:projens_2kgenerator}. While Proposition~\ref{prop:scrooge_from_temporal} does not formally utilize Theorem \ref{thm:2kgenerator}, the former's derivation is inspired by technicalities in the proof of the latter.

We note that Ref.~\cite{mark2024maximum} had already anticipated this result:  they proved that the \textit{unnormalized} projected ensemble (i.e., the collection of local post-measurements states before normalizing) is close to the \textit{unnormalized} Scrooge ensemble $\tilde{\mathcal{E}}$, as defined in Lemma \ref{lemma:scrooge_approx}. Using the analytical tools developed in this work, in particular Lemma~\ref{lemma:scrooge_approx}, our result Proposition~\ref{prop:scrooge_from_temporal} closes this conceptual gap and proves that Scrooge behavior emerges also in the \textit{normalized} projected ensemble, the physical ensemble of interest in deep thermalization.

\subsubsection{Example: Canonical thermal pure quantum (cTPQ) ensemble}
An example of a set of quantum states  which {\it does} form a relative-error Scrooge design, thereby satisfying the conditions of Theorem~\ref{thm:2kgenerator}, is the so-called ``canonical thermal pure quantum'' (cTPQ) ensemble, which was introduced by the quantum thermalization and thermodynamics communities~\cite{sugiura2013canonical,nakagawa2018universality} and is closely related to the notion of canonical typicality~\cite{popescu2006entanglement,goldstein2006canonical,reimann2007typicality}. 
Given a time-independent Hamiltonian $H$ and inverse temperature $\beta$, it is defined as
\begin{equation}
    \mathcal{E}_{\text{cTPQ}} = \bparens{\frac{1}{\mathcal{N}}\sum_{j=1}^D \xi_j e^{-\beta H/2} \ket{j}}_\xi.
\end{equation}
Above, $\xi = (\xi_1,\cdots,\xi_D)$ where   $\xi_j$ are taken to be independent %
zero-mean complex Gaussian variables\footnote{If $H$ has time-reversal symmetry, $\xi_j$ is a real-valued Gaussian variable.}, $\{\ket{j}\}$ is an arbitrary orthonormal basis, and $\mathcal{N}$ is the normalization factor.

cTPQ states constitute pure state approximations to thermal Gibbs states (at inverse temperature $\beta$), in the sense that expectation values of local observables are reproduced up to  fluctuations that vanish exponentially with system size~\cite{sugiura2013canonical,nakagawa2018universality},  since their associated density matrix is, by construction, exponentially close to the thermal state $\sigma_\beta \propto \exp(-\beta H)$~\cite{sugiura2013canonical}.
In fact, such states have been demonstrated numerically to also reproduce certain physical properties of scrambled many-body states at finite effective temperatures, such as entanglement entropy, thermodynamic quantities, and phase diagrams of lattice gauge theories~\cite{sugiura2013canonical,nakagawa2018universality,davoudi2023towards}.

In Appendix~\ref{app:global_scrooge}, we show that $\mathcal{E}_{\text{cTPQ}}$ forms a Scrooge$(\sigma_\beta)$ $k$-design with relative error $\bigO{4^k k \norm{\sigma_\beta}_2}$. 
Theorem~\ref{thm:2kgenerator}  therefore immediately implies that cTPQ states deeply thermalize, producing local Scrooge behavior with $\sigma_A = \Tr_B \sigma_\beta$. Since energy eigenstates of quantum many-body chaotic systems are expected to satisfy the eigenstate thermalization hypothesis, and have also been numerically shown to deep thermalize locally to Scrooge~\cite{cotler2023emergent,mark2024maximum}, cTPQ states offer a natural framework for modeling such behavior, inspired by random matrix theory and quantum typicality.

\subsection{Scrooge designs emerge in the projected ensemble for sufficiently complex measurement basis}
Thus far, we have analyzed the cases where the generator state $\ket{\Psi}_{AB}$ is assumed to be drawn from a statistical ensemble, and the measurement basis is fixed but arbitrary. A complementary scenario occurs when we allow $\ket{\Psi}_{AB}$ to be fixed but arbitrary, and the projected ensemble is generated by applying a scrambling unitary $U_B$ to $B$ and then measuring in a fixed basis. 
This is also a setting considered by Goldstein et al.~\cite{goldstein2016universal} wherein they considered $U_B$ to be drawn from the {\it exact} Haar distribution, but here we relax this stringent assumption and take the scrambling unitary to be drawn from an approximate Haar $2k$-design, which can be implemented efficiently with local random circuits~\cite{schuster2025extremely}.
We prove that, in this scenario, the projected ensemble is universally Scrooge. %

\begin{theorem}[Emergent Scrooge $k$-design from $2k$-design measurement basis]\label{thm:ScroogeByMeasBasis}
Let $\ket{\Psi}_{AB} = (I_A \otimes U_B)\ket{\Psi_0}_{AB}$, where $\ket{\Psi_0}_{AB}$ is an arbitrary bipartite state, and suppose the scrambling unitary $U_B$ is sampled from a unitary $2k$-design with relative error $\epsilon$. Denote the reduced density matrix of $\ket{\Psi}$ on $A$ by $\sigma_A$, with effective dimension $D_{A,\text{eff}} = \parens{\norm{\sigma_A}_2/\norm{\sigma_A}_4}^4$. $\mathcal{E}(\Psi)$ is the projected ensemble generated by $\ket{\Psi}_{AB}$ by measuring $B$ in an arbitrary orthonormal basis $\{\ket{z}\}_{z=1}^{D_B}$. Then, for $k^4 \ll {D_{A,\text{eff}}}$, and $1 \ll D_A \leq D_B$,
\begin{equation}
\begin{aligned}
    &\E_{U_B} \norm{\rho^{(k)}_\mathcal{E}(\Psi) - \rho_{\text{Scrooge}}^{(k)}(\sigma_A)}_1 \\ 
    \leq \, & \sqrt{\parens{D_A \norm{\sigma_A}_2^2}^k \bigO{\change{\frac{k^2}{D_{A,\text{eff}}}}} + \change{\frac{{D_A}^k}{k!}\bigO{\frac{k^{2k+2}}{D_B} + \epsilon}}}.
\end{aligned}
\end{equation}
\end{theorem}
This result is proven in \SM{}~\ref{app:projens_2kmeasbasis}, with Lemma~\ref{lemma:scrooge_approx} a key step in the proof. %
Theorem \ref{thm:ScroogeByMeasBasis} refines the expectation that Scrooge behavior emerges locally when the measurement basis on the complementary region is sufficiently entangling and complex, modeling the information scrambling arising in natural chaotic many-body systems. The condition $k^4 \ll D_{A,\text{eff}}$ required in Theorem~\ref{thm:ScroogeByMeasBasis} is typically satisfied for complex many-body states, where the effective dimension of $\sigma_A$ (roughly, how many eigenvalues contribute significantly), is approximately $D_A = 2^{N_A}$ for Haar random states when $D_A \ll D_B$ (a more physical example is discussed below). In such cases, Theorem~\ref{thm:ScroogeByMeasBasis} implies that the projected ensemble is, with high probability, close to a Scrooge$(\sigma_A)$ $k$-design for a sufficiently large $D_B$ and a sufficiently small $\epsilon$. This result generalizes the theorem of Wilming and Roth~\cite{wilming2022high}, who showed that the projected ensemble forms a Haar $k$-design when $\sigma_A$ is close to maximally mixed, and the unitary $U_B$ is drawn exactly from the Haar measure.

\subsubsection{Example: Many-body states at thermal equilibrium}
As a physical example, suppose the global state $\ket{\Psi}$ obeys quantum thermalization, i.e., the local reduced density matrix $\sigma_A$ is well approximated by the local  Gibbs state   $\exp(-\beta H_A)/\text{Tr}\sparens{\exp(-\beta H_A)}$ on region $A$. Here $H_A$ is the restriction of the global Hamiltonian $H$ to subsystem $A$ and $\beta$ is the inverse temperature. The local Gibbs state is often a good approximation for the marginal of the global Gibbs state $\text{Tr}_B(e^{-\beta H})/\text{Tr}(e^{-\beta H})$, when $N_A \gg 1$, and naturally arises from quench dynamics, or from energy eigenstates of non-integrable Hamiltonians (which are believed to obey the eigenstate thermalization hypothesis)~\cite{dymarsky2018subsystem,garrison2018does}. 

Now, we can always express $H_A$ as a linear combination of Pauli operators%
\begin{equation}
\label{eq:ham_A_pauliexpansion}
    H_A = \sum_{m=1}^{M} c_m P_m,
\end{equation}
 where $c_m$ are real coefficients, and $P_m$ are Pauli operators on $A$.
 Assuming that the bulk spectral density of $H_A$ is Gaussian distributed (which holds approximately for a wide range of local Hamiltonians~\cite{keating2015spectra,hartmann2005spectral}), 
 we obtain (see \SM{}~\ref{app:projens_2kmeasbasis})
\begin{equation}
    \norm{\sigma_A}_p \approx \frac{1}{{2}^{N_A(1 - 1/p)}} \exp\parens{\frac{p-1}{2} \beta^2 \Delta^2},  
\end{equation}
where $\Delta^2 = \sum_m c_m^2 \sim M$ is the variance of the Gaussian distribution. Thus, the condition $k^4 \ll D_{A,\text{eff}}$ needed in Theorem~\ref{thm:ScroogeByMeasBasis} becomes $k^4 \ll {2}^{N_A} \exp(-4\beta^2 \Delta^2)$. For Hamiltonians with geometrically local interactions, we expect $M \propto N_A$. This implies the emergence of Scrooge designs for $\beta$ smaller than some constant $\beta_c \propto \sqrt{N_A/M}$; in other words, if the temperature is sufficiently high. Crucially, this does not require $\sigma_A$ to be close to infinite temperature, as in previous works~\cite{wilming2022high}, thereby generalizing existing results to the finite-temperature scenario. 

Theorem~\ref{thm:ScroogeByMeasBasis} also provides a protocol to sample from thermal Scrooge designs, which are potentially useful for learning properties of thermal states~\cite{coopmans2023predicting}. In this protocol, one prepares a purification of the thermal state $\sigma_A \propto \exp(-\beta H_A)$. %
This can be done via quantum Gibbs sampling algorithms~\cite{chen2023quantum,chen2025efficient}, imaginary time evolution~\cite{motta2020determining,mcardle2019variational}, or other constructions of the thermofield double state~\cite{maldacena2018eternal,cottrell2019how}. Then, additional ancilla qubits are introduced. The reference system and the ancilla make up subsystem $B$, which is then scrambled with a unitary drawn from a $2k$-design. Finally, each projective measurement on $B$ results in a sample of a thermal Scrooge $k$-design on $A$.

\section{Physical ingredients behind emergent Scrooge behavior \change{in projected ensembles}}
\label{sec:physical_ingredients}

\begin{figure}
    \centering
    \includegraphics[width=\linewidth]{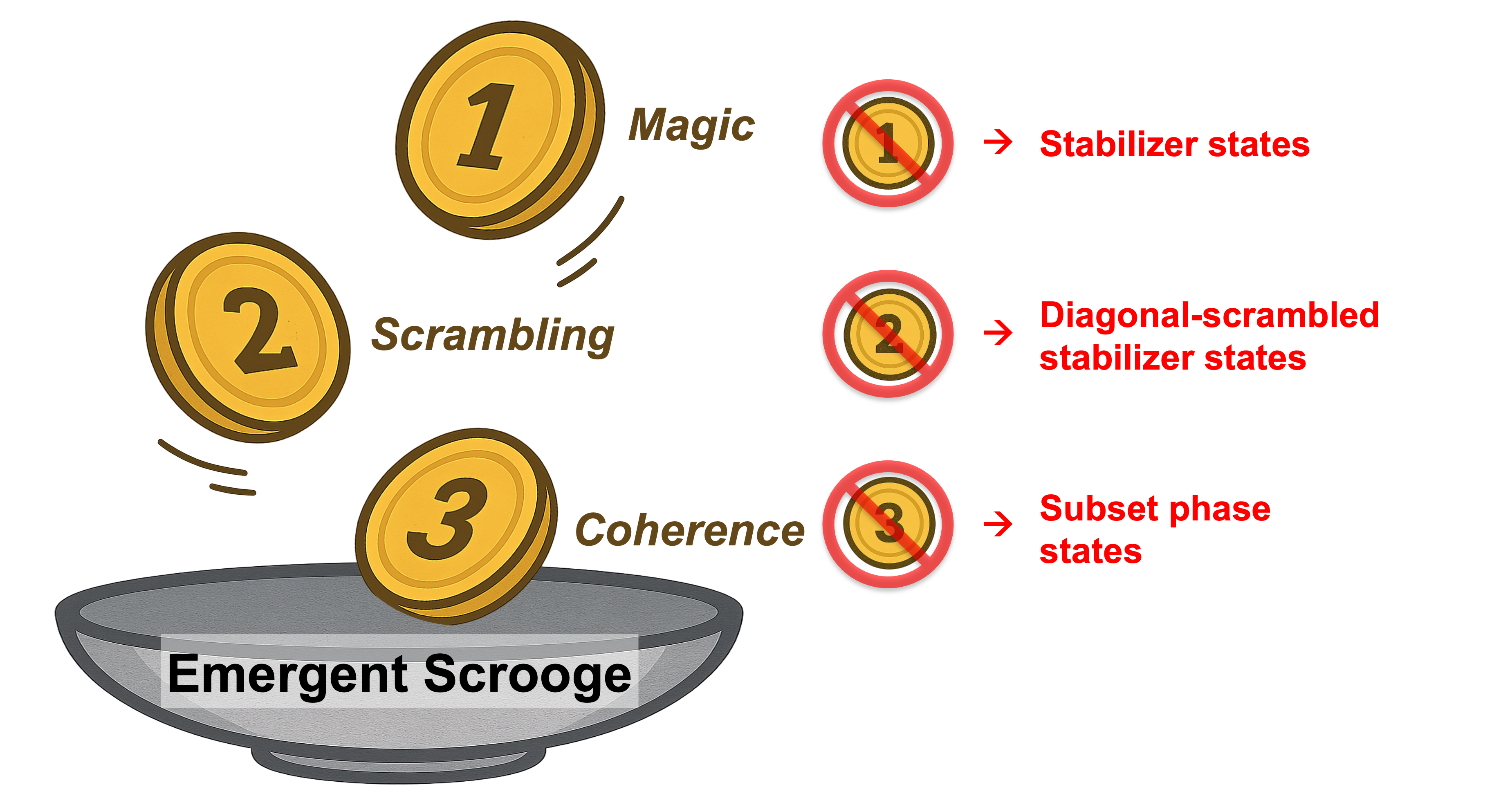}
    \caption{\textbf{ Essential physical ingredients for emergent Scrooge designs.} The emergence of Scrooge designs in projected ensembles requires the generator state $\ket{\Phi}_{AB}$ to exhibit magic (or nonstabilizerness), quantum information scrambling (via nonlocal entanglement), and coherence. When any of these ingredients is absent or insufficient, obstructions to Scrooge behavior can arise. Representative examples include stabilizer states, which lack magic; diagonal-scrambled stabilizer states (where a unitary diagonal in the computational basis is applied to subsystem $B$ of a stabilizer state), which only weakly scrambles quantum information between subsystems $A$ and $B$; and subset phase states, which may exhibit low coherence density depending on subset size.}
    \label{fig:resources}
\end{figure}

Our theorems have identified several general conditions under which quantum state ensembles with Scrooge-like behavior, as captured within the framework of Scrooge designs, can be provably shown to emerge: namely, dynamically in the temporal ensemble [Theorem~\ref{thm:global_scrooge}]; and via measurements in the projected ensemble, if the generator state is itself drawn from a Scrooge design [Theorem~\ref{thm:2kgenerator}] or if the measurement basis is rotated by a Haar design [Theorem~\ref{thm:ScroogeByMeasBasis}].
Importantly, these theorems demonstrate that the Scrooge ensemble appears in a myriad of quantum many-body settings. On the other hand, we do not expect our theorems to exhaustively cover all possible scenarios where Scrooge behavior can arise; moreover, they do not directly inform us as to the necessary {\it physical ingredients} that underlie its emergence.

Here, we discuss the necessary ingredients that a generator state $|\Psi\rangle_{AB}$ %
on a bipartite system $AB$ has to possess, in order for it to exhibit emergent Haar or Scrooge-like behavior in its projected ensemble.

Interestingly, recent works have begun to shed light on this question by quantitatively studying the effect of the amount of {\it quantum information resources}\footnote{These are certain properties of the state or evolution needed for quantum information processing to achieve quantum advantage  over classical information processing.} in governing the degree of universal randomness produced~\cite{vairogs2025extracting,loio2025quantum,liu2025coherence, varikuti2025deep}. 
\change{Recent work has shown that, for certain types of resources known as ``threshold-localizable'' resources, which include coherence and magic, a finite resource density is needed to achieve deep thermalization~\cite{feng2026quantum}.}
In particular, Ref.~\cite{liu2025coherence} studied the concept of {\it coherence}: a measure of the amount of superposition \change{in a fixed orthonormal reference basis, chosen to be the computational basis $\{\ket{z}\}$ by convention, motivated by experimental considerations}~\cite{baumgratz2014quantifying,streltsov2017colloquium}.
Intriguingly, they predicted that a generator state $|\Psi \rangle = \sum_z c_z |z\rangle$ may have near-maximal entanglement between $A$ and $B$ (hence its reduced density matrix is maximally mixed) yet coherence that is too low\footnote{This can be rigorously quantified by so-called relative entropy of coherence $C(\ket{\Psi_0}) = -\sum_z |c_z|^2 \ln \parens{|c_z|^2}$, which is just the Shannon entropy of the populations in the computational basis, if the state $|\Psi_0\rangle$ is pure. Low coherence here means $C = \alpha N$ for some $\alpha$ less than the model-dependent critical value $\alpha^*$.} such that it may fail to deeply thermalize to the Haar ensemble when measured in the computational basis. 

Explicit examples of such states are furnished by so-called random subset phase states~\cite{feng2025dynamics,chakraborty2025fast,liu2025coherence}, recently introduced in the field of pseudoentanglement~\cite{aaronson2024quantum}, with coherence tuned by the size of the subset in question.
Concretely, for low-coherence states, it was argued that the resulting distribution of the projected ensemble  is just a collection of (classical) computational states on the local subsystem --- clearly non-ergodic and having maximal accessible information instead of minimal. Further, this breaking of ergodicity happens in a robust fashion, with a phase transition separating the non-ergodic distribution from a deeply thermalized Haar distribution, upon tuning the coherence density past a critical finite value. 
Intuitively, the reason for this is because measurements, which we take to be in the computational basis, select only those global bit-strings $z = (z_A,z_B)$ in the decomposition of the generator state $|\Psi \rangle = \sum_z c_z |z\rangle$ in which $z_B$ agrees with the measurement outcome. If the number of such compatible $z$ (captured precisely by coherence) is below some critical value, this may result in a vanishing fraction of states $\ket{z_A}$ that contribute to the projected state --- clearly precluding it from behaving like a Haar random vector; see Ref.~\cite{liu2025coherence} for details.
The upshot is that coherence (and importantly not only a non-zero value of it, but a sufficiently high density of it!) is a necessary ingredient for the appearance of Scrooge-like behavior in the projected ensemble. 

We next focus on magic: a quantifier of the computational resources needed to describe quantum states beyond stabilizer states~\cite{kitaev2003fault,bravyi2005universal,liu2022many}. 
Recent works have quantified how magic governs the degree $k$ of Haar $k$-designs formed in the projected ensemble~\cite{vairogs2025extracting,loio2025quantum, varikuti2025deep}, but here we present a crisp example showing that magic is also a necessary ingredient for the appearance of Scrooge-like behavior in the projected ensemble. 
Consider an $N$-qubit stabilizer state $|\Psi_\text{Stab}\rangle_{AB}$  with reduced density matrix on a subsystem $A$  maximally mixed, and %
construct the projected ensemble by measuring $B$ in the standard computational basis (our argument in fact also applies more generally to any stabilizer basis). If the principle of maximum entropy for state ensembles applies, we should expect the projected ensemble to be well described by the Haar ensemble. Yet, as \change{previously} discussed in Sec.~\ref{sec:emergent_scrooge}, the projected ensemble is far from a 2-design, despite the fact that stabilizer states and Clifford unitaries form exact 3-designs. %
Our Theorems \ref{thm:2kgenerator} and~\ref{thm:ScroogeByMeasBasis} (henceforth collectively referred to as $``2k \to k"$ theorems), as well as Corollary~\ref{cor:2ktokdesign}, inform us that if we can augment the design properties of the initial state or measurement basis to be at least an approximate 4-design, then the projected ensemble is guaranteed to deeply thermalize to the Haar ensemble at the $k=2$ moment; this missing ingredient is provided by magic~\cite{lami2025quantum,vairogs2025extracting,haferkamp2022efficient}. 
More generally, our ``$2k \to k$" theorems can be harnessed to yield resource-theoretic bounds on the magic required for Scrooge universality to emerge.

However, simply adding non-Clifford elements to the system is not sufficient. %
To illustrate this point, consider the stabilizer state $|\Psi_\text{Stab}\rangle_{AB}$ of the previous example,  %
and apply a random diagonal unitary on $B$, i.e., $\ket{\Psi}_{AB} = (I_A \otimes U^\text{diag}_B)|\Psi_\text{Stab}\rangle_{AB}$, where $U^\text{diag}_B = \text{diag}\parens{e^{i\varphi_j}}$ in the computational basis, and $\varphi_j \in [0,2\pi)$ are uniformly distributed. 
The resulting state $|\Psi\rangle_{AB}$ (which we call a \textit{diagonal-scrambled} stabilizer state)
now possesses high magic, induced by the diagonal unitary. 
Nevertheless, it is clear that its projected ensemble $\mathcal{E}(\Psi)$ still forms only a $1$-design when $B$ is measured in the computational basis, since the  unitary $U_B$ commutes with the measurements, both of which are diagonal in the computational basis.
In contrast, if we had picked $U_B$ from a unitary $2k$-design, then Theorem~\ref{thm:ScroogeByMeasBasis} guarantees that the projected ensemble will now be close to a $k$-design.

What this example shows us is that the unitary (or dynamics) used to construct the generator state $|\Psi\rangle$ (which in the previous example consisted of the Clifford unitary used to prepare the stabilizer state followed by $U_B^{\text{diag}}$), starting from an unentangled product state on $AB$, needs to be sufficiently information-scrambling, and further in a {\it nonlocal way}, in order to achieve emergent Scrooge universality.  More precisely, we expect that if the scrambling unitary $U_B$ hides quantum information about $A$ (i.e., quantum correlations shared between $A$ and $B$) in nonlocal degrees of freedom in $B$, then the emergent Scrooge projected ensemble on $A$ also hides information about the measurement outcome on $B$ from any measurement on $A$ (i.e., attains minimal accessible information, precisely Scrooge behavior).

To summarize, the emergence of Scrooge behavior \change{in projected ensembles} requires the presence of coherence, magic, and information scrambling (nonlocal entanglement within the system). When any of these ingredients is absent or insufficient, obstructions to the formation of Scrooge projected ensembles can arise, see Fig.~\ref{fig:resources}. This parallels a well-known fact in quantum complexity theory: quantum advantage requires the coexistence of the resources of coherence, magic, and entanglement. 
When any of these ingredients is absent, quantum advantage is lost, i.e., efficient classical simulation becomes possible~\cite{josza2003on,aaronson2004improved,hugo2025role}.
 
In the next section, we perform numerical simulations supporting the aforementioned theoretical discussion, and explore the interplay between coherence, magic, and quantum information scrambling in producing Scrooge behavior within the projected ensemble 
across a variety of models.
These encompass regimes that lie beyond the settings required of our theorems, and thus complement the analytical results of Sec.~\ref{sec:emergent_scrooge}. 
Our simulations suggest that Scrooge designs can emerge even more broadly than in the settings for which we currently have rigorous guarantees.

\section{Numerical investigations}
\label{sec:numerics}

In this section, we numerically study the projected ensembles generated in several quantum many-body settings in which the key physical ingredients of coherence, magic and degree of information scrambling, can be systematically controlled:

(i) \textbf{Commuting quantum circuit evolution.} We consider generator states produced from so-called ``commuting quantum circuits'' (i.e., circuits composed of mutually commuting gates, studied in the quantum computing community)~\cite{shepherd2009temporally,bremner2010classical}. 
Concretely, starting from a product state $\ket{+}^{\otimes N}$, we apply a unitary circuit that is diagonal in the measurement (computational) basis. 
Although such circuits can generate states with extensive entanglement and magic, by changing the regions where interactions are applied, the circuits' commuting structure allows for controlled tuning of how quantum information is delocalized across the system, such that the resulting projected ensembles can either succeed or fail to form a good state design. In addition, by rotating the measurement basis, we can inject controlled amounts of coherence into the system. We will show that this results in a coherence-induced deep thermalization transition, as predicted recently by Ref.~\cite{liu2025coherence}.

(ii) \textbf{Doped Clifford circuit evolution.} Next, we study generator states produced from Clifford circuits doped with single-qubit non-Clifford gates. This construction allows information scrambling and magic to be tuned largely %
independently: Clifford dynamics efficiently scrambles quantum information, while the non-Clifford gates introduce controlled amounts of magic. We study how tuning these two control knobs affects the quality of the Scrooge designs formed in the projected ensemble, and find that  Scrooge behavior emerges only when both the circuit depth and density of non-Clifford gates are sufficiently large.

(iii) \textbf{Ground states of 1D many-body Hamiltonians.} Finally, we consider ground states of local 1D integrable Hamiltonians, which are characterized by low entanglement and low complexity.  %
To probe the role of the measurement basis, we apply various unitary rotations $U_B$ to the qubits on subsystem $B$ prior to measurement, which injects various amounts of magic and information scrambling. We find that if the rotations are insufficiently complex, no Scrooge behavior forms; while if the rotations are sufficiently complex Scrooge behavior can emerge.

Together, our numerical investigations allow us to %
showcase the roles of 
coherence, magic and quantum information scrambling in driving emergent local Scrooge behavior. They complement and go beyond existing numerical results which have demonstrated Scrooge behavior in late-time quenches or finite-temperature eigenstates of chaotic Hamiltonians~\cite{cotler2023emergent,mark2024maximum}: while the latter examples are physically relevant for experiments, it is difficult to control and isolate the individual ingredients driving Scrooge behavior in them as we do here.

In our numerical simulations, we quantify the closeness of a projected ensemble $\mathcal{E}(\Psi)$ generated by $\ket{\Psi}_{AB}$ to the Scrooge($\sigma_A$) $k$-designs via the $k$-th moment trace distance 
\begin{equation}
    \Delta^{(k)}(\sigma_A)=\frac{1}{2}\norm{\rho_{\mathcal{E}}^{(k)} - \rho_{\text{Scrooge}}^{(k)}(\sigma_A)}_1\,,
\end{equation}
where $\sigma_A$ is the reduced density matrix of $\ket{\Psi}_{AB}$ on $A$. By construction, the first-moment trace distance is $\Delta^{(1)} = 0$.

\subsection{Commuting quantum circuit evolution}
\begin{figure*}[t]
	\centering	
\subfigimg[width=0.95\textwidth]{}{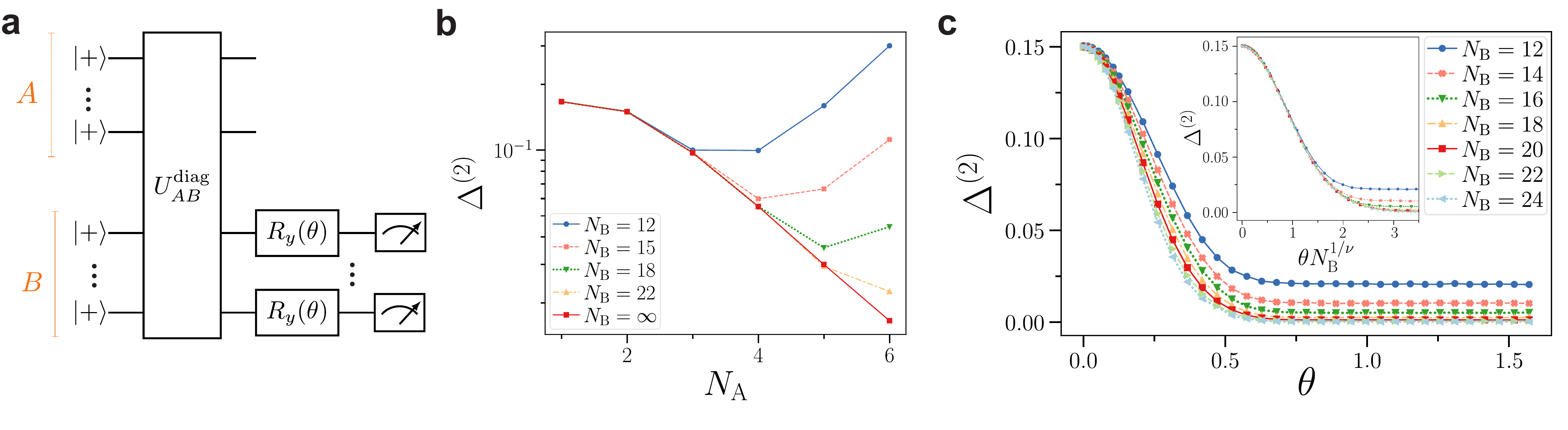}
    \caption{\textbf{Emergent $2$-designs from commuting quantum circuit evolution.} 
    \idg{a} Circuit to generate projected ensemble from commuting quantum circuit dynamics, Eq.~\eqref{eq:randphase}. We apply a random diagonal unitary $U_{AB}^{\text{diag}}$ on $\ket{+}^{\otimes N}$, rotate $B$ by angle $\theta$ around $y$-axis with single-qubit unitaries $R_y(\theta) = \exp(-i\theta Y/2)$, and measure subsystem $B$ in the computational basis.
    \idg{b} Trace distance to Haar $2$-design $\Delta^{(2)}$ against $N_\text{A}$ for $\theta=0$. For finite $N_\text{B}$, the moment operator of the projected ensemble is computed exactly (or using $3\times10^6$ measurement samples for $N_\text{B}=22$), and $\Delta^{(2)}$ itself is averaged over up to $20$ random instances of the phase state.
    $N_\text{B}=\infty$ is computed from the ensemble of uniform random phase states over $N_\text{A}$ qubits, where the moment operator is averaged over up to $10^7$ random instances. 
    \idg{c} %
    $\Delta^{(2)}$ against rotation angle $\theta$ for different $N_\text{B}$.
    The inset shows rescaled $\theta N_\text{B}^{1/\nu}$, where $\nu=2$. Here, $N_\text{A}=2$, $N=N_\text{A}+N_\text{B}$, and $\Delta^{(2)}$ is averaged over $10$ random instances. %
	}
	\label{fig:randphase}
\end{figure*}

First, we study how different  degrees of information scrambling in the circuit preparing a generator state, tuned by the support of the scrambler, affect the formation of Haar $k$-designs in its projected ensemble. 

We initialize a tripartite system $A \cup B_1 \cup B_2$ of qubits  in the state 
\begin{align}
\ket{\eta}=\prod_{i=1}^{N_A} \text{CZ}_{i,i+N_A} \ket{+}^{\otimes N},
\end{align}
where $\ket{+}=(\ket{0}+\ket{1})/\sqrt{2}$ and $\text{CZ}_{i,j}=\text{diag}(1,1,1,-1)$ is the controlled phase gate acting on the $i$th and $j$th qubits. We take the first $N_A$ qubits to constitute subsystem $A$ and its complement is $B = B_1 \cup B_2$. Here $|\eta\rangle$ describes a state composed of $N_A$ pairs \change{of maximally entangled states} between $A$ and $B_1$, with the remaining qubits in the $x$-polarized  $|+\rangle$ state. The reduced density matrix $\sigma_A$ on $A$ is hence maximally mixed. 

Next we scramble the quantum correlations initially shared between $A$ and $B_1$, into $B_2$. We consider applying a quantum circuit on $B = B_1 \cup B_2$ where the gates are all diagonal in the computational basis (which is also the measurement basis).
As these gates are all mutually commuting, this class of circuits has been termed ``commuting quantum circuits''\footnote{These are closely related to instantaneous quantum polynomial (IQP) circuits~\cite{shepherd2009temporally,bremner2010classical} studied in quantum complexity theory.}~\cite{ni2013commuting}.
We can model deep commuting circuit evolution, where the gates are randomly drawn, by a single random diagonal unitary $U^\text{diag}_B=\text{diag}(e^{i\varphi_j})$ applied on $B$, with $\varphi_j\in[0,2\pi)$ drawn uniformly at random. Then, the generator state reads\footnote{More generally, we may also  apply an arbitrary unitary $U_A$ on $A$, without affecting the discussion.}
\begin{align}
\ket{\Psi}_{AB} = (I_A\otimes U^\text{diag}_B) \ket{\eta}. 
\end{align}
Measuring $B$ now in the computational basis to construct the projected ensemble on $A$, it is clear that it fails to form a state $2$-design (which would be a Scrooge $2$-design in this case where $\sigma_A$ is maximally mixed), even though the state possesses both high entanglement and high magic. This failure occurs because
$U^\text{diag}_B$ commutes with the measurement operator, and because $\ket{\eta}$ is a stabilizer state; hence only a $1$-design projected ensemble can result from stabilizer measurements on $B$.

The failure to form a $k$-design (for $k > 1$) can  equivalently be understood by tracking the spread of quantum correlations in the Heisenberg picture: because the circuit is commuting, an operator initially supported on $B_2$ never has support on $A$ after the evolution. Thus, measurement outcomes on $B_2$ do not affect the projected ensemble, precluding the formation of $2$-designs on $A$. 
 
Suppose now we change the scrambler $U^\text{diag}_B$ into a diagonal unitary whose support {\it extends} to the full system $AB$ (see Fig.~\ref{fig:randphase}a), i.e.,
\begin{align}
\ket{\Psi}_{AB} &= U_{AB}^\text{diag}\ket{\eta}.
\end{align}
This yields so-called uniform random phase states~\cite{nakata2014generating}
\begin{align}
\label{eq:randphase}
\ket{\Psi}_{AB} & =  \frac{1}{2^{N/2}}\sum_{j}e^{-i\varphi_j}\ket{j}
\end{align}
where $\{\ket{j}\}_j$ are computational basis states on $AB$ and each $\varphi_j$ is uniformly distributed. 
In contrast to the previous case, now an operator supported on $B_2$ can spread throughout the entire system, allowing measurements on $B_2$ to be correlated with the state on $A$: quantum information is non-locally scrambled across the entire system. As we will now show, this stronger form of information scrambling enables Scrooge behavior to emerge in the projected ensemble.

Uniform random phase states~\eqref{eq:randphase} are known to form approximate $k$-designs with additive error $\Theta(k^2/2^{N})$~\cite{nakata2014generating}, while they are far from being a $k$-design in relative error for any $k>1$; this precludes the application of Corollary~\ref{cor:2ktokdesign}.
Nevertheless, we  investigate the nature of the projected ensemble through numerical simulations. 
In Fig.~\ref{fig:randphase}b we plot $\Delta^{(2)}$ against $N_\text{A}$ for different $N_\text{B}$. 
For finite $N_\text{B}$, we see $\Delta^{(2)}$ initially decays with $N_\text{A}$, then increases again when $N_\text{A}$ is on the same order as $N_\text{B}$, where our numerics indicate that the turning point (defined to be the minimal distance) occurs for $N_A \approx N_B/4$.
This suggests that as  we send $N_\text{B}\rightarrow\infty$, 
 $\Delta^{(2)}$ should decay with $N_A$ monotonically. 
 We see this in Fig.~\ref{fig:randphase}b, where the numerics in the $N_B\to\infty$ limit is performed by taking the projected ensemble to itself be uniform random phase states, but now over $N_A$ qubits (exact as $N_B \to \infty$). 
Note that the requirement of $N_A \gg 1$ arises because random phase states
form a $k$-design with additive error $\Theta(k^2/2^{N_A})$. Thus this is a conceptually different regime from the conventional setting of deep thermalization, which concerns the behavior of the projected ensemble in the thermodynamic limit $N_B\to\infty$ at fixed $N_A$. 

\begin{figure*}[t]
	\centering	
 \includegraphics[width=0.99\textwidth]{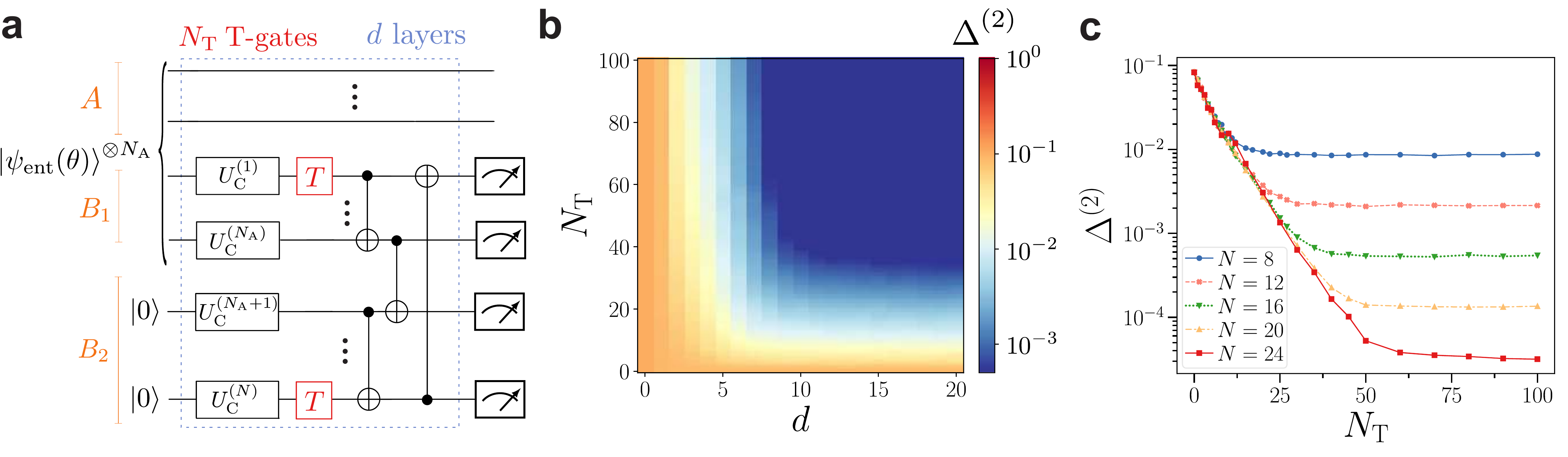}
    \caption{\textbf{Emergent Scrooge $2$-designs from doped Clifford circuit evolution.} \idg{a} We consider projected ensembles from the generator state Eq.~\eqref{eq:BellGen}, where we measure in the computational basis on $B=B_1\cup B_2$ after applying there a quantum circuit composed of $d$ layers of local Clifford gates (random single-qubit Clifford gates together with fixed CNOT gates) doped with  $N_\text{T}$ T-gates in total, randomly placed within the circuit. 
    \idg{b} Heat map of the trace distance to Scrooge $2$-design $\Delta^{(2)}$ in the  
 $d$ versus $N_\text{T}$ plane, for $\chi = \pi/6$. Here we fixed $N_A=1$, $N_B=N-N_\text{A}$, and $N=20$.
    \idg{c} $\Delta^{(2)}$ against $N_\text{T}$ for different total qubit numbers $N$ and fixed $d=30$.
    We choose $N_A=1$ and $\Delta^{(2)}$ is averaged over 500 random realizations of the circuit. Similar behavior is observed for higher moments $k > 2$, see \SM{}~\ref{sec:Cliffordmagicdepth}.
	}
	\label{fig:cliffdepth}
\end{figure*}

Deep thermalization behavior (local emergence of Haar designs for fixed $N_A$) can nevertheless be reproduced by applying a single-qubit rotation $R_y(\theta) = \exp(-i\theta Y/2)$, where $Y$ is the Pauli-$Y$ operator, to all qubits in subsystem $B$ prior to measurement. For any angle $\theta \neq 0$ (mod $\pi$) and fixed $N_A$, we provide numerical evidence in %
Fig.~\ref{fig:randphase}c that $\Delta^{(2)} \to 0$ as $N_B \to \infty$. This indicates that deep thermalization is obstructed for $\theta = 0$, but occurs for any nonzero rotation angle.

This observation aligns with the discussion in Ref.~\cite{liu2025coherence}. Those authors argued that a combination of the coherence of the initial state and the coherence of the measurement basis determines whether the projected ensemble on $A$ is deeply thermalized or not. Now, random phase states Eq.~\eqref{eq:randphase} have coherence $C = \alpha_0 N$ where $\alpha_0 = 1$  (recall Sec.~\ref{sec:physical_ingredients} for the definition of $C$) and the rotated measurement basis (the computational basis rotated by $R_y(\theta)$) has coherence $\alpha_m N_B$, where $\alpha_m=0$ for $\theta=0$ and $\alpha_m$ increases monotonically as $\theta$ increases (given by $\alpha_m \sim \theta^2 
\ln(1/\theta^2)$ for small $\theta$). Ref.~\cite{liu2025coherence} predicted a phase transition in the projected ensemble,  with $\alpha_0+\alpha_m<1$ corresponding to a non-deep-thermal phase and $\alpha_0+\alpha_m>1$ a deep-thermal one; thus $\alpha_0+\alpha_m=1$ is the critical point which maps to  $\theta_c = 0$ in our setup in which $\alpha_0 =1$. Our numerical findings in Fig.~\ref{fig:randphase}c confirm this; a finite-size scaling ansatz, which fits the data well, indicates that trace distance $\Delta^{(2)}$ approaches zero for any fixed nonzero $\theta$ as $N_B$ increases.

Thus,   our numerical investigations have shown  that the emergence of state designs requires sufficient scrambling (beyond just entanglement and magic) between $A$ and $B$, as well as coherence, which is induced by the global diagonal unitary and coherence-injecting $y$-rotations, respectively. %

\subsection{Doped Clifford circuits}

Next, we explore the necessity of both magic and information scrambling for the projected ensemble to exhibit emergent Scrooge designs. %
As depicted in Fig.~\ref{fig:cliffdepth}a, we consider an initial state defined on a tripartite system $A \cup B_1 \cup B_2$ such that $B = B_1 \cup B_2$, with tunable entanglement between $A$ and $B_1$:%
\begin{equation}\label{eq:BellGen}
    \ket{\Psi(\chi)}=\ket{\psi_\text{ent}(\chi)}^{\otimes N_\text{A}}\ket{0_{B_2}}^{\otimes N_\text{B}-N_\text{A}}\,.
\end{equation}
Here, $\ket{\psi_\text{ent}(\chi)}=\cos(\chi/2)\ket{0_A 0_{B_1}}+\sin(\chi/2)\ket{1_A 1_{B_1}}$
is a two-qubit state where the first qubit is in $A$ and the second in $B_1$. $\chi$ controls the entanglement between $A$ and $B_1$, %
with maximal entanglement between $A$ and $B_1$ achieved at $\chi=\pi/2$ (at this point, $|\psi_\text{ent}\rangle$ is a Bell state) resulting in the reduced density matrix $\sigma_A$ being maximally mixed. %
To probe the emergence of Scrooge designs away from the infinite-temperature limit, we fix $\chi = \pi/6$. 
We next apply on $B$ a circuit $U_B$ composed of $d$ layers, where each layer consists of randomly chosen single-qubit Clifford gates on each qubit and fixed CNOT gates arranged in a $1\text{D}$ nearest-neighbor geometry. Additionally, we dope the overall circuit with $N_\text{T}$ T-gates defined as $T=\text{diag}(1,e^{-i\pi/4})$~\cite{haferkamp2022efficient,Leone2021quantumchaosis}
randomly inserted over space and time. While the $d$ layers of Clifford CNOT gates introduce an increasing amount of information scrambling on $B$,
the $N_\text{T}$ T-gates inject magic into the circuit. 
This allows us to study the effects of information scrambling and magic on $\Delta^{(2)}$ relatively independently.

In Fig.~\ref{fig:cliffdepth}b, we show a heat map of $\Delta^{(2)}$ in the $d$-$N_T$ plane.
We see that for small $N_\text{T}$ or $d$, $\Delta^{(2)}$ is relatively large. 
Only with both sufficient magic and information scrambling (large $N_T$ and $d$) do we achieve a Scrooge $2$-design on $A$ with small $\Delta^{(2)}$~\cite{loio2025quantum}, as expected.
Our numerical results do not depend strongly on $\chi$; see \SM{}~\ref{sec:Cliffordmagicdepth}.

Next, in Fig.~\ref{fig:cliffdepth}c we study the role of magic and system size $N$ in more detail. We plot $\Delta^{(2)}$ for different $N$ and $N_\text{T}$ (for large $d$), finding that $\Delta^{(2)}$ decreases exponentially in $N_\text{T}$, reaching a saturation value due to finite-size effects. %
The saturation value of $\Delta^{(2)}$ decreases exponentially with $N$, which is reached for $N_\text{T}\approx 2.5 N$, closely matching the saturation transition found in magic resource theories~\cite{haug2025probing,tarabunga2025efficient}. 
This highlights that an extensive amount of magic is needed for deep thermalization to the Scrooge ensemble (more precisely, the emergence of Scrooge designs).

We note that similar behavior is observed when the doped Clifford circuit is applied on $AB$, instead of just $B$, which we show in \SM{}~\ref{sec:cliffTgenerator}. 

\begin{figure*}[t]
	\centering	
 \includegraphics[width=0.99\textwidth]{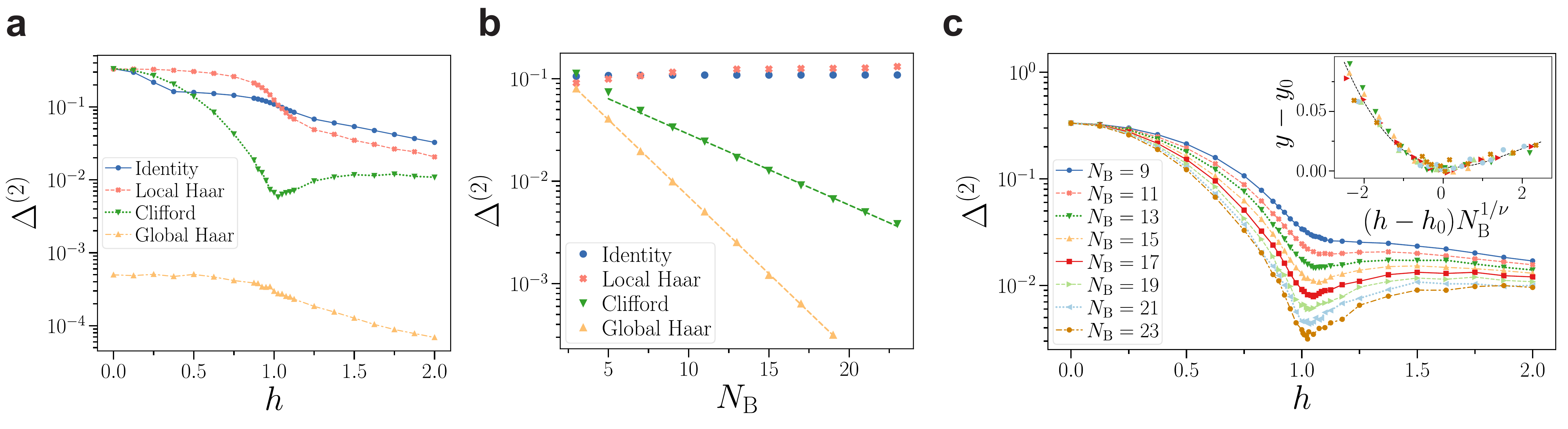}
    \caption{\textbf{Emergent Scrooge 2-designs from 1D integrable ground states.} We study projected ensembles generated from the ground state of the transverse-field Ising model~\eqref{eq:ising}, where we measure $B$ in various bases which arise from rotations of the computational basis by various unitaries $U_\text{B}$ applied on $B$ only.
    \idg{a} $\Delta^{(2)}$ against field $h$ for  $N_A = 1$ qubits, with $U_\text{B}$ chosen from the identity, single-qubit Haar random unitaries, random Clifford unitaries, or (global) unitaries drawn from the Haar measure on $N_\text{B}=19$ spins.  
    \idg{b} $\Delta^{(2)}$ against $N_\text{B}$ for different $U_\text{B}$ and $h=1$. Dashed lines show the fit with $\Delta^{(2)}\sim 2^{-\alpha N_\text{B}}$, where we find $\alpha_\text{Clifford}\approx0.23$ and $\alpha_\text{Haar}\approx0.5$.
    \idg{c} $\Delta^{(2)}$ against $h$ for $U_\text{B}$ being random Clifford unitaries and different $N_\text{B}$.
The inset shows the same data rescaled by defining $y=\log_2(\Delta^{(2)})/N_\text{B}$ and subtracting the field $h_0$ with minimal $y_0$ around the critical field $h_\text{c}=1$. When rescaling the field with $N_\text{B}^{1/\nu}$ where $\nu=1$ from the Ising universality class, the data for different $N_\text{B}$ nearly collapse to a single curve, a hallmark of universality of the critical field $h_\text{c}=1$ (see \SM{}~\ref{sec:ising}). We fit the collapsed data with a third-order polynomial as a dashed line. Similar behavior is observed for higher moments $k > 2$, see \SM{}~\ref{sec:ising}.
	}
	\label{fig:ising}
\end{figure*}

\subsection{Ground states of many-body Hamiltonians}

Finally, we study emergent Scrooge designs in %
naturally realizable many-body systems. 
In contrast to the usual paradigm of quantum chaotic systems where deep thermalization has been probed~\cite{cotler2023emergent,mark2024maximum,mok2025optimal}, here  we consider the projected ensemble generated from ground states of highly structured, integrable many-body systems, subject to various scramblers of quantum information on the system $B$ where measurements are taken. %
We will find that the projected ensemble can still form Scrooge $k$-designs if the measurement basis is sufficiently complex, demonstrating Theorem~\ref{thm:ScroogeByMeasBasis} in action. %

Concretely, we study the 1D transverse-field Ising Hamiltonian of $N$ spin-$\frac{1}{2}$ particles (which we will refer to also as qubits henceforth) with  external field $h$ on a chain with periodic boundary conditions,
\begin{equation}\label{eq:ising}
    H_\text{Ising}=-\sum_{j=1}^N X_jX_{j+1} -h \sum_{j=1}^N Y_j\,,
\end{equation}
where $X_j$ ($Y_j$) is the Pauli $X$ ($Y$) operator acting on the $j$th qubit. This model is well-known to be integrable, and in fact, mappable to free fermions by the Jordan-Wigner transformation~\cite{pfeuty1970one}. Consequently, its eigenstates are all highly structured, with low complexity, even at the quantum critical point $h = 1$.

We focus on the ground states $\ket{\Psi}_{AB}$ of $H_\text{Ising}$ in the spin-flip even sector with varying $h$, where $A$ is chosen to be a small contiguous region of spins and $B$ its complement. 
Prior to measuring $B$ in the computational basis to construct the projected ensemble, we apply different scrambling unitaries $U_B$ on $B$, chosen from: (i) the identity $I_\text{B}$, %
(ii) tensor products of $N_\text{B}$ (independently drawn) Haar random single-qubit unitaries, (iii) random Clifford unitaries, and (iv) Haar random unitaries over $N_\text{B}$ qubits.

In Fig.~\ref{fig:ising}a, we plot the second-moment trace distance $\Delta^{(2)}$ of the projected ensemble to the Scrooge ensemble constructed with the reduced density on $A$ for $N_A = 1$ and $N_B = 19$,
varying $h$ and for different choices of $U_\text{B}$. When $U_B$ is the identity or a product of single-qubit unitaries, %
we find the value of $\Delta^{(k)}$ remains high compared to the Clifford and Haar cases, even though $N_B \gg N_A$. This behavior can presumably be attributed to the low complexity nature of the ground state (even at the critical point $h = 1$), and to the fact that single-qubit measurements do not inject additional complexity. As a result, for product unitaries, quantum correlations between $A$ and $B$ are not efficiently scrambled into nonlocal correlations in $B$. 

In contrast, for the Clifford and Haar cases, the unitaries are apparently scrambling enough to yield comparatively lower $\Delta^{(2)}$, with Haar random unitaries on $B$ being (unsurprisingly) more scrambling in nature and hence exhibiting the lowest $\Delta^{(2)}$. 
We note that while Clifford unitaries on their own are insufficient to guarantee Scrooge $2$-designs due to an absence of magic, %
it is reasonable to expect that the ground state of $H_{\text{Ising}}$, despite its free-fermionic nature, supplies sufficient magic to enhance scrambling on $B$ toward effective $4$-design behavior, potentially allowing Theorem~\ref{thm:ScroogeByMeasBasis} to apply (although we do not establish this rigorously). Curiously, we find that $\Delta^{(2)}$ shows a pronounced dip at the quantum critical point $h\approx1$, suggesting that magic is highest among all ground states there: this is in line with observations made in previous works~\cite{haug2023quantifying}.

In Fig.~\ref{fig:ising}b, we study the scaling of the trace distance $\Delta^{(2)}$ with the size $N_\text{B}$ of the measured system at the critical point $h=1$.
We find that $\Delta^{(2)}$ is nearly independent of $N_B$ for local measurement bases  defined by rotating the computational basis by identity and single-qubit Haar unitaries, confirming that local measurement bases do not yield emergent Scrooge designs on $A$. %
In contrast, measuring in a basis that is scrambled by Cliffords or Haar-random unitaries on $B$ yields an exponentially decaying distance $\Delta^{(2)}\sim 2^{-\alpha N_B}$, with $\alpha \approx 0.2$ for Cliffords and $\alpha\approx 0.5$ for Haar, values that may indicate their scrambling power. %

Finally,  in Fig.~\ref{fig:ising}c, we study the Clifford case in greater detail for different $N_\text{B}$. We observe that the dip near $h\approx 1$ becomes more pronounced for larger $N_\text{B}$. This behavior is closely tied to the fact that magic is maximized at the critical point~\cite{haug2023quantifying}, and that magic is known to enhance the scrambling power of Clifford unitaries~\cite{vairogs2025extracting,lami2025quantum,leone2025non}.
Further, as shown in the inset, we find universal behavior around the critical point, with data for different $N_\text{B}$ nearly collapsing to a single curve when rescaling the field with $N_\text{B}^{1/\nu}$, where $\nu=1$, which is precisely the critical exponent of the Ising universality class~\cite{osterloh2002scaling,haug2023quantifying}. %
A qualitatively similar behavior is also observed for the ground state of the $1\text{D}$ Heisenberg model, %
as we show in \SM{}~\ref{sec:heisenberg}. %

Altogether, our numerical results show that magic and quantum information scrambling are %
key physical ingredients for emergent Scrooge designs \change{in projected ensembles}. These properties can manifest intrinsically in naturally occurring chaotic many-body states, or be injected externally by applying a scrambling unitary on subsystem $B$. We provide additional numerics supporting our observations in \SM{}~\ref{sec:stabmeasbasis}. \change{There, we show that a random stabilizer state can produce emergent state designs on $A$ if a single layer of $T$ gates is applied to $B$ before measuring in the $X$-basis, implying that it suffices for magic to be injected locally~\cite{cotler2023emergent}.} %
Our examples extend beyond the regimes in which Scrooge behavior can be rigorously established by our theorems or in those considered by other deep thermalization works, providing evidence that the emergence of information-stingy ensembles may be a generic phenomenon in quantum many-body systems.

\section{Discussion and Outlook}
\label{sec:outlook}

This work addresses a fundamental question: What universal principle governs the higher-order statistical fluctuations of quantum many-body systems? Building on prior insights into the information-stinginess of Nature, we provide a unified and rigorous framework for characterizing the emergence of maximally entropic Scrooge-like behavior. Our results, by clarifying when and why Scrooge universality arises, place earlier observations on a precise quantitative footing.

Our approach also highlights the robustness of Scrooge universality across diverse physical settings. In particular, Theorem~\ref{thm:2kgenerator} shows that local Scrooge behavior emerges whenever the global generator state is drawn from a Scrooge design, regardless of the microscopic mechanism by which that state is prepared. This mechanism-agnostic perspective opens the door to identifying new physical scenarios in which Scrooge behavior may arise.

Our numerical investigations further elucidate the physical resources underlying this universality. We find that quantum coherence, entanglement, non-stabilizerness (magic), and information scrambling play essential roles: removing any one of these ingredients can obstruct the emergence of Scrooge behavior \change{in projected ensembles}. Beyond these core ingredients, other resources may also play a role, including non-Gaussianity~\cite{walschaers2021nongaussian}, whose absence leads to restricted forms of deep thermalization~\cite{lucas2023generalized,liu2024deep}, and imaginarity~\cite{wu2021operational}, which can impose effective realness on projected ensembles and alter their universal structure from unitary to orthogonal symmetry classes~\cite{bhore2023deep}. Developing a systematic classification of such resources, and incorporating them into refined maximum-entropy principles, is an important direction for future work.

Several concrete open questions arise from our work. First, while our $``2k \to k"$ theorems are optimal for $k=2$, it remains open whether the required design order can be reduced for high moments, for example to $``k + \bigO{1} \to k"$.
Second, our findings suggest a potential unifying perspective on two complementary notions of information hiding. 
Applying a scrambling unitary to the measurement basis on $B$ (as in Theorem~\ref{thm:ScroogeByMeasBasis}) hides \textit{quantum} information about $A$ in highly nonlocal observables on $B$. On the other hand, the resulting projected ensemble on $A$ hides \textit{classical} information about the measurement outcome on $B$. Whether these two mechanisms fit into a common theoretical framework remains an intriguing question. 
Third, the universality of Scrooge ensembles may enable new applications in quantum information science. Emergent state designs in infinite-temperature deep thermalization have already led to advances in learning and benchmarking analog quantum simulators~\cite{choi2023preparing,tran2023measuring,mcginley2023shadow,mok2025optimal}. Extending these ideas to Scrooge designs could yield analogous protocols operating in more general, physically constrained settings. 

Finally, our results raise the broader question of whether an information-theoretic analogue of the maximum entropy principle extends to even more general settings, such as incomplete measurements or noisy dynamics, where each state in a projected ensemble becomes mixed rather than pure~\cite{milekhin2025observable,yu2025mixed}. An affirmative answer would point toward a unified framework for the emergence of universal randomness in both pure and mixed quantum ensembles, marking a significant conceptual advance in quantum statistical mechanics.

\section*{Acknowledgments}
We thank Andreas Elben, Soumik Ghosh, Hsin-Yuan Huang, Daniel Mark, Max McGinley, Alexey Milekhin, Arjun Mirani, Yihui Quek, Thomas Schuster, Federica Surace, Shreya Vardhan, and Michelle Xu for insightful discussions.  
We are especially grateful to Soonwon Choi for improving the presentation of our manuscript, and to Max McGinley and Thomas Schuster for sharing their preliminary results~\cite{mcginley2025scrooge} with us, which were invaluable for our analysis.
W.~W.~H.~is supported by the National Research Foundation (NRF), Singapore, through the NRF Fellowship NRF-NRFF15-2023-0008, and through the National Quantum Office, hosted in A*STAR, under its Centre for Quantum Technologies Funding Initiative (S24Q2d0009).
J.~P.~acknowledges support from the U.S. Department of Energy, Office of Science, National Quantum Information Science Research Centers, Quantum Systems Accelerator, and the National Science Foundation (PHY-2317110). The Institute for Quantum Information and Matter is an NSF Physics Frontiers Center.

\bibliography{bib}
\clearpage

\let\addcontentsline\oldaddcontentsline

\appendix

\onecolumngrid
\newpage 

\setcounter{secnumdepth}{2}
\setcounter{equation}{0}
\setcounter{figure}{0}
\setcounter{section}{0}

\renewcommand{\thesection}{\Alph{section}}
\renewcommand{\thesubsection}{\arabic{subsection}}
\renewcommand*{\theHsection}{\thesection}

\clearpage
\begin{center}

\textbf{\large Appendix}
\end{center}
\setcounter{equation}{0}
\setcounter{figure}{0}
\setcounter{table}{0}

\makeatletter

\renewcommand{\thefigure}{S\arabic{figure}}

Here we provide proofs and additional details supporting the claims in the main text.

\makeatletter
\@starttoc{toc}

\makeatother

\section{Notation and preliminaries}
\label{app:preliminaries}
In this section, we provide a self-contained introduction to the notation and mathematical preliminaries, which we use throughout this work.

\subsection{Norms of operators and random variables}
For an arbitrary operator $A$, we denote its Schatten-$p$ norm by $\norm{A}_p$, where
\begin{equation}
    \norm{A}_p = [\Tr(|A|^p)]^{1/p}, \quad p \in [1,\infty),
\end{equation}
and $\norm{A}_\infty$ is the operator norm. Schatten-$p$ norms satisfy monotonicity: for $1 \leq p \leq q \leq \infty$,
\begin{equation}
\label{eq:monotonicity_schatten}
    \norm{A}_p \geq \norm{A}_q.
\end{equation}

For an arbitrary random variable $X$, we denote its $L^p$-norm by $\norm{X}_{L^p}$, where
\begin{equation}
    \norm{X}_{L^p} = [\E (|X|^p)]^{1/p}, \quad p \in [1,\infty),
\end{equation}
and $\norm{X}_{L^\infty} = \text{ess} \sup |X|$ is the essential supremum of $X$. $L^p$ norms satisfy the following inequality: for $1 \leq p \leq q \leq \infty$,
\begin{equation}
    \norm{X}_{L^p} \leq \norm{X}_{L^q}.
\end{equation}

\subsection{Symmetric subspace, permutation operators, and moments of the Haar ensemble}

Given the $k$-fold Hilbert space $\mathcal{H}^{\otimes k}$, the symmetric subspace of $\mathcal{H}^{\otimes k}$, denoted $\mathcal{H}_{\text{sym}}^{(k)}$, is the vector space spanned by all states that are invariant under an arbitrary permutation of the $k$ replicas. To each permutation $\pi \in S_k$ (where $S_k$ is the symmetric group of order $k$), we can associate a corresponding permutation operator $\hat{\pi} \in \mathcal{L}\parens{\mathcal{H}^{\otimes k}}$ that permutes between the $k$ copies. The permutation operator $\hat{\pi}$ is a unitary representation of $\pi$. Thus,
\begin{equation}
   \mathcal{H}_{\text{sym}}^{(k)} = \text{span}\bparens{\ket{\psi} \in \mathcal{H}^{\otimes k} \bigg| \hat{\pi}\ket{\psi} = \ket{\psi} \quad \forall \pi \in S_k}.
\end{equation}
The dimension of the symmetric subspace $\mathcal{H}_{\text{sym}}^{(k)}$ is
\begin{equation}
    D_k \equiv \binom{D+k-1}{k},
\end{equation}
where $D$ is the dimension of the Hilbert space $\mathcal{H}$. For $k^2 \ll D$, it is useful to write
\begin{equation}
    D_k = \frac{D^k}{k!} \parens{1 + \bigO{\frac{k^2}{D}}}.
\end{equation}
Thus, $D_k \approx D^k/k!$. The orthogonal projector onto the symmetric subspace has a special meaning: it is proportional to the $k$th moment of the Haar ensemble $\text{Haar}(D)$,
\begin{equation}
    \rho_{\text{Haar}}^{(k)} = \frac{1}{D_k} \hat{P}_{\text{sym}}^{(k)} = \frac{1}{k! D_k} \sum_{\pi \in S_k} \hat{\pi},
\label{eq:rho_haar_k}
\end{equation}
where
\begin{equation}
    \hat{P}_{\text{sym}}^{(k)} = \frac{1}{k!}\sum_{\pi \in S_k} \hat{\pi}
\end{equation}
is the orthogonal projector onto $\mathcal{H}_{\text{sym}}^{(k)}$. This relationship can be derived using the Schur-Weyl duality~\cite{harrow2013church}. In a slight abuse of notation, we will also use $\text{Haar}$($D$) to denote the Haar measure on the unitary group $\text{U}$(D). The mathematical object in consideration (pure quantum state or unitary operator) will be clear from the context.

In this paper, we often decompose the system into complementary subsystems $A$ and $B$. Thus, $\mathcal{H} = \mathcal{H}_A \otimes \mathcal{H}_B$, with Hilbert space dimensions $D_A$ and $D_B$, respectively ($D = D_A D_B$). To disambiguate the notation, we will use subscripts to indicate the subsystem in consideration, where appropriate. For example, the dimension of the $k$-fold symmetric subspaces for subsystems $A$ and $B$ will be denoted $D_{A,k}$ and $D_{B,k}$, respectively.

\subsection{Weingarten calculus for the unitary group}
The Weingarten calculus provides a very useful tool for evaluating polynomial functions of Haar random states and unitaries. Here, we will briefly introduce the results relevant for this work, and establish the notation used throughout the manuscript. A detailed treatment can be found in Refs.~\cite{collins2017weingarten,collins2022weingarten,kostenberger2021weingarten}.

For permutations $\sigma, \pi \in S_k$, we define the $k! \times k!$ Gram matrix $G$, with matrix elements
\begin{equation}
    G_{\sigma \pi} = \Tr \parens{\hat{\sigma}^\dag \hat{\pi}} = D^{\#\text{cycles}(\sigma^{-1}\pi)},
\end{equation}
where $\hat{\sigma}$ and $\hat{\pi}$ are the permutation operators defined above. $\#\text{cycles}(\pi)$ counts the number of disjoint cycles in the permutation $\pi$. The inverse of $G$ gives the Weingarten matrix, where the matrix elements
\begin{equation}
    \text{Wg}(\sigma^{-1}\pi,D) \equiv (G^{-1})_{\sigma \pi}
\end{equation}
are known as the Weingarten functions associated to the unitary group. Note that the Gram matrix is only invertible for $1 \leq k \leq D$, which we will assume throughout this work. The main formula for our purposes is the $k$-fold twirling identity,
\begin{equation}
    \Eset{U \sim \text{Haar}(D)} \sparens{U^{\otimes k} A U^{\dag \otimes k}} = \sum_{\sigma,\pi \in S_k} \text{Wg}(\sigma^{-1}\pi,\change{D}) \text{Tr}(A\hat{\pi}^\dag) \hat{\sigma}.
\end{equation}
As a consistency check, setting $A = \parens{\ket{0}\bra{0}}^{\otimes k}$ reproduces Eq.~\eqref{eq:rho_haar_k}, for any arbitrary reference state $\ket{0}$.  

\subsection{Projected ensemble}

Here, we give a concise review of the projected ensemble, and establish the notation used in this work. Starting from a bipartite quantum state $\ket{\Psi}_{AB}$, which we refer to as the \textit{generator state}, we measure subsystem $B$ in a complete orthonormal basis $\{\ket{z}\}_{z=1}^{D_B}$. By default, we will choose the measurement basis to be the computational basis, unless stated otherwise. Note that there is no loss of generality here, since any rotation to the measurement basis can be absorbed into the definition of $\ket{\Psi}_{AB}$. This yields the projected state
\begin{equation}
    \ket{\psi_z} = \frac{(I_A \otimes \bra{z}_B)\ket{\Psi}_{AB}}{\sqrt{p_z}} \equiv \frac{\ket{\tilde{\psi}_z}}{\sqrt{p_z}}
\end{equation}
on subsystem $A$, where
\begin{equation}
    p_z = \braket{\tilde{\psi}_z|\tilde{\psi}_z} = (I_A \otimes \bra{z}_B) \parens{\ket{\Psi}\bra{\Psi}} (I_A \otimes \ket{z}_B)
\end{equation}
is the Born probability of measuring the outcome $z$ on $B$. Collectively, this defines the projected ensemble on $A$ generated by $\ket{\Psi}_{AB}$,
\begin{equation}
    \mathcal{E}(\Psi) = \bparens{p_z, \ket{\psi_z}}_{z=1}^{D_B}.
\end{equation}
The $k$th moment of the projected ensemble reads
\begin{equation}
    \rho_{\mathcal{E}}^{(k)}(\Psi) = \sum_{z=1}^{D_B} p_z \parens{\ket{\psi_z}\bra{\psi_z}}^{\otimes k} = \sum_{z=1}^{D_B} \frac{\parens{\ket{\tilde{\psi}_z}\bra{\tilde{\psi}_z}}^{\otimes k}}{\braket{\tilde{\psi}_z|\tilde{\psi}_z}^{k-1}}.
\end{equation}
For $k = 1$, this is exactly the reduced density matrix of $\ket{\Psi}_{AB}$ on $A$. To quantify the statistical closeness between the projected ensemble and a reference ensemble $\mathcal{E}_{\text{ref}}$ (which, in the context of this work, is the Scrooge ensemble or a variant of it), we compute the trace distance of their $k$th moments, 
\begin{equation}
    \Delta^{(k)} = \frac{1}{2} \norm{\rho_{\mathcal{E}}^{(k)}(\Psi) - \rho_{\text{ref}}^{(k)}}_1,
\end{equation}
where $\rho_{\text{ref}}^{(k)}$ is the $k$th moment of $\mathcal{E}_\text{ref}$. In this work, we will often consider the scenario where $\ket{\Psi}_{AB}$ is drawn from some distribution. Thus, a useful metric is the average trace distance
\begin{equation}
    \E_\Psi \sparens{\Delta^{(k)}} = \frac{1}{2} \E_{\Psi} \norm{\rho_{\mathcal{E}}^{(k)}(\Psi) - \rho_{\text{ref}}^{(k)}}_1.
\end{equation}
It is important that the trace distance is computed before the average. This implies that if $\E_\Psi \sparens{\Delta^{(k)}}$ is small, the projected ensemble $\rho_\mathcal{E}^{(k)}(\Psi)$ is close to $\mathcal{E}_{\text{ref}}$ up to the $k$th moment, for \textit{any} generator state $\ket{\Psi}$ drawn from the distribution, with high probability. This fact can be easily derived from Markov's inequality, which gives
\begin{equation}
    \mathbb{P}_{\Psi}\parens{\Delta^{(k)} > \delta} \leq \frac{\E_\Psi\sparens{\Delta^{(k)}}}{\delta}
\end{equation}
for any $\delta > 0$. Computing the trace distance analytically is often difficult. It is more analytically tractable to compute the Schatten $2$-norm distance (also known as the Hilbert-Schmidt distance), which gives an upper bound for $\E_\Psi\sparens{\Delta^{(k)}}$:
\begin{equation}
    \E_\Psi \sparens{\Delta^{(k)}} \leq \frac{1}{2} \sqrt{D_{A,k}} \parens{\E_\Psi \sparens{\norm{\rho_{\mathcal{E}}^{(k)}(\Psi) - \rho_\text{ref}^{(k)}}_2^2}}^{1/2}.
\end{equation}

\subsection{Scrooge ensemble}

The Scrooge ensemble is defined in Ref.~\cite{josza1994lower} as the state ensemble $\mathcal{E}$ that attains the minimum accessible information $\mathcal{I}_{\text{acc}}(\mathcal{E})$, among all ensembles that realize a given density matrix $\sigma$. Here, we briefly review the concept of accessible information~\cite{nielsen2011quantum,preskill1998lecture}, which can be understood by the following scenario. Bob samples a state drawn from the ensemble $\mathcal{E} = \{p_z, \ket{\psi_z}\}_z$ and sends the state to Alice. The ensemble $\mathcal{E}$ is known to both parties a priori. Alice's goal is to determine which state Bob has sent, i.e., the classical label $z$, by performing a measurement (mathematically represented by the POVM $M$) on the state. In the context of the projected ensemble, $p_z$ is the probability of measuring the outcome $z$ on subsystem $B$, and $\ket{\psi_z}$ is the normalized state on subsystem $A$ conditioned on the outcome $z$. Note that, as explained in Sec.~\ref{sec:preliminaries} of the main text, we can more generally consider $z$ in $\mathcal{E}$ to be a continuous label, upon which the probability $p_z$ should be replaced by a probability measure $\mu(z)$.

The accessible information quantifies how much classical information about $z$ can be gained by Alice via applying an optimal singly-copy measurement on the quantum state. Thus, we define
\begin{equation}
    \mathcal{I}_{\text{acc}}(\mathcal{E}) = \sup_{M} \, \mathcal{I}(\mathcal{E}:M),
\end{equation}
where $\mathcal{I}(\mathcal{E}:M)$ is the mutual information between the measurement $M$ and the ensemble $\mathcal{E}$. This leads to the Scrooge ensemble (Definition~\ref{defn:scrooge})
\begin{equation}
\begin{aligned}
    \text{Scrooge}(\sigma) = \argmin_{\mathcal{E}}& \quad \mathcal{I}_{\text{acc}}(\mathcal{E}), \\
   \text{subject to}& \quad \Eset{\psi \sim \mathcal{E}} \ket{\psi}\bra{\psi} = \sigma.
\end{aligned}
\end{equation}

The Scrooge ensemble can be explicitly constructed as
\begin{equation}
    \text{Scrooge}(\sigma) = \bparens{D\braket{\phi|\sigma|\phi} \text{d}\phi, \frac{\sqrt{\sigma}\ket{\phi}}{\norm{\sqrt{\sigma}\ket{\phi}}}},
\end{equation}
with $\text{d}\phi$ the Haar measure on $\mathbb{C}^D$. The $k$th moment of the Scrooge ensemble is computed as follows
\begin{equation}
\begin{aligned}\rho_{\text{Scrooge}}^{(k)}(\sigma) &= \Eset{\psi\sim\text{Scrooge}(\sigma)} \sparens{\parens{\ket{\psi}\bra{\psi}}^{\otimes k}} \\ &= \Eset{\phi \sim \text{Haar}(D)} \sparens{D\braket{\phi|\sigma|\phi} \frac{\parens{\sqrt{\sigma}\ket{\phi}\bra{\phi}\sqrt{\sigma}}^{\otimes k}}{\braket{\phi|\sigma|\phi}^{k}}} \\&= D \Eset{\phi \sim \text{Haar}(D)}\sparens{\frac{\parens{\sqrt{\sigma}\ket{\phi}\bra{\phi}\sqrt{\sigma}}^{\otimes k}}{\braket{\phi|\sigma|\phi}^{k-1}}}.
\end{aligned}
\end{equation}
In the limit $\sigma \to I/D$ (with $I$ the identity operator), $\rho_{\text{Scrooge}}^{(k)}(\sigma)$ reduces to $\rho_{\text{Haar}}^{(k)}$ in Eq.~\eqref{eq:rho_haar_k}. More generally, for any measurable function $f$,
\begin{equation}
    \Eset{\psi \sim \text{Scrooge}(\sigma)}[f(\psi)] = D \Eset{\phi \sim \text{Haar}(D)} \sparens{\braket{\phi|\sigma|\phi}f\parens{\frac{\sqrt{\sigma}\ket{\phi}}{\braket{\phi|\sigma|\phi}^{1/2}}}}.
\end{equation}
Unlike the Haar ensemble, the $k$th moment of Scrooge($\sigma$) does not have a simplified form in terms of permutation operators. The main technical difficulty arises from the form of $\rho_{\text{Scrooge}}^{(k)}(\sigma)$ written above, which is a rational function of the Haar random state $\ket{\phi}$. This implies that the standard Weingarten calculus, which is very useful in evaluating polynomial functions of Haar random states, does not apply here. Nonetheless, if $\sigma$ is a low-purity state, we can obtain a simple approximation to $\rho_{\text{Scrooge}}^{(k)}(\sigma)$, with an error that is controlled by the purity of $\sigma$. This will be discussed in Appendix~\ref{app:scrooge_approximation}.

\subsection{Scrooge $k$-designs}
State $k$-designs provide a low-order approximation of the Haar ensemble, by matching only the first $k$ statistical moments of the Haar ensemble. We can generalize this definition to Scrooge$(\sigma)$ $k$-designs, which match the first $k$ moments of Scrooge$(\sigma)$. The moments do not need to match exactly; for practical purposes it suffice for the moments to be close, up to a small error $\epsilon$. We reproduce Definition~\ref{defn:scrooge_designs} from the main text.
\begin{definition}[Approximate Scrooge $k$-designs] Let $\mathcal{E}$ be an ensemble of pure states, with $k$th moment
\begin{equation}
    \rho_{\mathcal{E}}^{(k)} = \Eset{\psi \sim \mathcal{E}} \sparens{\parens{\ket{\psi}\bra{\psi}}^{\otimes k}}.
\end{equation}
$\mathcal{E}$ is a Scrooge$(\sigma)$ $k$-design with additive error $\epsilon$ if
\begin{equation}
    \frac{1}{2}\norm{\rho_{\mathcal{E}}^{(k)} - \rho_{\text{Scrooge}}^{(k)}(\sigma)}_1 \leq \epsilon\,,
\end{equation}
and a Scrooge$(\sigma)$ $k$-design with relative error $\epsilon$ if 
\begin{equation}
(1-\epsilon)\rho_{\text{Scrooge}}^{(k)}(\sigma) \preceq \rho_{\mathcal{E}}^{(k)} \preceq (1+\epsilon)\rho_{\text{Scrooge}}^{(k)}(\sigma)\,.
\end{equation}
Here, $\rho_{\text{Scrooge}}^{(k)}(\sigma)$ is the $k$th moment of Scrooge$(\sigma)$, given in Eq.~\eqref{eq:scrooge_kth_mom_exact}.
\end{definition}

The notation $A \preceq B$ for positive semidefinite operators $A$ and $B$ is equivalent to the statement that $B - A$ is positive semidefinite.

A relative error of $\epsilon$ implies an additive error of $\epsilon$. This can be easily derived: Suppose $\mathcal{E}$ has relative error $\epsilon$. Then,
\begin{equation}
\begin{aligned}
    \frac{1}{2}\norm{\rho_{\mathcal{E}}^{(k)} - \rho_{\text{Scrooge}}^{(k)}(\sigma)}_1 &= \sup_{0 \preceq P \preceq I} \Tr\sparens{P\parens{\rho_{\mathcal{E}}^{(k)} - \rho_{\text{Scrooge}}^{(k)}(\sigma)}} \\
    &\leq \epsilon \sup_{0 \preceq P \preceq I} \Tr \parens{P \rho_{\text{Scrooge}}^{(k)}(\sigma)} \\
    &\leq \epsilon \norm{\rho_{\text{Scrooge}}^{(k)}(\sigma)}_1 \\
    &= \epsilon.
\end{aligned}
\end{equation}
However, the converse is not true. The relative error can be arbitrarily large, even if the additive error is small. For example, if $\rho_{\mathcal{E}}^{(k)}$ has support on the null space of $\rho_{\text{Scrooge}}^{(k)}(\sigma)$, the relative error is infinitely large.

For ensembles of quantum states, the additive error is a meaningful measure of statistical closeness between the ensembles. In the context above, if $\mathcal{E}$ is a Scrooge($\sigma$) $k$-design with additive error $\epsilon$, then any quantum measurement that can act collectively on up to $k$ copies of the quantum state can only distinguish between $\mathcal{E}$ and Scrooge$(\sigma)$ with a success probability of at most~\cite{nielsen2011quantum}
\begin{equation}
    p_\text{succ} = \frac{1 + \epsilon}{2}.
\end{equation}
In other words, if the additive error $\epsilon$ is small, the optimal measurement strategy is only marginally better than the naive strategy of random guessing, which has a success probability of $1/2$.

\subsection{Approximation of the $k$th moment}
To aid the analysis in the paper, we introduce the following technical lemma, which provides an approximation to the $k$th moment of an ensemble.
\begin{lemma}\normalfont[$k$th moment approximation]\label{lemma:general_approx}
Consider the ensemble of pure states
\begin{equation}
    \mathcal{E} = \bparens{p_j = \braket{\psi_j^\prime|\psi_j^\prime}, \ket{\psi_j} = \frac{\ket{\psi_j^\prime}}{\sqrt{p_j}} }_{j=1,\ldots,|\mathcal{E}|}
\end{equation}
with $k$th moment
\begin{equation}
    \rho_{\mathcal{E}}^{(k)} = \sum_{j=1}^{|\mathcal{E}|} \frac{(\ket{\psi_j^\prime}\bra{\psi_j^\prime})^{\otimes k}}{p_j^{k-1}},
\end{equation}
where each unnormalized state $\ket{\psi_j^\prime}$ in the ensemble is randomly drawn from a distribution $\mathcal{D}_j$.
Construct the proxy
\begin{equation}
    \tilde{\rho}_{\mathcal{E}}^{(k)} = \sum_{j=1}^{|\mathcal{E}|}\frac{ (\ket{\psi_j^\prime}\bra{\psi_j^\prime})^{\otimes k}}{q_j^{k-1}},
\end{equation}
where $q_j > 0$ are independent of $\ket{\psi_j^\prime}$. Then,
\begin{equation}
    \Eset{\psi_j^\prime \sim \mathcal{D}_j}\norm{\rho_{\mathcal{E}}^{(k)} - \tilde{\rho}_{\mathcal{E}}^{(k)}}_1 \leq \change{\sum_{j=1}^{|\mathcal{E}|}} \parens{\Eset{\psi_j^\prime \sim \mathcal{D}_j} p_j^2}^{1/2} \parens{1 - \frac{2}{q_j^{k-1}} \Eset{\psi_j^\prime \sim \mathcal{D}_j} p_j^{k-1} + \frac{1}{q_j^{2k-2}} \Eset{\psi_j^\prime \sim \mathcal{D}_j} p_j^{2k-2}}^{1/2}
\label{eq:approx_error_general}
\end{equation}
\end{lemma}
\begin{proof}
\begin{equation}
\begin{aligned}
    \Eset{\psi_j^\prime \sim \mathcal{D}_j} \norm{\rho_{\mathcal{E}}^{(k)} - \tilde{\rho}_{\mathcal{E}}^{(k)}}_1 &\leq \Eset{\psi_j^\prime \sim \mathcal{D}_j} \sum_{j=1}^{|\mathcal{E}|} p_j^k \abs{\frac{1}{p_j^{k-1}} - \frac{1}{q_j^{k-1}}} \quad &\text{(Triangle inequality)} \\
    &= \sum_{j=1}^{|\mathcal{E}|} \Eset{\psi_j^\prime \sim \mathcal{D}_j} p_j \abs{1 - \frac{p_j^{k-1}}{q_j^{k-1}}} \quad &\text{(Linearity of expectation)} \\
    &\leq \change{\sum_{j=1}^{|\mathcal{E}|}} \parens{\Eset{\psi_j^\prime \sim \mathcal{D}_j} p_j^2}^{1/2} \parens{1 - \frac{2}{q_j^{k-1}} \Eset{\psi_j^\prime \sim \mathcal{D}_j} p_j^{k-1} + \frac{1}{q_j^{2k-2}} \Eset{\psi_j^\prime \sim \mathcal{D}_j} p_j^{2k-2}}^{1/2} \quad &\text{(Cauchy-Schwarz)}.
\end{aligned}
\end{equation}
\end{proof}
Lemma~\ref{lemma:general_approx} formalizes the intuition that $\rho_{\mathcal{E}}^{(k)}$ and the proxy $\tilde{\rho}_{\mathcal{E}}^{(k)}$ are close, if the random variables $p_j$ are close to $q_j$. This allows us to conveniently replace $p_j$ (which depends on the distribution $\mathcal{D}_j$), by the weight $q_j$ (independent of the distribution $\mathcal{D}_j$), with an error that can be controlled.

The following lemma gives us control over the unnormalized projected ensemble, generated by states drawn from distributions that are $\epsilon$-close in relative error.

\begin{lemma}\label{lemma:projens_relerrorgens}
Let $\mathcal{E}_1$ and $\mathcal{E}_2$ be pure state ensembles which are $\epsilon$-close in relative error up to the first $2k$ moments, for some $k \in \mathbb{N}$, i.e.,
\begin{equation}
    (1-\epsilon) \Eset{\Psi \sim \mathcal{E}_1} \parens{\ket{\Psi}\bra{\Psi}}^{\otimes 2k} \preceq \Eset{\Psi \sim \mathcal{E}_2} \parens{\ket{\Psi}\bra{\Psi}}^{\otimes 2k} \preceq (1+\epsilon)\Eset{\Psi \sim \mathcal{E}_1} \parens{\ket{\Psi}\bra{\Psi}}^{\otimes 2k},
\label{eq:relative_error_defn}
\end{equation}
with $\epsilon \geq 0$. Let $\tilde{\rho}_{\mathcal{E}}^{(k)}(\Psi)$ be the $k$th moment of the unnormalized projected ensemble generated by the bipartite pure state $\ket{\Psi}_{AB}$,
\begin{equation}
    \tilde{\rho}_{\mathcal{E}}^{(k)}(\Psi) = \sum_{z=1}^{D_B} q_z^{1-k} [(I_A \otimes \bra{z})\ket{\Psi}\bra{\Psi} (I_A \otimes \ket{z})]^{\otimes k},
\end{equation}
for an arbitrary orthonormal basis $\{\ket{z}\}$ on subsystem $B$, and $q_z \geq 0$ are arbitrary non-negative weights independent of $\ket{\Psi}_{AB}$. Define
\begin{equation}
    \Delta = \abs{\Eset{\Psi \sim \mathcal{E}_1} \norm{\tilde{\rho}_{\mathcal{E}}^{(k)}(\Psi) - M}_2^2 - \Eset{\Psi \sim \mathcal{E}_2} \norm{\tilde{\rho}_{\mathcal{E}}^{(k)}(\Psi) - M}_2^2}
\end{equation}
for some arbitrary fixed positive operator $M$. Then, $\Delta$ satisfies the inequality
\begin{equation}
    \Delta \leq \epsilon \Eset{\Psi \sim \mathcal{E}_1} \sparens{2\Tr\parens{ \tilde{\rho}_{\mathcal{E}}^{(k)}(\Psi) M} + \parens{\sum_{z=1}^{D_B} q_z^{1-k} \braket{\Psi|(I_A \otimes \ket{z}\bra{z})|\Psi}^k}^2}.
\end{equation}
\end{lemma}
\begin{proof}
Expanding the definition of $\Delta$ gives
\begin{equation}
\begin{aligned}
    \Delta &= \abs{\Eset{\Psi \sim \mathcal{E}_1} \Tr \parens{{\tilde{\rho}_{\mathcal{E}}^{(k)}(\Psi)}^2} - \Eset{\Psi \sim \mathcal{E}_2} \Tr \parens{{\tilde{\rho}_{\mathcal{E}}^{(k)}(\Psi)}^2} - 2 \Eset{\Psi \sim \mathcal{E}_1} \Tr \parens{\tilde{\rho}_{\mathcal{E}}^{(k)}(\Psi) M} + 2 \Eset{\Psi \sim \mathcal{E}_2} \Tr \parens{\tilde{\rho}_{\mathcal{E}}^{(k)}(\Psi) M}}\\
&\leq 2 \abs{\Eset{\Psi \sim \mathcal{E}_1} \Tr \parens{\tilde{\rho}_{\mathcal{E}}^{(k)}(\Psi) M} - \Eset{\Psi \sim \mathcal{E}_2} \Tr \parens{\tilde{\rho}_{\mathcal{E}}^{(k)}(\Psi) M}} + \abs{\Eset{\Psi \sim \mathcal{E}_1} \Tr \parens{{\tilde{\rho}_{\mathcal{E}}^{(k)}(\Psi)}^2} - \Eset{\Psi \sim \mathcal{E}_2} \Tr \parens{{\tilde{\rho}_{\mathcal{E}}^{(k)}(\Psi)}^2}}.
\end{aligned}
\end{equation}
Since $\tilde{\rho}_{\mathcal{E}}^{(k)}(\Psi)$ can be regarded as a positive linear map acting on $\parens{\ket{\Psi}\bra{\Psi}}^{\otimes k}$, the positive semidefinite ordering~\eqref{eq:relative_error_defn} is preserved, i.e.,
\begin{equation}
    (1-\epsilon) \Eset{\Psi \sim \mathcal{E}_1} \tilde{\rho}_{\mathcal{E}}^{(k)}(\Psi) \preceq \Eset{\Psi \sim \mathcal{E}_2} \tilde{\rho}_{\mathcal{E}}^{(k)}(\Psi) \preceq (1+\epsilon) \Eset{\Psi \sim \mathcal{E}_1} \tilde{\rho}_{\mathcal{E}}^{(k)}(\Psi).
\end{equation}
For a positive operator $M$, this implies
\begin{equation}
    (1-\epsilon) \Eset{\Psi \sim \mathcal{E}_1} \Tr \parens{\tilde{\rho}_{\mathcal{E}}^{(k)}(\Psi) M} \preceq \Eset{\Psi \sim \mathcal{E}_2} \Tr \parens{\tilde{\rho}_{\mathcal{E}}^{(k)}(\Psi) M} \preceq (1+\epsilon) \Eset{\Psi \sim \mathcal{E}_1} \Tr \parens{\tilde{\rho}_{\mathcal{E}}^{(k)}(\Psi) M}.
\end{equation}
Thus,
\begin{equation}
    \abs{\Eset{\Psi \sim \mathcal{E}_1} \Tr \parens{\tilde{\rho}_{\mathcal{E}}^{(k)}(\Psi) M} - \Eset{\Psi \sim \mathcal{E}_2} \Tr \parens{\tilde{\rho}_{\mathcal{E}}^{(k)}(\Psi) M}} \leq \epsilon \Eset{\Psi \sim \mathcal{E}_1} \Tr \parens{\tilde{\rho}_{\mathcal{E}}^{(k)}(\Psi) M}.
\label{eq:aux1}
\end{equation}
Next, we write
\begin{equation}
    \Tr \parens{{\tilde{\rho}_{\mathcal{E}}^{(k)}(\Psi)}^2} = \Tr \sparens{\mathcal{N}(\ket{\Psi}\bra{\Psi}^{\otimes 2k}) \hat{\tau}_A},
\end{equation}
where
\begin{equation}
    \mathcal{N}(X) = \sum_{z,z^\prime=1}^{D_B} (q_z q_{z^\prime})^{1-k} (I_A^{\otimes 2k} \otimes \bra{z^{\otimes k} z^{\prime \otimes k}}) X (I_A \otimes \ket{z^{\otimes k} z^{\prime \otimes k}})
\end{equation}
is a positive linear map, and $\hat{\tau}_A$ is a permutation which swaps between the first and second $k$ tensor factors of subsystem $A$, and acts identically on subsystem $B$. Therefore,
\begin{equation}
\begin{aligned}
    \abs{\Eset{\Psi \sim \mathcal{E}_1} \Tr \parens{{\tilde{\rho}_{\mathcal{E}}^{(k)}(\Psi)}^2} - \Eset{\Psi \sim \mathcal{E}_2} \Tr \parens{{\tilde{\rho}_{\mathcal{E}}^{(k)}(\Psi)}^2}} &\leq \epsilon \Tr \mathcal{N}\parens{\Eset{\Psi \sim \mathcal{E}_1} \parens{\ket{\Psi}\bra{\Psi}}^{\otimes \change{2}k}} \\
    &= \epsilon \Eset{\Psi \sim \mathcal{E}_1} \sum_{z,z^\prime} (q_z q_{z^\prime})^{1-k} \braket{\Psi|(I_A \otimes \ket{z}\bra{z})|\Psi}^{k} \braket{\Psi|(I_A \otimes \ket{z^\prime}\bra{z^\prime})|\Psi}^{k} \\
    &= \epsilon \parens{\sum_{z=1}^{D_B} q_z^{1-k} \braket{\Psi|(I_A \otimes \ket{z}\bra{z})|\Psi}^k}^2\,.
\end{aligned}
\label{eq:aux2}
\end{equation}
Combining~\eqref{eq:aux1} and~\eqref{eq:aux2} gives the desired bound for $\Delta$.
\end{proof}

\section{Scrooge approximation lemma}
\label{app:scrooge_approximation}
Let us define the ensemble of unnormalized states
\begin{equation}
    \tilde{\text{S}}\text{crooge}(\sigma) = \bparens{d\phi, \sqrt{D\sigma}\ket{\phi}},
\label{eq:scrooge_proxy_ensemble}
\end{equation}
where $d\phi$ is the Haar measure on the unit sphere in $\mathbb{C}^D$ and $\sigma$ is an arbitrary density matrix with dimension $D$. Note that $\tilde{\text{S}}$crooge($\sigma$) is referred to as $\tilde{\mathcal{E}}$ in the main text. The $k$th moment of $\tilde{\text{S}}\text{crooge}(\sigma)$ is given by
\begin{equation}
    \tilde{\rho}_{\text{Scrooge}}^{(k)}(\sigma) = \sqrt{D\sigma}^{\otimes k} \rho_{\text{Haar}}^{(k)} \sqrt{D\sigma}^{\otimes k} = (D\sigma)^{\otimes k} \rho_{\text{Haar}}^{(k)},
\end{equation}
where in the final equality we used the fact that $\rho_{\text{Haar}}^{(k)}$ commutes with $\sigma^{\otimes k}$. While the ensemble $\tilde{\text{S}}\text{crooge}(\sigma)$ describes unnormalized states and is unphysical, it is a useful mathematical object in proving the main results of this paper. First, we note that in the low-purity regime, the states in $\tilde{\text{S}}\text{crooge}(\sigma)$ are close to normalized.
\begin{lemma}[Moment bounds]\label{lemma:mom_bound} Let $\sigma$ be a positive operator with unit trace. Denote the overlap between $\sigma$ and $\ket{\phi}$ by $p_\sigma(\phi) \equiv \braket{\phi|\sigma|\phi}$, where $\ket{\phi}$ is a Haar random state. Assuming $\binom{k}{2} \norm{\sigma}_2 < 1$, the $k$th moment of $p_\sigma(\phi)$ satisfies the inequality
\begin{equation}
    1 + \binom{k}{2} \norm{\sigma}_2^2 \leq k! D_k\Eset{\phi \sim \text{Haar}(D)}\sparens{p_\sigma(\phi)^k} \leq 1 + \frac{\binom{k}{2}\norm{\sigma}_2^2}{1-\binom{k}{2}\norm{\sigma}_2}.
\end{equation}
\label{lemma:lowpurity_bounds}
\end{lemma}
\begin{proof}
Using Eq.~\eqref{eq:rho_haar_k}, we can write the $k$th moment of $p_\sigma(\phi)$ as
\begin{equation}
    \Eset{\phi \sim \text{Haar}(D)}\sparens{p_\sigma(\phi)^k} = \Tr\parens{\sigma^{\otimes k} \rho_{\text{Haar}}^{(k)}} = \frac{1}{k!D_k} \sum_{\pi \in S_k} \Tr\parens{\sigma^{\otimes k} \hat{\pi}}.
\end{equation}
Thus, the quantity we would like to bound is
\begin{equation}
    k! D_k \Eset{\phi \sim \text{Haar}(D)}\sparens{p_\sigma(\phi)^k} = \sum_{\pi \in S_k} \Tr\parens{\sigma^{\otimes k} \hat{\pi}}.
\end{equation}
First, note that if $\pi$ is the identity permutation, $\Tr\parens{\sigma^{\otimes k} \hat{\pi}} = [\Tr\parens{\sigma}]^{k} = 1$, since $\sigma$ has unit trace. Next, consider the case where $\pi$ contains a single transposition, i.e., $\pi$ has a Cayley distance of $1$. There are $\binom{k}{2}$ such permutations, each contributing $\Tr\parens{\sigma^{\otimes k} \hat{\pi}} = \Tr\parens{\sigma^2} = \norm{\sigma}_2^2$. All the remaining terms omitted are positive. Thus, we have the lower bound
\begin{equation}
    \sum_{\pi \in S_k} \Tr\parens{\sigma^{\otimes k} \hat{\pi}} \geq 1 + \binom{k}{2} \norm{\sigma}_2^2.
\end{equation}
To prove the upper bound, let us examine the contribution from all permutations with a fixed Cayley distance $l \geq 1$. There are at most $\binom{k}{2}^l$ such permutations, and the contribution of each permutation can be upper bounded by $\Tr\parens{\sigma^{\otimes k} \hat{\pi}} \leq \Tr\parens{\sigma^{l+1}}$. This bound is attained if $\pi$ contains a cycle of length $l+1$. Thus, we have
\begin{equation}
    \sum_{\pi \in S_k} \Tr\parens{\sigma^{\otimes k} \hat{\pi}} \leq 1 + \binom{k}{2} \norm{\sigma}_2^2 + \sum_{l=2}^{k-1} \binom{k}{2}^l \Tr\parens{\sigma^{l+1}}.
\end{equation}
Using the inequality
\begin{equation}
    \Tr\parens{\sigma^{l+1}} \leq \sparens{\Tr\parens{\sigma^2}}^{\frac{l+1}{2}} = \norm{\sigma}_2^{l+1},
\end{equation}
we obtain
\begin{equation}
\begin{aligned}
    \sum_{\pi \in S_k} \Tr\parens{\sigma^{\otimes k} \hat{\pi}} &\leq 1 + \norm{\sigma}_2 \sum_{l=1}^{k-1} \sparens{\binom{k}{2} \norm{\sigma}_2}^l \\
    &\leq 1 + \norm{\sigma}_2 \sum_{l=1}^{\infty} \sparens{\binom{k}{2} \norm{\sigma}_2}^l \\
    &= 1 + \frac{\binom{k}{2}\norm{\sigma}_2^2}{1 - \binom{k}{2}\norm{\sigma}_2},
\end{aligned}
\end{equation}
where the infinite series converges because we assumed $\binom{k}{2} \norm{\sigma}_2 < 1$.
\end{proof}
The above lemma tells us that, up to corrections proportional to the purity of $\sigma$,
\begin{equation}
    \norm{\sqrt{D \sigma} \ket{\phi}}^2 = D \braket{\phi|\sigma|\phi} \approx 1,
\end{equation}
where $\ket{\phi} \sim \text{Haar}(D)$.

The negative moments of $\braket{\phi|\sigma|\phi}$ can also be upper bounded in terms of $\norm{\sigma}_\infty$, as shown in Ref.~\cite{mcginley2025scrooge}.
\begin{lemma}[Lemma 2,~\cite{mcginley2025scrooge}]
Let $\sigma$ be an arbitrary density operator with $m = \lfloor \norm{\sigma}_\infty^{-1} \rfloor$. For $0 \leq q < m$,
\begin{equation}
    \Eset{\phi \sim \text{Haar}(D)} \braket{\phi|D\sigma|\phi}^{-q} \leq \exp\parens{\frac{q^2}{2(m-q)}}. %
\end{equation}
For $q \ll m$, this simplifies to
\begin{equation}
    \Eset{\phi \sim \text{Haar}(D)} \braket{\phi|D\sigma|\phi}^{-q} \leq 1 + \bigO{\frac{q^2}{m}}.
\end{equation}
\label{lemma:neg_mom_bounds}
\end{lemma}

We now prove Lemma~\ref{lemma:scrooge_approx} from the main text, which we reproduce here, stated more formally.
\begin{lemma}[Scrooge approximation] 
\label{lemma:app_scrooge_approx}
Consider the ensemble of unnormalized states $\tilde{\text{S}}\text{crooge}(\sigma) = \{\sqrt{D\sigma} \ket{\phi}\}$, where $\ket{\phi} \sim \text{Haar}(D)$, and $\sigma$ is an arbitrary density matrix with dimension $D$. The $k$th moment of $\tilde{\text{S}}\text{crooge}(\sigma)$ is $\tilde{\rho}^{(k)}_{\text{Scrooge}}(\sigma) =  (D\sigma)^{\otimes k} \rho_{\text{Haar}}^{(k)}$. For any $k^2 \norm{\sigma}_2 \ll 1$, $\tilde{\text{S}}\text{crooge}(\sigma)$ forms a Scrooge$(\sigma)$ $k$-design with additive error $\bigO{k \norm{\sigma}_2}$, i.e.,
\begin{equation}
\norm{\rho_{\text{Scrooge}}^{(k)}(\sigma) - \tilde{\rho}_{\text{Scrooge}}^{(k)}(\sigma)}_1 \leq  \bigO{k \norm{\sigma}_2},    
\end{equation}
and a Scrooge$(\sigma)$ $k$-design with relative error $\bigO{4^k k \norm{\sigma}_2}$, i.e.,
\begin{equation}
    (1-\epsilon) \tilde{\rho}_{\text{Scrooge}}^{(k)}(\sigma) \preceq \rho_{\text{Scrooge}}^{(k)}(\sigma) \preceq (1+\epsilon)\tilde{\rho}_{\text{Scrooge}}^{(k)}(\sigma),
\end{equation}
where $\epsilon = \bigO{4^k k \norm{\sigma}_2}$.
\end{lemma}
\begin{proof}[Proof of additive error $\bigO{k\norm{\sigma}_2}$]
\begin{equation}
\begin{aligned}
    \norm{\rho_{\text{Scrooge}}^{(k)}(\sigma) - \tilde{\rho}_{\text{Scrooge}}^{(k)}(\sigma)}_1 &= D \norm{\Eset{\phi \sim \text{Haar}(D)} \parens{\sqrt{\sigma}\ket{\phi}\bra{\phi}\sqrt{\sigma}}^{\otimes k} \parens{\frac{1}{\braket{\phi|\sigma|\phi}^{k-1}} - D^{k-1}}}_1 \\
    &\leq D \Eset{\phi \sim \text{Haar}(D)} \braket{\phi|\sigma|\phi}^{k} \abs{\frac{1}{\braket{\phi|\sigma|\phi}^{k-1}} - D^{k-1}},
\end{aligned}
\end{equation}
where we used
\begin{equation}
    \tilde{\rho}_{\text{Scrooge}}^{(k)}(\sigma) = D^k \sqrt{\sigma}^{\otimes k} \rho_{\text{Haar}}^{(k)} \sqrt{\sigma}^{\otimes k} = D^k \Eset{\phi \sim \text{Haar}(D)} \parens{\sqrt{\sigma}\ket{\phi}\bra{\phi}\sqrt{\sigma}}^{\otimes k}
\end{equation}
in the first line. Defining the random variable
\begin{equation}
    X = D\braket{\phi|\sigma|\phi},
\end{equation}
we have
\begin{equation}
\begin{aligned}
\norm{\rho_{\text{Scrooge}}^{(k)}(\sigma) - \tilde{\rho}_{\text{Scrooge}}^{(k)}(\sigma)}_1 &\leq \Eset{\phi \sim \text{Haar}(D)} |X - X^{k}| \\
&= \Eset{\phi \sim \text{Haar}(D)} \abs{(1-X)(X+X^2 + \ldots + X^{k-1})} \\
&\leq \sparens{\Eset{\phi \sim \text{Haar}(D)} \parens{1-X}^2}^{1/2} \sparens{\Eset{\phi \sim \text{Haar}(D)}(X+X^2 + \ldots + X^{k-1})^2}^{1/2}.
\end{aligned}
\end{equation}
Explicitly evaluating the Haar averages gives
\begin{equation}
    \Eset{\phi \sim \text{Haar}(D)} X = 1
\end{equation}
and
\begin{equation}
    \Eset{\phi \sim \text{Haar}(D)} X^2 = \frac{D}{D+1} \parens{1 + \norm{\sigma}_2^2},
\end{equation}
which yields
\begin{equation}
    \Eset{\phi \sim \text{Haar}(D)}\parens{1-X}^2 = \frac{D\norm{\sigma}_2^2 - 1}{D+1}.
\end{equation}
For $k^2 \norm{\sigma}_2 \ll 1$, Lemma~\ref{lemma:lowpurity_bounds} implies that
\begin{equation}
    \Eset{\phi \sim \text{Haar}(D)} X^k = 1 + \bigO{k^2 \norm{\sigma}_2^2}, 
\end{equation}
which leads to
\begin{equation}
    \Eset{\phi \sim \text{Haar}(D)}(X+X^2 + \ldots + X^{k-1})^2 \leq (k-1)^2 \parens{1 + \bigO{k^2 \norm{\sigma}_2^2}},
\end{equation}
since the left hand side can be expanded as a sum of $(k-1)^2$ terms, each contributing $1 + \bigO{k^2 \norm{\sigma}_2^2}$. Therefore,
\begin{equation}\norm{\rho_{\text{Scrooge}}^{(k)}(\sigma) - \tilde{\rho}_{\text{Scrooge}}^{(k)}(\sigma)}_1 \leq (k-1) \sqrt{\frac{D \norm{\sigma}_2^2 - 1}{D+1}}\parens{1 + \bigO{k^2 \norm{\sigma}_2^2}}.
\end{equation}
This vanishes exactly when $k = 1$, or when $\sigma$ is the maximally mixed state (i.e., Haar limit). In general, the error bound is $\bigO{k \norm{\sigma}_2}$.
\end{proof}

\begin{proof}[Proof of relative error $\bigO{4^k k\norm{\sigma}_2}$]
It suffices to show that $\rho_\text{Haar}^{(k)}$ and
\begin{equation}
    A_\text{Scrooge}^{(k)}(\sigma) = \Eset{\phi \sim \text{Haar}(D)} \frac{\parens{\ket{\phi}\bra{\phi}}^{\otimes k}}{\braket{\phi|D\sigma|\phi}^{k-1}}
\end{equation}
are $\epsilon$-close in relative error, since multiplying by $\sqrt{D\sigma}^{\otimes k}$ on both sides yields $\tilde{\rho}_{\text{Scrooge}}^{(k)}(\sigma)$ and ${\rho}_{\text{Scrooge}}^{(k)}(\sigma)$ respectively, which does not increase the relative error. Both $\rho_{\text{Haar}}^{(k)}$ and $A_{\text{Scrooge}}^{(k)}(\sigma)$ are supported on $\mathcal{H}_{\text{sym}}^{(k)}$, the symmetric subspace of $\mathcal{H}^{\otimes k}$. For any $\ket{x} \in \mathcal{H}_{\text{sym}}^{(k)}$,
\begin{equation}
\begin{aligned}
    \abs{\braket{x|\rho_{\text{Haar}}^{(k)} - A_{\text{Scrooge}}^{(k)}(\sigma)|x}} &= \abs{\Eset{\phi \sim \text{Haar}(D)} \braket{x|\parens{\ket{\phi}\bra{\phi}}^{\otimes k}|x}\parens{\frac{1}{\braket{\phi|D\sigma|\phi}^{k-1}} - 1}} \\
    &= \abs{\Eset{\phi \sim \text{Haar}} \braket{x|\parens{\ket{\phi}\bra{\phi}}^{\otimes k}|x} \frac{1}{\braket{\phi|D\sigma|\phi}^{k-1}}\parens{1 - \braket{\phi|D\sigma|\phi}^{k-1}}} \\
    &\leq \Eset{\phi \sim \text{Haar}(D)} \braket{x|\parens{\ket{\phi}\bra{\phi}}^{\otimes k}|x} \frac{1}{\braket{\phi|D\sigma|\phi}^{k-1}}\abs{1 - \braket{\phi|D\sigma|\phi}^{k-1}}.
\end{aligned}    
\end{equation}
Applying H\"older's inequality,
\begin{equation}
\begin{aligned}
    \abs{\braket{x|\rho_{\text{Haar}}^{(k)} - A_{\text{Scrooge}}^{(k)}(\sigma)|x}} \leq \parens{\E_\phi \braket{x|\parens{\ket{\phi}\bra{\phi}}^{\otimes k}|x}^4}^{1/4} \parens{\E_\phi \braket{\phi|D\sigma|\phi}^{4(1-k)}}^{1/4} \sparens{\E_\phi \parens{1 - \braket{\phi|D\sigma|\phi}^{k-1}}^2}^{1/2}.
\end{aligned}   
\end{equation}
Now,
\begin{equation}
\begin{aligned}
    \Eset{\phi \sim \text{Haar}(D)} \braket{x|\parens{\ket{\phi}\bra{\phi}}^{\otimes k}|x}^4 = \braket{x^{\otimes 4}|\rho_{\text{Haar}}^{(4k)}|x^{\otimes 4}} \leq \norm{\rho_{\text{Haar}}^{(4k)}}_\infty = \frac{1}{D_{4k}}.
\end{aligned}
\end{equation}
Next, using Lemma~\ref{lemma:neg_mom_bounds}~\cite{mcginley2025scrooge}, for $k \ll m$,
\begin{equation}
    \Eset{\phi \sim \text{Haar}(D)}\braket{\phi|D\sigma|\phi}^{4(1-k)} \leq 1 + \bigO{\frac{k^2}{m}},
\end{equation}
where $m = \lfloor \norm{\sigma}_\infty^{-1} \rfloor$. 
Finally, for $k^2 \norm{\sigma}_2 \ll 1$, we have
\begin{equation}
    \E_{\phi} \parens{1 - \braket{\phi|D\sigma|\phi}^{k-1}}^2 = \bigO{k^2 \norm{\sigma}_2^2}.
\end{equation}
Note that $k^2 \norm{\sigma}_2 \ll 1$ implies $k \ll m$, since $\norm{\sigma}_\infty^2 \leq \norm{\sigma}_2$. Combining these bounds, and using the identity $1 = D_k \braket{x|\rho_{\text{Haar}}^{(k)}|x}$, we find
\begin{equation}
    \abs{\braket{x|\rho_{\text{Haar}}^{(k)} - A_{\text{Scrooge}}^{(k)}|x}} \leq \bigO{\frac{k \norm{\sigma}_2}{D_{4k}^{1/4}}} = \bigO{\frac{k \norm{\sigma}_2 D_k}{D_{4k}^{1/4}}} \braket{x|\rho_{\text{Haar}}^{(k)}|x} = \bigO{4^k k \norm{\sigma}_2}\braket{x|\rho_{\text{Haar}}^{(k)}|x}.
\end{equation}
Thus, the relative error is $\bigO{4^k k \norm{\sigma}_2}$.
\end{proof}
We speculate that the prefactor of $4^k$ in the relative error can be improved. An analogous relative-error bound $\bigO{k^2 \norm{\sigma}_\infty^{1/2}}$, which is controlled by $\norm{\sigma}_\infty$ instead of $\norm{\sigma}_2$, was proven in Ref.~\cite{mcginley2025scrooge}.
Our relative error in Lemma~\ref{lemma:scrooge_approx} has a better scaling with the norm of $\sigma$, since $\norm{\sigma}_2 \leq \norm{\sigma}_\infty^{1/2}$, but has a worse scaling in $k$. In essence, the Scrooge approximation lemma (Lemma~\ref{lemma:scrooge_approx}) tells us that
\begin{equation}
    \rho_{\text{Scrooge}}^{(k)}(\sigma) \approx (D\sigma)^{\otimes k} \rho_{\text{Haar}}^{(k)},
\end{equation}
for arbitrary low-purity density matrix $\sigma$. Crucially, the right hand side involves only the $k$th moment of the Haar ensemble, while $\rho_{\text{Scrooge}}^{(k)}(\sigma)$ is a rational function of Haar random states. This enables us to analyze the scenarios where the Haar ensemble is replaced by a $k$-design, with the rigorous guarantee given by Lemma~\ref{lemma:scrooge_approx}. We will make extensive use of this in the following Appendices to prove the main theorems in the paper.

\section{Examples of Scrooge $k$-designs}
\label{app:global_scrooge}
In the main text, we assert that the random phase ensemble,
\begin{equation}\label{eq:app_randomphase}
    \mathcal{E}_{\text{Random Phase}} = \bparens{\sum_{j=1}^{D} |\braket{E_j|\Psi_0}|e^{i\varphi_j} \ket{E_j}}_{\varphi},
\end{equation}
and canonical thermal pure quantum (cTPQ) states,
\begin{equation}
\label{eq:app_ctpq}
    \mathcal{E}_{\text{cTPQ}} = \bparens{\frac{1}{\mathcal{N}}\sum_{j=1}^D \xi_j e^{-\beta H/2} \ket{j}}_\xi,
\end{equation}
form approximate Scrooge$(\sigma)$ designs, with $\sigma$ given by the diagonal state $\sigma_{\text{diag}}$ and thermal Gibbs state $\sigma_\beta \propto \exp(-\beta H)$, respectively. These results can be derived using Lemma~\ref{lemma:scrooge_approx} by relating these ensembles to the proxy $\tilde{\text{S}}$crooge($\sigma$), which we will do so in this section. 
\subsection{Random phase ensemble}
Here, we prove Theorem~\ref{thm:global_scrooge} in the main text. Since the temporal ensemble obtained by late-time Hamiltonian dynamics is described by the random phase ensemble $\mathcal{E}_{\text{Random Phase}}$, assuming the Hamiltonian $H$ satisfies the $k$th no-resonance condition~\cite{mark2024maximum}, it suffices to show that the random phase ensemble indeed forms a Scrooge $k$-design.

\begin{theorem}[Random phase ensembles form Scrooge $k$-designs]
    Let $\mathcal{E}_{\text{Random Phase}}$ be the random phase ensemble defined in Eq.~\eqref{eq:app_randomphase}. $\mathcal{E}_{\text{Random Phase}}$ forms a Scrooge$(\sigma_\text{diag})$ $k$-design with additive error
    \change{
    \begin{equation}
        \epsilon = \bigO{k \norm{\sigma_{\text{diag}}}_2},
    \end{equation}}
    for $k^2 \norm{\sigma_\text{diag}}_2\ll 1$, where $\sigma_\text{diag} = \sum_{j} |\braket{E_j|\Psi_0}|^2 \ket{E_j}\bra{E_j}$ is the diagonal state in the energy eigenbasis.
\end{theorem}
\change{
\begin{proof}
Let $r_i \equiv |\braket{E_i|\Psi_0}|^2$, so that $\sigma_\text{diag} = \sum_i r_i \ket{E_i}\bra{E_i}$. Let $\{\ket{n}\}$ be the normalized occupation number basis with respect to the energy eigenbasis $\{\ket{E_i}\}$, i.e., $\ket{n} = \ket{n_1, n_2, \ldots, n_D}$ with $n_i$ denoting the number of copies of $\ket{E_i}$ in the symmetric state $\ket{n}$. In this basis, the random phase ensemble has the $k$th moment
\begin{equation}
    \Eset{\psi \sim {\text{Random Phase}}} \parens{\ket{\psi}\bra{\psi}}^{\otimes k} = \sum_{|n|=k} k! \parens{\prod_{i=1}^{D} \frac{r_i^{n_i}}{n_i!}} \ket{n}\bra{n},
\end{equation}
where the sum runs over all occupation number vectors $n = (n_1, n_2, \ldots, n_D)$ such that $\sum_{i=1}^{D} n_i = k$. Let us denote
\begin{equation}
    G^{(k)}(\sigma_{\text{diag}}) = \sigma_{\text{diag}}^{\otimes k} \sum_{\pi \in S_k} \hat{\pi} = \sum_{|n|=k} k! \parens{\prod_{i=1}^{D} r_i^{n_i}} \ket{n}\bra{n}.
\end{equation}
Thus,
\begin{equation}
    \norm{\Eset{\psi \sim {\text{Random Phase}}} \parens{\ket{\psi}\bra{\psi}}^{\otimes k} - G^{(k)}(\sigma_{\text{diag}})}_1 = \sum_{|n|=k} k! \parens{\prod_{i=1}^{D} \frac{r_i^{n_i}}{n_i!}} \parens{\prod_{i=1}^{D} n_i! - 1}.
\end{equation}
The sum over $n$ can be split into two contributions. The first contribution comes from the `collision-free' sector, where $n_i \in \{0,1\}$ for all $i$. This contribution is zero, since $n_i! = 1$. The second contribution comes from the `collision' sector, where there exists at least one $i$ such that $n_i \geq 2$. Therefore,
\begin{equation}
\begin{aligned}
    \norm{\Eset{\psi \sim {\text{Random Phase}}} \parens{\ket{\psi}\bra{\psi}}^{\otimes k} - G^{(k)}(\sigma_{\text{diag}})}_1 &\leq \sum_{\substack{|n|=k \\ \text{collision}}} k! \prod_{i=1}^{D} r_i^{n_i} \\
    &= \Tr \sparens{G^{(k)}(\sigma_{\text{diag}})} - \sum_{\substack{|n|=k \\ \text{collision-free}}} k! \prod_{i=1}^{D} r_i^{n_i}.
\end{aligned}
\end{equation}
Consider $k$ independent draws from the distribution $\{r_i\}$. The probability of a collision (i.e., at least two draws with the same outcome) is given by (from the birthday paradox)
\begin{equation}
    \text{Pr}(\text{collision}) = 1 - \sum_{\substack{|n|=k \\ \text{collision-free}}} k! \prod_{i=1}^{D} r_i^{n_i} \leq \binom{k}{2} \sum_{i=1}^{D} r_i^2 = \bigO{k^2 \norm{\sigma_{\text{diag}}}_2^2}.
\end{equation} 
This yields
\begin{equation}
    \norm{\Eset{\psi \sim {\text{Random Phase}}} \parens{\ket{\psi}\bra{\psi}}^{\otimes k} - G^{(k)}(\sigma_{\text{diag}})}_1 \leq \Tr \sparens{G^{(k)}(\sigma_{\text{diag}})} - 1 + \bigO{k^2 \norm{\sigma_{\text{diag}}}_2^2} = \bigO{k^2 \norm{\sigma_{\text{diag}}}_2^2},
\end{equation}
where in the last line we used the fact that $\Tr \sparens{G^{(k)}(\sigma_{\text{diag}})} = 1 + \bigO{k^2 \norm{\sigma_{\text{diag}}}_2^2}$ in the regime $k^2 \norm{\sigma_{\text{diag}}}_2 \ll 1$, which can be shown using Lemma~\ref{lemma:lowpurity_bounds}. Finally, using
\begin{equation}
    \tilde{\rho}_{\text{Scrooge}}^{(k)}(\sigma_{\text{diag}}) = \frac{D^k}{k! D_k} G^{(k)}(\sigma_{\text{diag}}) = \sparens{1 + \bigO{\frac{k^2}{D}}} G^{(k)}(\sigma_{\text{diag}}).
\end{equation}
and Lemma~\ref{lemma:scrooge_approx}, together with the triangle inequality, we obtain
\begin{equation}
    \norm{\Eset{\psi \sim {\text{Random Phase}}} \parens{\ket{\psi}\bra{\psi}}^{\otimes k} - \rho_{\text{Scrooge}}^{(k)}(\sigma_{\text{diag}})}_1 \leq \bigO{\frac{k^2}{D}} + \bigO{k^2 \norm{\sigma_{\text{diag}}}_2^2} + \bigO{k \norm{\sigma_{\text{diag}}}_2} = \bigO{k \norm{\sigma_{\text{diag}}}_2}
\end{equation}
which concludes the proof.
\end{proof}
}
\subsection{Canonical thermal pure quantum (cTPQ) states}
Next, we show that the cTPQ ensemble forms a Scrooge $k$-design in relative error, with respect to the density matrix
\begin{equation}
    \sigma_\beta = \frac{e^{-\beta H}}{\Tr \parens{e^{-\beta H}}}
\end{equation}
for an arbitrary Hamiltonian $H$.
\begin{theorem}[cTPQ states form Scrooge $k$-designs]
    Let $\mathcal{E}_{\text{cTPQ}}$ be the ensemble of canonical thermal pure quantum states defined in Eq.~\eqref{eq:app_ctpq}. For $k^2 \norm{\sigma_\beta}_2 \ll 1$, $\mathcal{E}_{\text{cTPQ}}$ forms a Scrooge$(\sigma_\beta)$ $k$-design with relative error $\epsilon$ satisfying
    \begin{equation}
        1+\epsilon = \sparens{1+\bigO{4^k k \norm{\sigma_\beta}_2}}^2,
    \end{equation} 
    where $\sigma_\beta = \exp(-\beta H)/\Tr\sparens{\exp(-\beta H)}$. For a fixed $k \in \mathbb{N}$, $\epsilon = \bigO{4^k k \norm{\sigma_\beta}_2}$.
\end{theorem}
\begin{proof}
Define $\ket{g} = \sum_{j=1}^{D} \xi_j \ket{j}$, which is a random (unnormalized) vector of $D$ independent zero-mean complex Gaussian variables with unit variance. Writing
\begin{equation}
    \ket{g} = \norm{\ket{g}} \times \ket{\phi},
\end{equation}
each Gaussian vector $\ket{g}$ is associated with a Haar random state $\ket{\phi}$. Note that $\norm{\ket{g}}$ and $\ket{\phi}$ are statistically independent. Thus, we can write the $k$th moment of $\mathcal{E}_{\text{cTPQ}}$ as
\begin{equation}
    \Eset{\psi \sim \text{cTPQ}} \parens{\ket{\psi}\bra{\psi}}^{\otimes k} = \Eset{\phi \sim \text{Haar}(D)} \frac{\parens{\sqrt{\sigma_\beta}\ket{\phi}\bra{\phi}\sqrt{\sigma_\beta}}^{\otimes k}}{\braket{\phi|\sigma_\beta|\phi}^{k}}.
\end{equation}
For any $\ket{x} \in \mathcal{H}_{\text{sym}}^{(k)}$,
\begin{equation}
\begin{aligned}
    &\abs{\braket{x|\parens{\Eset{\phi \sim \text{Haar}(D)}\frac{\ket{\phi}\bra{\phi}^{\otimes k}}{\braket{\phi|D\sigma_\beta|\phi}^k} - \Eset{\phi \sim \text{Haar}(D)} \ket{\phi}\bra{\phi}^{\otimes k}}|x}} \\= &\abs{\Eset{\phi \sim \text{Haar}(D)}\braket{x|\parens{\ket{\phi}\bra{\phi}}^{\otimes k}|x} \braket{\phi|D\sigma_\beta|\phi}^{-k}\parens{1 - \braket{\phi|D\sigma_\beta|\phi}^k}} \\
    \leq &\parens{\E_\phi \braket{x|\parens{\ket{\phi}\bra{\phi}}^{\otimes k}|x}^4}^{1/4} \parens{\E_\phi \braket{\phi|D\sigma_\beta|\phi}^{-4k}}^{1/4} \sparens{\E_\phi \parens{1 - \braket{\phi|D\sigma_\beta|\phi}^k}^2}^{1/2},
\end{aligned}
\end{equation}
where we used H\"older's inequality in the final line. Now,
\begin{equation}
\begin{aligned}
    \Eset{\phi \sim \text{Haar}(D)} \braket{x|\parens{\ket{\phi}\bra{\phi}}^{\otimes k}|x}^4 = \braket{x^{\otimes 4}|\rho_{\text{Haar}}^{(4k)}|x^{\otimes 4}} \leq \norm{\rho_{\text{Haar}}^{(4k)}}_\infty = \frac{1}{D_{4k}}.
\end{aligned}
\end{equation}
Next, using Lemma~\ref{lemma:neg_mom_bounds}, for $k \ll m$,
\begin{equation}
    \Eset{\phi \sim \text{Haar}(D)}\braket{\phi|D\sigma_\beta|\phi}^{-4k} \leq 1 + \bigO{\frac{k^2}{m}},
\end{equation}
where $m = \lfloor \norm{\sigma_\beta}_\infty^{-1} \rfloor$. 
Note that our assumption $k^2 \norm{\sigma_\beta}_2 \ll 1$ already implies $k \ll m$, since $\norm{\sigma_\beta}_\infty^2 \leq \norm{\sigma_\beta}_2$. 
Finally,
\begin{equation}
    \E_{\phi} \parens{1 - \braket{\phi|D\sigma_\beta|\phi}^{\change{k}}}^2 = \bigO{k^2 \norm{\sigma_\beta}_2^2}.
\end{equation}
Combining these bounds, and using the identity $1 = D_k \braket{x|\rho_{\text{Haar}}^{(k)}|x}$, we find 
\begin{equation}
\begin{aligned}
    \abs{\braket{x|\parens{\Eset{\phi \sim \text{Haar}(D)}\frac{\ket{\phi}\bra{\phi}^{\otimes k}}{\braket{\phi|D\sigma_\beta|\phi}^k} - \rho_{\text{Haar}}^{(k)}}|x}} \leq \frac{1}{D_{4k}^{1/4}} \parens{1 + \bigO{\frac{k^2}{m}}} \bigO{k \norm{\sigma_\beta}_2} = \bigO{4^k k \norm{\sigma_\beta}_2} \braket{x|\rho_{\text{Haar}}^{(k)}|x}.
\end{aligned}
\end{equation}
Thus,
\begin{equation}
    (1 \change{-} \epsilon^\prime) \rho_{\text{Haar}}^{(k)} \preceq \Eset{\phi \sim \text{Haar}(D)}\frac{\ket{\phi}\bra{\phi}^{\otimes k}}{\braket{\phi|D\sigma_\beta|\phi}^k} \preceq (1 + \epsilon^\prime) \rho_{\text{Haar}}^{(k)},
\end{equation}
where $\epsilon^\prime = \bigO{4^k k \norm{\sigma_\beta}_2}$. Multiplying both sides of the operator by $\sqrt{D\sigma_\beta}^{\otimes k}$ does not change the relative error, and we get
\begin{equation}
    (1 \change{-} \epsilon^\prime) \tilde{\rho}_{\text{Scrooge}}^{(k)}(\sigma_\beta) \preceq \Eset{\psi \sim \text{cTPQ}} \parens{\ket{\psi}\bra{\psi}}^{\otimes k} \preceq (1 + \epsilon^\prime) \tilde{\rho}_{\text{Scrooge}}^{(k)}(\sigma_\beta).
\end{equation}
Recall from Lemma~\ref{lemma:scrooge_approx} that $\tilde{\text{S}}\text{crooge}(\sigma_\beta)$ forms a Scrooge$(\sigma_\beta)$ $k$-design with relative error $\epsilon^{\prime\prime} = \bigO{4^k k \norm{\sigma_\beta}_2}$. Combining the relative errors, we obtain the desired result
\begin{equation}
    (1 \change{-} \epsilon) {\rho}_{\text{Scrooge}}^{(k)}(\sigma_\beta) \preceq \Eset{\psi \sim \text{cTPQ}} \parens{\ket{\psi}\bra{\psi}}^{\otimes k} \preceq (1 + \epsilon) {\rho}_{\text{Scrooge}}^{(k)}(\sigma_\beta),
\end{equation}
where $1 + \epsilon = (1+\epsilon^\prime)(1 + \epsilon^{\prime\prime})$.
\end{proof}
If we further assume the low-purity regime $4^k k \norm{\sigma_\beta}_2 \ll 1$, this tells us that the cTPQ ensemble is a Scrooge$(\sigma_\beta)$ $k$-design with relative error
\begin{equation}
    \epsilon = \bigO{4^k k \norm{\sigma_\beta}_2}.
\end{equation}

\section{Projected ensemble generated by a state drawn from a Scrooge $2k$-design}
\label{app:projens_2kgenerator}
In this section, we prove Theorem~\ref{thm:2kgenerator}, one of the main results of the paper. Theorem~\ref{thm:2kgenerator} says that the projected ensemble generated by a state drawn from a Scrooge $2k$-design approximates a probabilistic mixture of Scrooge $k$-designs (for brevity, we will refer to such a mixture as a generalized Scrooge $k$-design). The proof of Theorem~\ref{thm:2kgenerator} is rather lengthy, and proceeds in several steps, outlined below.
\begin{enumerate}
    \item First, we show that for the unnormalized generator state sampled from $\tilde{\text{S}}\text{crooge}(\sigma)$ in Eq.~\eqref{eq:scrooge_proxy_ensemble}, the projected ensemble is locally close to a generalized Scrooge $k$-design (Lemma~\ref{lemma:projens_approxScrooge}).
    \item Next, with the help of Lemma~\ref{lemma:projens_approxScrooge}, we show that the projected ensemble generated by a global state drawn from an exact Scrooge $2k$-design forms a generalized Scrooge $k$-design (Theorem~\ref{thm:projens_approxScrooge}).
    \item Finally, we relax the assumption that the generator state is drawn from an exact Scrooge($\sigma$) $2k$-design, and instead consider the case where the generator state is drawn from an approximate Scrooge$(\sigma)$ $2k$-design. Using the results of Theorem~\ref{thm:projens_approxScrooge}, we arrive at Theorem~\ref{thm:2kgenerator} in the main text, which we reproduce here (Theorem~\ref{thm:2ktok_scrooge_app}).
\end{enumerate}

We then specialize our results to the infinite-temperature limit where the generator state is drawn from an approximate Haar $2k$-design, giving Corollary~\ref{cor:2ktokdesign} in the main text. In this case, we can obtain an improved bound on the trace distance between the $k$th moments of the projected and Haar ensembles.

\begin{lemma}\label{lemma:projens_approxScrooge}\normalfont[Projected ensemble generated by an unnormalized Scrooge state]
  Let $\ket{\Psi}_{AB}$ be an unnormalized state drawn from the ensemble $\tilde{\text{S}}$crooge($\sigma$). Denote the reduced density matrix of $\sigma$ on $A$ and $B$ by $\sigma_A$ and $\sigma_B$, respectively. Consider the projected ensemble $\mathcal{E}(\Psi)$ obtained by applying projective measurements on $B$ in an arbitrary orthonormal basis $\{\ket{z}\}_{z=1}^{D_B}$. Then, assuming that $k^2 \norm{\hat{\sigma}_{A|z}}_2 \ll 1$ for all $z$, and $1 \ll D_A \leq D_B$,
\begin{equation}
\begin{aligned}
\Eset{\Psi \sim \text{$\tilde{S}$crooge}(\sigma)} \norm{\rho_{\mathcal{E}}^{(k)} - \sum_{z=1}^{D_B}\braket{z|\sigma_B|z}\rho_\text{Scrooge}^{(k)}(\hat{\sigma}_{A|z})}_1  \leq\bigO{\sqrt{\frac{k^{k+2} D_A^{k-1}}{D_B
   } \parens{1 + \change{\frac{4^k}{k^2} {D_A}^2 D_B} \norm{\sigma}_2^2}}} + \sum_{z=1}^{D_B} \braket{z|\sigma_B|z} \bigO{k \norm{\hat{\sigma}_{A|z}}_2}.
\end{aligned}   
\end{equation}
where $\sigma_{A|z} \equiv (I_A \otimes \bra{z})\sigma(I_A \otimes \ket{z})$, and $\hat{\sigma}_{A|z} =  \sigma_{A|z}/\braket{z|\sigma_B|z}$ is the normalized conditional mixed state on $A$.
\end{lemma}
\begin{proof}
For a given generator state $\ket{\Psi}$, the projected ensemble has the $k$th moment
\begin{equation}
    \rho_{\mathcal{E}}^{(k)} = \sum_{z=1}^{D_B} \frac{\parens{\ket{\tilde{\psi}_z}\bra{\tilde{\psi}_z}}^{\otimes k}}{p_z^{k-1}},
\end{equation}
where
\begin{equation}
    \ket{\tilde{\psi}_z} = (I_A \otimes \bra{z}) \ket{\Psi}
\end{equation}
is the unnormalized projected state, and
\begin{equation}
    p_z = \braket{\tilde{\psi}_z|\tilde{\psi}_z} = \bra{\Psi}(I_A \otimes  \ket{z}\bra{z})\ket{\Psi}
\end{equation}
is the measurement outcome probability. Let us construct the proxy
\begin{equation}
    \tilde{\rho}_{\mathcal{E}}^{(k)} = \sum_{z=1}^{D_B} \frac{\parens{\ket{\tilde{\psi}_z}\bra{\tilde{\psi}_z}}^{\otimes k}}{\braket{z|\sigma_B|z}^{k-1}},
\end{equation}
\change{where the sum over $z$ is restricted to measurement outcomes satisfying $\braket{z|\sigma_B|z} > 0$.} Then, we have from Lemma~\ref{lemma:general_approx},
\begin{equation}\label{eq:td_vs_approxprojens}
\begin{aligned}
    \Eset{\Psi \sim \text{$\tilde{S}$crooge}(\sigma)} \norm{\rho_{\mathcal{E}}^{(k)} - \tilde{\rho}_{\mathcal{E}}^{(k)}}_1 
    &\leq \sum_{z=1}^{D_B} \parens{\E_{\Psi} p_z^2}^{1/2} \sparens{1 - 2 \frac{\E_\Psi p_z^{k-1}}{\braket{z|\sigma_B|z}^{k-1}}+ \frac{\E_{\Psi} p_z^{2k-2}}{\braket{z|\sigma_B|z}^{2k-2}}}^{1/2}. 
\end{aligned}
\end{equation}
Now,
\begin{equation}
    \Eset{\Psi \sim \text{$\tilde{S}$crooge}(\sigma)} p_z^k = \frac{D^k}{k! D_k} \sum_{\pi \in S_k} \Tr\parens{\sigma_{A|z}^{\otimes k} \hat{\pi}_A} = \sparens{1 + \bigO{\frac{k^2}{D}}} \braket{z|\sigma_B|z}^k\sum_{\pi \in S_k} \Tr\parens{\hat{\sigma}_{A|z}^{\otimes k} \hat{\pi}_A},
\end{equation}
using the fact that $\Tr \sigma_{A|z} = \braket{z|\sigma_B|z}$, and $\hat{\sigma}_{A|z} = \sigma_{A|z} /\braket{z|\sigma_B|z}$. For $k^2 \norm{\hat{\sigma}_{A|z}}_2 \ll 1$, we can use Lemma~\ref{lemma:lowpurity_bounds} to obtain
\begin{equation}\label{eq:probz_approxScr}
\Eset{\Psi \sim \text{$\tilde{S}$crooge}(\sigma)} p_z^k = \braket{z|\sigma_B|z}^k \sparens{1 + \bigO{k^2 \norm{\hat{\sigma}_{A|z}}_2^2}}.    
\end{equation}
Substituting this into the upper bound above gives
\begin{equation}
\label{eq:expectedTD_bound_term1}
\begin{aligned}
    \Eset{\Psi \sim \text{$\tilde{S}$crooge}(\sigma)} \norm{\rho_{\mathcal{E}}^{(k)} - \tilde{\rho}_{\mathcal{E}}^{(k)}}_1 &\leq \sum_{z=1}^{D_B} \braket{z|\sigma_B|z} \bigO{k \norm{\hat{\sigma}_{A|z}}_2}.
\end{aligned}
\end{equation} 
Next, we seek to bound
\begin{equation}
\label{eq:two_norm_error}
\begin{aligned}
    \Eset{\Psi \sim \text{$\tilde{S}$crooge}(\sigma)} \norm{\tilde{\rho}_{\mathcal{E}}^{(k)} - \Eset{\Psi \sim \text{$\tilde{S}$crooge}(\sigma)}\tilde{\rho}_{\mathcal{E}}^{(k)}}_2^2 = \underbrace{\E_\Psi \Tr \parens{\tilde{\rho}_{\mathcal{E}}^{(k) 2}}}_{(*)} - \underbrace{\Tr \sparens{\parens{E_\Psi \tilde{\rho}_\mathcal{E}^{(k)}}^2}}_{(**)}.
\end{aligned}
\end{equation}
We first evaluate the second term $(**)$. Note that
\begin{equation}
\label{eq:average_approx_projens}
\begin{aligned}
    \Eset{\Psi \sim \text{$\tilde{S}$crooge}(\sigma)} \tilde{\rho}_\mathcal{E}^{(k)} &= D^{k} \sum_{z=1}^{D_B} \braket{z|\sigma_B|z}^{1-k} \Eset{\phi \sim \text{Haar}(D)} \sparens{(I_A \otimes \bra{z}) \sqrt{\sigma} \ket{\phi}\bra{\phi} \sqrt{\sigma} (I_A \otimes \ket{z})}^{\otimes k} \\
    &= \frac{D^{k} D_{A,k}}{D_k} \sum_{z=1}^{D_B} \braket{z|\sigma_B|z}^{1-k} \sigma_{A|z}^{\otimes k} \rho_{\text{Haar,A}}^{(k)} \\
    &= \frac{D^{k} D_{A,k}}{D_k} \sum_{z=1}^{D_B} \braket{z|\sigma_B|z} \hat{\sigma}_{A|z}^{\otimes k} \rho_{\text{Haar,A}}^{(k)} \\
    &= \parens{1 + \bigO{\frac{k^2}{D_A}}} \sum_{z=1}^{D_B} \braket{z|\sigma_B|z} \tilde{\rho}_{\text{Scrooge}}^{(k)}(\hat{\sigma}_{A|z}),
\end{aligned}
\end{equation}
which can be interpreted as a mixture of approximate Scrooge$(\hat{\sigma}_{A|z})$ ensembles, weighted by the average probability $\braket{z|\sigma_B|z}$ of measuring the outcome $z$. Thus,
\begin{equation}\label{eq:starstar1}
\begin{aligned}
    (**) &= \parens{\frac{D^{k} D_{A,k}}{D_k}}^2 \sum_{z,z^\prime = 1}^{D_B} \braket{z|\sigma_B|z}\braket{z^\prime|\sigma_B|z^\prime} \Tr\parens{\hat{\sigma}_{A|z}^{\otimes k} \rho_{\text{Haar},A}^{(k)} \hat{\sigma}_{A|z^\prime}^{\otimes k} \rho_{\text{Haar},A}^{(k)}} \\
    &= \frac{D^{2k} D_{A,k}}{D_k^2}\sum_{z,z^\prime = 1}^{D_B} \braket{z|\sigma_B|z}\braket{z^\prime|\sigma_B|z^\prime} \Tr\parens{\hat{\sigma}_{A|z}^{\otimes k} \hat{\sigma}_{A|z^\prime}^{\otimes k} \rho_{\text{Haar},A}^{(k)}}, 
\end{aligned}
\end{equation}
using the fact that ${\rho_{\text{Haar},A}^{(k) 2}} = \rho_{\text{Haar},A}^{(k)} / D_{A,k}$. 
The first term in Eq.~\eqref{eq:two_norm_error}, denoted $(*)$, evaluates to
\begin{equation}
\begin{aligned}
    (*) &= D^{2k} \sum_{z,z^\prime = 1}^{D_B} \braket{z|\sigma_B|z}^{1-k}\braket{z^\prime|\sigma_B|z^\prime}^{1-k} \Tr \sparens{(I_A^{\otimes 2k} \otimes \ket{z^{\prime k} z^k}\bra{z^k z^{\prime k}}) \sigma^{\otimes 2k} \change{\rho_{\text{Haar}}^{(2k)}} \hat{\tau}}.
\end{aligned}
\end{equation}
Here, we defined the shorthand notation $\ket{z^k z^{\prime k}} \equiv \ket{z}^{\otimes k} \otimes \ket{z^\prime}^{\otimes k}$. $\hat{\tau}$ is a permutation operator acting on the $2k$-fold replica Hilbert space $\mathcal{H}^{\otimes 2k}$, which simply swaps the first and second $k$ replicas, e.g., $\hat{\tau} (\ket{\psi_1}^{\otimes k} \otimes \ket{\psi_2}^{\otimes k}) = \ket{\psi_2}^{\otimes k} \otimes \ket{\psi_1}^{\otimes k}$. Now, expand \change{$\rho_{\text{Haar}}^{(2k)}$} in terms of the permutation operators $\hat{\pi}$, where $\pi$ is an element of the symmetric group $S_{2k}$, to get
\begin{equation}
\begin{aligned}
(*) &= \frac{D^{2k}}{(2k)!D_{2k}}\sum_{z,z^\prime = 1}^{D_B} \sum_{\pi \in S_{2k}} \braket{z|\sigma_B|z}^{1-k}\braket{z^\prime|\sigma_B|z^\prime}^{1-k} \Tr \sparens{(I_A^{\otimes 2k} \otimes \ket{z^{\prime k} z^k}\bra{z^k z^{\prime k}}) \sigma^{\otimes 2k} \hat{\pi} \hat{\tau}}.
\end{aligned}
\end{equation}
The sum over $\pi \in S_{2k}$ can be split into sums over two classes of permutations: (i) permutations which can be decomposed in the form $\hat{\pi} = \hat{\pi}_1 \otimes \hat{\pi}_2$, where $\pi_1 \in S_k$ and $\pi_2 \in S_k$ are arbitrary permutations acting on the first and second $k$ replicas respectively, and (ii) permutations that cannot be decomposed in this form, which we denote by the shorthand $\pi \neq \pi_1 \pi_2$. This gives
\begin{equation}
\begin{aligned}
    (*) &= \frac{D^{2k}}{(2k)!D_{2k}}\sum_{z,z^\prime = 1}^{D_B} \sum_{\pi_1,\pi_2 \in S_{k}} \braket{z|\sigma_B|z}^{1-k}\braket{z^\prime|\sigma_B|z^\prime}^{1-k} \Tr \sparens{(I_A^{\otimes 2k} \otimes\ket{z^{\prime k} z^k}\bra{z^k z^{\prime k}}) \sigma^{\otimes 2k} (\hat{\pi}_1 \otimes \hat{\pi}_2) \hat{\tau}} \\
    &\quad +\frac{D^{2k}}{(2k)!D_{2k}}\sum_{z,z^\prime = 1}^{D_B} \sum_{\substack{\pi \in S_{2k} \\ \pi \neq \pi_1\pi_2}} \braket{z|\sigma_B|z}^{1-k}\braket{z^\prime|\sigma_B|z^\prime}^{1-k} \Tr \sparens{(I_A^{\otimes 2k} \otimes \ket{z^{\prime k} z^k}\bra{z^k z^{\prime k}}) \sigma^{\otimes 2k} \hat{\pi} \hat{\tau}} \\
    & = \frac{D^{2k}}{(2k)!D_{2k}} \sum_{z,z^\prime = 1}^{D_B} \sum_{\pi_1,\pi_2 \in S_k} \braket{z|\sigma_B|z}\braket{z^\prime|\sigma_B|z^\prime} \Tr\parens{\hat{\sigma}_{A|z}^{\otimes k} \hat{\pi}_1 \hat{\sigma}_{A|z^\prime}^{\otimes k} \hat{\pi}_2} \\
    &\quad + \frac{D^{2k}}{(2k)!D_{2k}}\sum_{z,z^\prime = 1}^{D_B} \sum_{\substack{\pi \in S_{2k} \\ \pi \neq \pi_1\pi_2}} \braket{z|\sigma_B|z}^{1-k}\braket{z^\prime|\sigma_B|z^\prime}^{1-k} \Tr \sparens{(I_A^{\otimes 2k} \otimes \ket{z^{\prime k} z^k}\bra{z^k z^{\prime k}}) \sigma^{\otimes 2k} \hat{\pi} \hat{\tau}} \\
    &= \frac{D^{2k}}{(2k)!D_{2k}} (k!)^2 D_{A,k} \sum_{z,z^\prime = 1}^{D_B} \braket{z|\sigma_B|z}\braket{z^\prime|\sigma_B|z^\prime} \Tr\parens{\hat{\sigma}_{A|z}^{\otimes k} \hat{\sigma}_{A|z^\prime}^{\otimes k} \rho_{\text{Haar},A}^{(k)}} \\
    &\quad + \frac{D^{2k}}{(2k)!D_{2k}}\sum_{z,z^\prime = 1}^{D_B} \sum_{\substack{\pi \in S_{2k} \\ \pi \neq \pi_1\pi_2}} \braket{z|\sigma_B|z}^{1-k}\braket{z^\prime|\sigma_B|z^\prime}^{1-k} \Tr \sparens{(I_A^{\otimes 2k} \otimes \ket{z^{\prime k} z^k}\bra{z^k z^{\prime k}}) \sigma^{\otimes 2k} \hat{\pi} \hat{\tau}}.
\end{aligned}
\end{equation}
Substituting the expressions for $(*)$ and $(**)$ into Eq.~\eqref{eq:two_norm_error} yields
\begin{equation}
\label{eq:two_norm_error_simplified1}
\begin{aligned}
    &\Eset{\Psi \sim \text{$\tilde{S}$crooge}(\sigma)} \norm{\tilde{\rho}_{\mathcal{E}}^{(k)} - \Eset{\Psi \sim \text{$\tilde{S}$crooge}(\sigma)}\tilde{\rho}_{\mathcal{E}}^{(k)}}_2^2  \\
    &\change{\leq} \bigO{\frac{k^2}{D}} (k!)^2 D_{A,k} \sum_{z,z^\prime = 1}^{D_B} \braket{z|\sigma_B|z}\braket{z^\prime|\sigma_B|z^\prime} \Tr\parens{\hat{\sigma}_{A|z}^{\otimes k} \hat{\sigma}_{A|z^\prime}^{\otimes k} \rho_{\text{Haar},A}^{(k)}} \\
    &\quad + \frac{D^{2k}}{(2k)!D_{2k}}\sum_{z,z^\prime = 1}^{D_B} \sum_{\substack{\pi \in S_{2k} \\ \pi \neq \pi_1\pi_2}} \braket{z|\sigma_B|z}^{1-k}\braket{z^\prime|\sigma_B|z^\prime}^{1-k} \Tr \sparens{(I_A^{\otimes 2k} \otimes \ket{z^{\prime k} z^k}\bra{z^k z^{\prime k}}) \sigma^{\otimes 2k} \hat{\pi} \hat{\tau}}.
\end{aligned}
\end{equation}
To bound the first term on the right hand side of Eq.~\eqref{eq:two_norm_error_simplified1}, we use H\"older's inequality to get
\begin{equation}\label{eq:bound1}
\Tr\parens{\hat{\sigma}_{A|z}^{\otimes k} \hat{\sigma}_{A|z^\prime}^{\otimes k} \rho_{\text{Haar},A}^{(k)}} \leq \frac{1}{D_{A,k}} \Tr^k\parens{\hat{\sigma}_{A|z}} \Tr^{k}\parens{\hat{\sigma}_{A|z^\prime}} = \frac{1}{D_{A,k}}. 
\end{equation}
To bound the second term, note that it must necessarily involve cross-terms like $\sigma_{A|zz^\prime} \equiv (I_A \otimes \bra{z})\sigma(I_A \otimes \ket{z^\prime})$ and $\sigma_{A|z^\prime z} = \sigma_{A|zz^\prime}^\dag$, due to the constraint $\pi \neq \pi_1 \pi_2$. Moreover, $\sigma_{A|zz^\prime}$ and $\sigma_{A|z^\prime z}$ appear in pairs. %
Thus, using H\"older's inequality, $\norm{\sigma_{A|zz^\prime}}_1^2 \leq \norm{\sigma_{A|z}}_1 \norm{\sigma_{A|z^\prime}}_1$, and
\begin{equation}
\begin{aligned}
    \sum_{\substack{\pi \in S_{2k} \\ \pi \neq \pi_1\pi_2}} \Tr \sparens{(I_A^{\otimes 2k} \otimes \ket{z^{\prime k} z^k}\bra{z^k z^{\prime k}}) \sigma^{\otimes 2k} \hat{\pi} \hat{\tau}} &\leq \parens{(2k)! - k!^2} \abs{\Tr \sparens{(I_A^{\otimes 2k} \otimes \ket{z^{\prime k} z^k}\bra{z^k z^{\prime k}}) \sigma^{\otimes 2k} \hat{\pi} \hat{\tau}}} \\
    &= \parens{(2k)! - k!^2} \norm{\sigma_{A|zz^\prime}}_1^2 \norm{\sigma_{A|z}}_1^{k-1} \norm{\sigma_{A|z^\prime}}_1^{k-1} \\
    &= \parens{(2k)! - k!^2} \norm{\sigma_{A|zz^\prime}}_1^2 \braket{z|\sigma_B|z}^{k-1} \braket{z^\prime|\sigma_B|z^\prime}^{k-1}.
\end{aligned}
\end{equation}
Substituting these bounds into Eq.~\eqref{eq:two_norm_error_simplified1}, we get
\begin{equation}\label{eq:two_norm_error_simplified2}
\begin{aligned}
\Eset{\Psi \sim \text{$\tilde{S}$crooge}(\sigma)} \norm{\tilde{\rho}_{\mathcal{E}}^{(k)} - \Eset{\Psi \sim \text{$\tilde{S}$crooge}(\sigma)}\tilde{\rho}_{\mathcal{E}}^{(k)}}_2^2  &\change{\leq} \bigO{\frac{k^{2k+2}}{D}}+ \frac{D^{2k}}{(2k)!D_{2k}} ((2k)! - k!^2)\sum_{z,z^\prime = 1}^{D_B} \norm{\sigma_{A|zz^\prime}}_1^2 \\
&\leq \bigO{\frac{k^{2k+2}}{D}} + \bigO{\change{4^k k^{2k} D_A} \norm{\sigma}_2^2},
\end{aligned}    
\end{equation}
where we have used the inequality
\begin{equation}
\begin{aligned}
\sum_{z,z^\prime = 1}^{D_B} \norm{\sigma_{A|zz^\prime}}_1^2 \leq D_A \sum_{z,z^\prime = 1}^{D_B} \norm{\sigma_{A|zz^\prime}}_2^2 = D_A \norm{\sigma}_2^2.
\end{aligned}
\end{equation}
Combining all the results above, we have an upper bound on the average trace distance,
\begin{equation}
\label{eq:expectedTD_bound_term2}
\begin{aligned}
   \parens{\Eset{\Psi \sim \text{$\tilde{S}$crooge}(\sigma)} \norm{\tilde{\rho}_{\mathcal{E}}^{(k)} - \Eset{\Psi \sim \text{$\tilde{S}$crooge}(\sigma)}\tilde{\rho}_{\mathcal{E}}^{(k)}}_1}^2 &\leq D_{A,k} \Eset{\Psi \sim \text{$\tilde{S}$crooge}(\sigma)} \norm{\tilde{\rho}_{\mathcal{E}}^{(k)} - \Eset{\Psi \sim \text{$\tilde{S}$crooge}(\sigma)}\tilde{\rho}_{\mathcal{E}}^{(k)}}_2^2 \\
   &\leq \bigO{\frac{k^{k+2} D_A^{k-1}}{D_B
   } \parens{1 + \change{\frac{4^k}{k^2} {D_A}^2 D_B} \norm{\sigma}_2^2}}.
\end{aligned}
\end{equation}
Combining Eqs.~\eqref{eq:expectedTD_bound_term1} and~\eqref{eq:expectedTD_bound_term2}, and using the triangle inequality, we get
\begin{equation}
\begin{aligned}
\Eset{\Psi \sim \text{$\tilde{S}$crooge}(\sigma)} \norm{\rho_{\mathcal{E}}^{(k)} - \Eset{\Psi \sim \text{$\tilde{S}$crooge}(\sigma)} \tilde{\rho}_{\mathcal{E}}^{(k)}}_1 &\leq \Eset{\Psi \sim \text{$\tilde{S}$crooge}(\sigma)} \norm{\rho_{\mathcal{E}}^{(k)} - \tilde{\rho}_{\mathcal{E}}^{(k)}}_1 + \Eset{\Psi \sim \text{$\tilde{S}$crooge}(\sigma)} \norm{\tilde{\rho}_{\mathcal{E}}^{(k)} - \Eset{\Psi \sim \text{$\tilde{S}$crooge}(\sigma)} \tilde{\rho}_{\mathcal{E}}^{(k)}}_1 \\
&\leq \sum_{z=1}^{D_B} \braket{z|\sigma_B|z} \bigO{k \norm{\hat{\sigma}_{A|z}}_2} + \bigO{\sqrt{\frac{k^{k+2} D_A^{k-1}}{D_B
   } \parens{1 + \change{\frac{4^k}{k^2} {D_A}^2 D_B} \norm{\sigma}_2^2}}},
\end{aligned}
\end{equation}
where $\E_\Psi \tilde{\rho}_\mathcal{E}^{(k)}$ is given by
Eq.~\eqref{eq:average_approx_projens}. This implies that
\begin{equation}
\Eset{\Psi \sim \text{$\tilde{S}$crooge}(\sigma)} \norm{\rho_{\mathcal{E}}^{(k)} - \sum_{z=1}^{D_B}\braket{z|\sigma_B|z}\tilde{\rho}_\text{Scrooge}^{(k)}(\hat{\sigma}_{A|z})}_1 \leq \sum_{z=1}^{D_B} \braket{z|\sigma_B|z} \bigO{k \norm{\hat{\sigma}_{A|z}}_2} + \bigO{\sqrt{\frac{k^{k+2} D_A^{k-1}}{D_B
   } \parens{1 + \change{\frac{4^k}{k^2} {D_A}^2 D_B} \norm{\sigma}_2^2}}}.
\end{equation}
To relate this to a mixture of $\rho_{\text{Scrooge}}^{(k)}(\hat{\sigma}_{A|z})$, we use Lemma~\ref{lemma:scrooge_approx} and the triangle inequality to get
\begin{equation}
\begin{aligned}
\norm{\sum_{z=1}^{D_B} \braket{z|\sigma_B|z} \parens{\rho_{\text{Scrooge}}^{(k)}(\hat{\sigma}_{A|z}) - \tilde{\rho}_{\text{Scrooge}}^{(k)}(\hat{\sigma}_{A|z})}}_1 &\leq\sum_{z=1}^{D_B} \braket{z|\sigma_B|z}\norm{ \rho_{\text{Scrooge}}^{(k)}(\hat{\sigma}_{A|z}) - \tilde{\rho}_{\text{Scrooge}}^{(k)}(\hat{\sigma}_{A|z})}_1 \\
&= \sum_{z=1}^{D_B}\braket{z|\sigma_B|z} \bigO{k \norm{\hat{\sigma}_{A|z}}_2}. \quad &\text{(Lemma~\ref{lemma:scrooge_approx})} 
\end{aligned}
\end{equation}
Finally, applying the triangle inequality again yields the desired result
\begin{equation}
\begin{aligned}
\Eset{\Psi \sim \text{$\tilde{S}$crooge}(\sigma)} \norm{\rho_{\mathcal{E}}^{(k)} - \sum_{z=1}^{D_B}\braket{z|\sigma_B|z}\rho_\text{Scrooge}^{(k)}(\hat{\sigma}_{A|z})}_1  \leq\bigO{\sqrt{\frac{k^{k+2} D_A^{k-1}}{D_B
   } \parens{1 + \change{\frac{4^k}{k^2} {D_A}^2 D_B} \norm{\sigma}_2^2}}} + \sum_{z=1}^{D_B} \braket{z|\sigma_B|z} \bigO{k \norm{\hat{\sigma}_{A|z}}_2}.
\end{aligned}   
\end{equation}
\end{proof}
We will also make use of the following lemma.
\begin{lemma}\normalfont
    Let $\rho_{\mathcal{E}}^{(k)}(\Psi)$ be the $k$th moment of the projected ensemble $\mathcal{E}(\Psi)$ on subsystem $A$ generated by $\ket{\Psi}_{AB}$, measured in an arbitrary orthonormal basis $\{\ket{z}\}_{z=1}^{D_B}$ on subsystem $B$. For any fixed operator $M$, and any density matrix $\sigma$ with dimension $D$,
\begin{equation}
    \abs{\Eset{\Psi \sim \text{Scrooge}(\sigma)} \norm{\rho^{(k)}_\mathcal{E}(\Psi) - M}_1 - \Eset{\Psi \sim \text{$\tilde{\text{S}}$crooge}(\sigma)} \norm{\rho^{(k)}_\mathcal{E}(\Psi) - M}_1} \leq \change{\sqrt{\frac{D\norm{\sigma}_2^2 - 1}{D+1}}} \norm{M}_1.
\end{equation}
\label{lemma:projens_approxScrooge2}
\end{lemma}
\begin{proof}
Let us write $\rho_{\mathcal{E}}^{(k)}$ explicitly as
\begin{equation}
    \rho_{\mathcal{E}}^{(k)} = \sum_{z=1}^{D_B} \frac{[(I_A \otimes \bra{z})\ket{\Psi}\bra{\Psi}(I_A \otimes \ket{z})]^{\otimes k}}{\braket{\Psi|(I_A \otimes \ket{z}\bra{z})|\Psi}^{k-1}}.
\end{equation}
By the definitions of the ensembles $\text{Scrooge}(\sigma)$ and $\tilde{\text{S}}$crooge($\sigma$), we have
\begin{equation}
\begin{aligned}
\Eset{\Psi \sim \text{Scrooge}(\sigma)} \norm{\rho_{\mathcal{E}}^{(k)} - M}_1 &= D \Eset{\phi \sim \text{Haar}(D)} \braket{\phi|\sigma|\phi} \norm{\frac{1}{\braket{\phi|\sigma|\phi}} \sum_{z=1}^{D_B} \frac{[(I_A \otimes \bra{z})\sqrt{\sigma}\ket{\phi}\bra{\phi}\sqrt{\sigma}(I_A \otimes \ket{z})]^{\otimes k}}{\braket{\phi|\sqrt{\sigma}(I_A \otimes \ket{z}\bra{z})\sqrt{\sigma}|\phi}^{k-1}} - M}_1 \\
&= \Eset{\phi \sim \text{Haar}(D)} \norm{D\sum_{z=1}^{D_B} \frac{[(I_A \otimes \bra{z})\sqrt{\sigma}\ket{\phi}\bra{\phi}\sqrt{\sigma}(I_A \otimes \ket{z})]^{\otimes k}}{\braket{\phi|\sqrt{\sigma}(I_A \otimes \ket{z}\bra{z})\sqrt{\sigma}|\phi}^{k-1}} - D \braket{\phi|\sigma|\phi} M}_1,
\end{aligned}
\end{equation}
and
\begin{equation}
\begin{aligned}
   \Eset{\Psi \sim \text{$\tilde{\text{S}}$crooge}(\sigma)} \norm{\rho_{\mathcal{E}}^{(k)} - M}_1 &= \Eset{\phi \sim \text{Haar}(D)} \norm{D\sum_{z=1}^{D_B} \frac{[(I_A \otimes \bra{z})\sqrt{\sigma}\ket{\phi}\bra{\phi}\sqrt{\sigma}(I_A \otimes \ket{z})]^{\otimes k}}{\braket{\phi|\sqrt{\sigma}(I_A \otimes \ket{z}\bra{z})\sqrt{\sigma}|\phi}^{k-1}} - M}_1.
\end{aligned}
\end{equation}
For convenience, let us denote
\begin{equation}
    K(\phi) \equiv D\sum_{z=1}^{D_B} \frac{[(I_A \otimes \bra{z})\sqrt{\sigma}\ket{\phi}\bra{\phi}\sqrt{\sigma}(I_A \otimes \ket{z})]^{\otimes k}}{\braket{\phi|\sqrt{\sigma}(I_A \otimes \ket{z}\bra{z})\sqrt{\sigma}|\phi}^{k-1}}.
\end{equation}
Therefore,
\begin{equation}
\begin{aligned}
    &\abs{\Eset{\Psi \sim \text{Scrooge}(\sigma)} \norm{\rho^{(k)}_\mathcal{E}(\Psi) - M}_1 - \Eset{\Psi \sim \text{$\tilde{\text{S}}$crooge}(\sigma)} \norm{\rho^{(k)}_\mathcal{E}(\Psi) - M}_1} \\ = &\abs{\Eset{\phi \sim \text{Haar}(D)} \parens{ \norm{K(\phi) - D\braket{\phi|\sigma|\phi} M}_1 - \norm{K(\phi) - M}_1 } } \\
    \leq&\Eset{\phi \sim \text{Haar}(D)}\abs{\norm{K(\phi) - D\braket{\phi|\sigma|\phi} M}_1 - \norm{K(\phi) - M}_1} \quad & \text{(Jensen's inequality)} \\
    \leq&\Eset{\phi\sim \text{Haar}(D)} \norm{D\braket{\phi|\sigma|\phi} M - M}_1 \quad &\text{(Reverse triangle inequality)} \\
    =& \parens{\Eset{\phi \sim \text{Haar}(D)} \abs{1 - D\change{\braket{\phi|\sigma|\phi}}}} \norm{M}_1 \\
    \leq& \sparens{\Eset{\phi \sim \text{Haar}(D)}\parens{1 - D\braket{\phi|\sigma|\phi}}^2}^{1/2} \norm{M}_1 \quad &\text{(Cauchy-Schwarz)} \\
    =&\change{\sqrt{\frac{D\norm{\sigma}_2^2 - 1}{D+1}}} \norm{M}_1.
\end{aligned}
\end{equation}
\end{proof}
Now, using Lemmas~\ref{lemma:projens_approxScrooge} and~\ref{lemma:projens_approxScrooge2}, we can analyze the case where the generator state is drawn from the exact Scrooge$(\sigma)$.
\begin{theorem}\normalfont[Projected ensemble generated by a Scrooge state]\label{thm:projens_approxScrooge}
Let $\ket{\Psi}_{AB}$ be a state drawn from Scrooge$(\sigma)$. Denote the reduced density matrices of $\sigma$ on $A$ and $B$ by $\sigma_A$ and $\sigma_B$ respectively. Consider the projected ensemble $\mathcal{E}(\Psi)$ obtained by applying projective measurements on $B$ in an arbitrary orthonormal basis $\{\ket{z}\}_{z=1}^{D_B}$. Then, assuming that $k^2 \norm{\hat{\sigma}_{A|z}}_2 \ll 1$ for all $z$, and $1 \ll D_A \leq D_B$,
\begin{equation}
\begin{aligned}
&\Eset{\Psi \sim \text{Scrooge}(\sigma)} \norm{\rho_{\mathcal{E}}^{(k)} - \sum_{z=1}^{D_B}\braket{z|\sigma_B|z}\rho_\text{Scrooge}^{(k)}(\hat{\sigma}_{A|z})}_1 \\ \leq \, & \bigO{\sqrt{\frac{k^{k+2} D_A^{k-1}}{D_B
   } \parens{1 + \change{\frac{4^k}{k^2} {D_A}^2 D_B} \norm{\sigma}_2^2}}} + \sum_{z=1}^{D_B} \braket{z|\sigma_B|z} \bigO{k \norm{\hat{\sigma}_{A|z}}_2} + \change{\sqrt{\frac{D\norm{\sigma}_2^2 - 1}{D+1}}}.
\end{aligned}   
\end{equation}
where $\sigma_{A|z} \equiv (I_A \otimes \bra{z})\sigma(I_A \otimes \ket{z})$, and $\hat{\sigma}_{A|z} =  \sigma_{A|z}/\braket{z|\sigma_B|z}$ is the normalized conditional mixed state on $A$.
\end{theorem}
\begin{proof}
Using Lemma~\ref{lemma:projens_approxScrooge2} with
\begin{equation}
    M = \sum_{z=1}^{D_B} \braket{z|\sigma_B|z} {\rho}_{\text{Scrooge}}^{(k)}(\hat{\sigma}_{A|z}),
\end{equation}
we have
\begin{equation}
    \Eset{\Psi \sim \text{Scrooge$(\sigma)$}} \norm{\rho_\mathcal{E}^{(k)} - M}_1 \leq \Eset{\Psi \sim \text{$\tilde{\text{S}}$crooge$(\sigma)$}} \norm{\rho_\mathcal{E}^{(k)} - M}_1 + \change{\sqrt{\frac{D\norm{\sigma}_2^2 - 1}{D+1}}} \norm{M}_1.
\end{equation}
From Lemma~\ref{lemma:projens_approxScrooge}, we get
\begin{equation}
\begin{aligned}
&\Eset{\Psi \sim \text{$\tilde{S}$crooge}(\sigma)} \norm{\rho_{\mathcal{E}}^{(k)} - M}_1 \\ \leq \, &\bigO{\sqrt{\frac{k^{k+2} D_A^{k-1}}{D_B
   } \parens{1 + \change{\frac{4^k}{k^2} {D_A}^2 D_B} \norm{\sigma}_2^2}}} + \sum_{z=1}^{D_B} \braket{z|\sigma_B|z} \bigO{k \norm{\hat{\sigma}_{A|z}}_2}.
\end{aligned}   
\end{equation}
We can also bound $\norm{M}_1$ via
\begin{equation}
    \norm{M}_1 \leq \sum_{z=1}^{D_B} \braket{z|\sigma_B|z} \norm{\rho_{\text{Scrooge}}^{(k)}(\hat{\sigma}_{A|z})}_1 = \sum_{z=1}^{D_B} \braket{z|\sigma_B|z},
\end{equation}
since $\rho^{(k)}_{\text{Scrooge}}(\hat{\sigma}_{A|z})$ is a normalized density operator. Therefore,
\begin{equation}
\begin{aligned}
&\Eset{\Psi \sim \text{Scrooge}(\sigma)} \norm{\rho_{\mathcal{E}}^{(k)} - \sum_{z=1}^{D_B}\braket{z|\sigma_B|z}\rho_\text{Scrooge}^{(k)}(\hat{\sigma}_{A|z})}_1 \\ \leq \, & \bigO{\sqrt{\frac{k^{k+2} D_A^{k-1}}{D_B
   } \parens{1 + \change{\frac{4^k}{k^2} {D_A}^2 D_B} \norm{\sigma}_2^2}}} + \sum_{z=1}^{D_B} \braket{z|\sigma_B|z} \bigO{k \norm{\hat{\sigma}_{A|z}}_2} + \change{\sqrt{\frac{D\norm{\sigma}_2^2 - 1}{D+1}}}.
\end{aligned}   
\end{equation}
\end{proof}
 
\begin{theorem}\normalfont[Projected ensemble generated by an approximate Scrooge $2k$-design state]\label{thm:2ktok_scrooge_app}
Let $\ket{\Psi}_{AB}$ be sampled from an approximate Scrooge$(\sigma)$ $2k$-design, with relative error $\epsilon$. Denote the reduced state of $\sigma$ on $A$ and $B$ by $\sigma_A$ and $\sigma_B$ respectively. Consider the projected ensemble $\mathcal{E}$ obtained by applying projective measurements on $B$ in an arbitrary orthonormal basis $\{\ket{z}\}_{z=1}^{D_B}$. Then, assuming that $k^2 \norm{\hat{\sigma}_{A|z}}_2 \ll 1$ for all $z$, and $1 \ll D_A \leq D_B$,
\begin{equation}
\begin{aligned}
&\Eset{\Psi \sim \text{$2k$-design}} \norm{{\rho}_{\mathcal{E}}^{(k)} - \sum_{z=1}^{D_B} \braket{z|\sigma_B|z} \rho_{\text{Scrooge}}^{(k)}(\hat{\sigma}_{A|z})}_1 \\&\leq \bigO{\sqrt{\frac{k^{k+2}D_A^{k-1}}{D_B}\parens{1 + \change{\frac{4^k}{k^2} {D_A}^2 D_B} \norm{\sigma}_2^2 + \frac{D \epsilon}{k^2} + \frac{4^k}{k} D \norm{\sigma}_2}}} + \sum_{z=1}^{D_B} \braket{z|\sigma_B|z} \bigO{\sqrt{\epsilon + 4^k k \norm{\sigma}_2 + k^2 \norm{\hat{\sigma}_{A|z}}_2^2}}. 
\end{aligned}
\end{equation}
where $\sigma_{A|z} \equiv (I_A \otimes \bra{z})\sigma(I_A \otimes \ket{z})$, and $\hat{\sigma}_{A|z} =  \sigma_{A|z}/\braket{z|\sigma_B|z}$ is the normalized conditional mixed state on $A$.
\end{theorem}
\begin{proof}
From Lemma~\ref{lemma:scrooge_approx}, the proxy $\tilde{\text{S}}$crooge$(\sigma)$ is an approximate Scrooge$(\sigma)$ $2k$-design with relative error $\epsilon^\prime = \bigO{4^k k \norm{\sigma}_2}$. This implies that any approximate Scrooge$(\sigma)$ $2k$-design with relative error $\epsilon$ must also be close to $\tilde{\text{S}}$crooge$(\sigma)$ up to $2k$ moments with a relative error $\epsilon^{\prime \prime}$ that satisfies
\begin{equation}
    1+\epsilon^{\prime \prime} = (1+\epsilon)(1+\epsilon^\prime)
\end{equation}
which gives $\epsilon^{\prime \prime} = \epsilon +  \epsilon^\prime + \epsilon \epsilon^\prime = \bigO{\epsilon + \epsilon^\prime} = \bigO{\epsilon + 4^k k \norm{\sigma}_2}$. In other words, we have
\begin{equation}
    (1 - \epsilon^{\prime \prime}) \tilde{\rho}_{\text{Scrooge}}^{(2k)}(\sigma) \preceq \Eset{\Psi \sim \text{$2k$-design}} \parens{\ket{\Psi}\bra{\Psi}}^{\otimes 2k} \preceq (1 + \epsilon^{\prime \prime}) \tilde{\rho}_{\text{Scrooge}}^{(2k)}(\sigma).
\end{equation}
Now, using Lemma~\ref{lemma:projens_relerrorgens}, where we choose $\mathcal{E}_1$ to be $\tilde{\text{S}}$crooge$(\sigma)$, and $\mathcal{E}_2$ to be the approximate Scrooge$(\sigma)$ $2k$-design, the fixed operator
\begin{equation}
    M = \Eset{\Psi \sim \text{$\tilde{\text{S}}$crooge}(\sigma)} \tilde{\rho}_\mathcal{E}^{(k)} = \parens{1 + \bigO{\frac{k^2}{D_A}}}\sum_{z=1}^{D_B}\braket{z|\sigma_B|z}\change{\tilde{\rho}_\text{Scrooge}^{(k)}}(\hat{\sigma}_{A|z}),
\end{equation}
where we have defined
\begin{equation}
    \tilde{\rho}_{\mathcal{E}}^{(k)}(\Psi)  = \sum_{z=1}^{D_B} \braket{z|\sigma_B|z}^{1-k} [(I_A \otimes \bra{z})\ket{\Psi}\bra{\Psi}(I_A \otimes \ket{z})]^{\otimes k}
\end{equation}
for the projected ensemble $\mathcal{E}$ with $q_z = \braket{z|\sigma_B|z}$, we get
\begin{equation}
\begin{aligned}
    \Delta &= \abs{\Eset{\Psi \sim \text{$\tilde{\text{S}}$crooge}(\sigma)} \norm{\tilde{\rho}_{\mathcal{E}}^{(k)}(\Psi) - M}_2^2 - \Eset{\Psi \sim \text{$2k$-design}} \norm{\tilde{\rho}_{\mathcal{E}}^{(k)}(\Psi) - M}_2^2} \\
    &\leq \epsilon^{\prime \prime} \Eset{\Psi \sim \text{$\tilde{\text{S}}$crooge}(\sigma)} \sparens{2\Tr\parens{ \tilde{\rho}_{\mathcal{E}}^{(k)}(\Psi) M} + \parens{\sum_{z=1}^{D_B} \braket{z|\sigma_B|z}^{1-k} \braket{\Psi|(I_A \otimes \ket{z}\bra{z})|\Psi}^k}^2}.
\end{aligned}
\end{equation}
Rearranging,
\begin{equation}
\begin{aligned}
    \Eset{\Psi \sim \text{$2k$-design}} \norm{\tilde{\rho}_{\mathcal{E}}^{(k)}(\Psi) - M}_2^2 &\leq \Eset{\Psi \sim \text{$\tilde{\text{S}}$crooge}(\sigma)} \norm{\tilde{\rho}_{\mathcal{E}}^{(k)}(\Psi) - M}_2^2 \\
    &\quad +\epsilon^{\prime \prime} \Eset{\Psi \sim \text{$\tilde{\text{S}}$crooge}(\sigma)} \sparens{2\Tr\parens{ \tilde{\rho}_{\mathcal{E}}^{(k)}(\Psi) M} + \parens{\sum_{z=1}^{D_B} \braket{z|\sigma_B|z}^{1-k} \braket{\Psi|(I_A \otimes \ket{z}\bra{z})|\Psi}^k}^2}.
\end{aligned}
\end{equation}
The first term on the right hand side is given by Eq.~\eqref{eq:two_norm_error_simplified2} in the proof of Theorem~\ref{thm:projens_approxScrooge} above, with
\begin{equation}
    \Eset{\Psi \sim \text{$\tilde{\text{S}}$crooge}(\sigma)} \norm{\tilde{\rho}_{\mathcal{E}}^{(k)}(\Psi) - M}_2^2 \leq \bigO{\frac{k^{2k+2}}{D}} + \bigO{\change{4^k k^{2k} D_A} \norm{\sigma}_2^2}.
\end{equation}
Next, by linearity of the expectation and trace,
\begin{equation}
    \Eset{\Psi \sim \text{$\tilde{\text{S}}$crooge}(\sigma)} \Tr\parens{ \tilde{\rho}_{\mathcal{E}}^{(k)}(\Psi) M} = \Tr \parens{\Eset{\Psi \sim \text{$\tilde{\text{S}}$crooge}(\sigma)} \tilde{\rho}_{\mathcal{E}}^{(k)}(\Psi)}^2 \leq \frac{D^{2k}}{D_k^2} = \bigO{k!^2},
\end{equation}
which can be obtained from Eqs.~\eqref{eq:starstar1} and~\eqref{eq:bound1}. Furthermore, \change{by the Cauchy-Schwarz inequality,}
\begin{equation}
\begin{aligned}
   \Eset{\Psi \sim \text{$\tilde{\text{S}}$crooge}(\sigma)}  \parens{\sum_{z=1}^{D_B} \braket{z|\sigma_B|z}^{1-k} \braket{\Psi|(I_A \otimes \ket{z}\bra{z})|\Psi}^k}^2 &\change{\leq \sparens{\sum_{z=1}^{D_B} \braket{z|\sigma_B|z}^{1-k} \parens{\Eset{\Psi \sim \text{$\tilde{S}$crooge}(\sigma)}\braket{\Psi|(I_A \otimes \ket{z}\bra{z})|\Psi}^{2k}}^{1/2}}^2} 
   \\&= \parens{1 + \bigO{\frac{k^2}{D}}} \sum_{z,z^\prime=1}^{D_B} \braket{z|\sigma_B|z}\braket{z^\prime|\sigma_B|z^\prime} \parens{1 + \bigO{k^2 \norm{\hat{\sigma}_{A|z}}_2^2}} \\
   &= 1 + o(1),
\end{aligned}
\end{equation}
since we assumed $k^2 \norm{\hat{\sigma}_{A|z}}_2 \ll 1$ for all $z$. The notation $o(1)$ is a shorthand for terms that vanish asymptotically in the limit of large Hilbert space dimensions. Thus,
\begin{equation}
\begin{aligned}
    \Eset{\Psi \sim \text{$2k$-design}} \norm{\tilde{\rho}_{\mathcal{E}}^{(k)}(\Psi) - M}_2^2 &\leq \bigO{\frac{k^{2k+2}}{D}} + \bigO{\change{4^k k^{2k} D_A} \norm{\sigma}_2^2} + \bigO{\epsilon^{\prime \prime} k^{2k}}
    \\&= \bigO{\frac{k^{2k+2}}{D}\parens{1 + \change{\frac{4^k}{k^2} {D_A}^2 D_B} \norm{\sigma}_2^2 + \frac{D \epsilon^{\prime \prime}}{k^2}}}.
\end{aligned}
\end{equation}
This implies the trace distance bound
\begin{equation}
\Eset{\Psi \sim \text{$2k$-design}} \norm{\tilde{\rho}_{\mathcal{E}}^{(k)}(\Psi) - M}_1 \leq \sqrt{D_{A,k}}\parens{\Eset{\Psi \sim \text{$2k$-design}} \norm{\tilde{\rho}_{\mathcal{E}}^{(k)}(\Psi) - M}_2^2}^{1/2} \leq \bigO{\sqrt{\frac{k^{k+2} {D_A}^{k-1}}{D_B}\parens{1 + \change{\frac{4^k}{k^2} {D_A}^2 D_B} \norm{\sigma}_2^2 + \frac{D \epsilon^{\prime \prime}}{k^2}}}}.    
\end{equation}
The next step is to show that ${\rho}_{\mathcal{E}}^{(k)}$ is, with high probability, close to $\tilde{\rho}_{\mathcal{E}}^{(k)}$, when $\ket{\Psi}_{AB}$ is sampled from the Scrooge$(\sigma)$ 2k-design. To this end, we use Lemma~\ref{lemma:general_approx} to write
\begin{equation}
\begin{aligned}
    \Eset{\Psi \sim \text{$2k$-design}} \norm{\rho_{\mathcal{E}}^{(k)} - \tilde{\rho}_{\mathcal{E}}^{(k)}}_1 &\leq \sum_{z=1}^{D_B} \parens{\Eset{\Psi \sim \text{$2k$-design}} p_z^2}^{1/2} \sparens{1 - 2 \frac{\Eset{\Psi \sim \text{$2k$-design}} p_z^{k-1}}{\braket{z|\sigma_B|z}^{k-1}}+ \frac{\Eset{\Psi \sim \text{$2k$-design}} p_z^{2k-2}}{\braket{z|\sigma_B|z}^{2k-2}}}^{1/2}. 
\end{aligned}
\end{equation}
Using the relative error property
\begin{equation}
    (1-\epsilon^{\prime \prime}) \Eset{\Psi \sim \text{$\tilde{\text{S}}$crooge$(\sigma)$}} {p_z}^\ell \leq \Eset{\Psi \sim \text{$2k$-design}} {p_z}^\ell \leq (1+\epsilon^{\prime \prime}) \Eset{\Psi \sim \text{$\tilde{\text{S}}$crooge$(\sigma)$}} {p_z}^\ell
\end{equation}
for $1 \leq \ell \leq 2k$. From Eq.~\eqref{eq:probz_approxScr},
\begin{equation}
    \Eset{\Psi \sim \text{$\tilde{\text{S}}$crooge}(\sigma)} p_z^\ell = \braket{z|\sigma_B|z}^\ell \sparens{1 + \bigO{\ell^2 \norm{\hat{\sigma}_{A|z}}_2^2}}.
\end{equation}
Therefore,
\begin{equation}
\Eset{\Psi \sim \text{$2k$-design}} \norm{\rho_{\mathcal{E}}^{(k)} - \tilde{\rho}_{\mathcal{E}}^{(k)}}_1 \leq \sum_{z=1}^{D_B} \braket{z|\sigma_B|z} \bigO{\sqrt{\epsilon^{\prime \prime} + k^2 \norm{\hat{\sigma}_{A|z}}_2^2}}.
\end{equation}
By the triangle inequality, and substituting the definition of $M$,
\begin{equation}
\begin{aligned}
\Eset{\Psi \sim \text{$2k$-design}} \norm{{\rho}_{\mathcal{E}}^{(k)} - \sum_{z=1}^{D_B} \braket{z|\sigma_B|z} \rho_{\text{Scrooge}}^{(k)}(\hat{\sigma}_{A|z})}_1 &\leq \bigO{\sqrt{\frac{k^{k+2} {D_A}^{k-1}}{D_B}\parens{1 + \change{\frac{4^k}{k^2} {D_A}^2 D_B} \norm{\sigma}_2^2 + \frac{D \epsilon^{\prime \prime}}{k^2}}}} \\&+ \sum_{z=1}^{D_B} \braket{z|\sigma_B|z} \bigO{\sqrt{\epsilon^{\prime \prime} + k^2 \norm{\hat{\sigma}_{A|z}}_2^2}}.   
\end{aligned}
\end{equation}
\change{In the above, we have replaced $\tilde{\rho}_\text{Scrooge}^{(k)}(\hat{\sigma}_{A|z})$ by $\rho_\text{Scrooge}^{(k)}(\hat{\sigma}_{A|z})$ in $M$, because this only incurs an additional trace distance error of $\bigO{\sum_z \braket{z|\sigma_B|z} k \norm{\hat{\sigma}_{A|z}}_2}$ by the triangle inequality.} It can be verified that this recovers Theorem~\ref{thm:projens_approxScrooge} when $\epsilon^{\prime \prime} = 0$. As stated above, $\epsilon^{\prime \prime} = \bigO{\epsilon + 4^k k \norm{\sigma}_2}$, thus
\begin{equation}
\label{eq:gen_scrooge_proj_errorbound}
\begin{aligned}
\Eset{\Psi \sim \text{$2k$-design}} \norm{{\rho}_{\mathcal{E}}^{(k)} - \sum_{z=1}^{D_B} \braket{z|\sigma_B|z} \rho_{\text{Scrooge}}^{(k)}(\hat{\sigma}_{A|z})}_1 &\leq  \bigO{\sqrt{\frac{k^{k+2}D_A^{k-1}}{D_B}\parens{1 + \change{\frac{4^k}{k^2} {D_A}^2 D_B} \norm{\sigma}_2^2 + \frac{D \epsilon}{k^2} + \frac{4^k}{k} D \norm{\sigma}_2}}} \\&+ \sum_{z=1}^{D_B} \braket{z|\sigma_B|z} \bigO{\sqrt{\epsilon + 4^k k \norm{\sigma}_2 + k^2 \norm{\hat{\sigma}_{A|z}}_2^2}}.
\end{aligned}
\end{equation}
\end{proof}
For a fixed $k \in \mathbb{N}$, Eq.~\eqref{eq:gen_scrooge_proj_errorbound} simplifies to
\begin{equation}
\Eset{\Psi \sim \text{$2k$-design}} \norm{{\rho}_{\mathcal{E}}^{(k)} - \sum_{z=1}^{D_B} \braket{z|\sigma_B|z} \rho_{\text{Scrooge}}^{(k)}(\hat{\sigma}_{A|z})}_1 \leq \bigO{\sqrt{{D_A}^k \parens{\epsilon + \norm{\sigma}_2 + \change{D_A \norm{\sigma}_2^2}}}}    
\end{equation}
to leading order, giving Theorem~\ref{thm:2kgenerator} in the main text.

\subsection{Projected ensemble generated by a state drawn from a Haar $2k$-design}
In the special case where $\sigma = I/D$ is maximally mixed, the Scrooge ensemble reduces to the Haar ensemble, giving Corollary~\ref{cor:2ktokdesign}. In this limit, it turns out that the error bound can be improved compared to Theorem~\ref{thm:2kgenerator}, for technical reasons. The proof follows similarly as Theorem~\ref{thm:2ktok_scrooge_app} above.
\begin{theorem}[Projected ensemble generated by an approximate $2k$-design state]
Let $\ket{\Psi}_{AB}$ be sampled from a $2k$-design, with relative error $\epsilon$. Consider the projected ensemble $\mathcal{E}(\Psi)$ obtained by applying projective measurements on $B$ in an arbitrary orthonormal basis $\{\ket{z}\}_{z=1}^{D_B}$. Then, assuming $k^2 \ll D_A$, with $D_A,D_B \gg 1$,
\begin{equation}
\Eset{\Psi \sim \text{$2k$-design}} \norm{{\rho}_{\mathcal{E}}^{(k)} - \rho_{\text{Haar},A}^{(k)}}_1 \leq \sqrt{\frac{D_{A,k}}{D_B} + \bigO{\frac{k^2}{D_A} + D_{A,k} \epsilon}}.   
\end{equation}
\end{theorem}
\begin{proof}
The $k$th moment of the projected ensemble $\mathcal{E}(\Psi)$ can be written as
\begin{equation}
    \rho_{\mathcal{E}}^{(k)} = \sum_{z=1}^{D_B} p_z^{1-k} [(I_A \otimes \bra{z})\ket{\Psi}\bra{\Psi}(I_A \otimes \ket{z})]^{\otimes k},
\end{equation}
where
\begin{equation}
    p_z = \braket{\Psi|(I_A \otimes \ket{z}\bra{z})|\Psi}
\end{equation}
is the probability of measuring the outcome $z$. Let us construct the proxy
\begin{equation}
    \tilde{\rho}_\mathcal{E}^{(k)} = \sum_{z=1}^{D_B} D_B^{k-1} [(I_A \otimes \bra{z})\ket{\Psi}\bra{\Psi}(I_A \otimes \ket{z})]^{\otimes k}.
\end{equation}
The average trace distance between $\rho_{\mathcal{E}}^{(k)}$ and $\tilde{\rho}_{\mathcal{E}}^{(k)}$ can be bounded using Lemma~\ref{lemma:general_approx},
\begin{equation}
\begin{aligned}
    \Eset{\Psi \sim \text{$2k$-design}} \norm{\rho_{\mathcal{E}}^{(k)} - \tilde{\rho}_\mathcal{E}^{(k)}}_1 \leq \sum_{z=1}^{D_B} \parens{\Eset{\Psi \sim \text{$2k$-design}} p_z^2}^{1/2} \parens{1 - 2 D_B^{k-1}\Eset{\Psi \sim \text{$2k$-design}} p_z^{k-1} + D_B^{2k-2} \Eset{\Psi \sim \text{$2k$-design}} p_z^{2k-2}}^{1/2}.
\end{aligned}
\end{equation}
For any $1 \leq \ell \leq 2k$,
\begin{equation}
    (1-\epsilon) \Eset{\Psi \sim \text{Haar}(D)} {p_z}^\ell \leq \Eset{\Psi \sim \text{$2k$-design}} {p_z}^\ell \leq (1+\epsilon) \Eset{\Psi \sim \text{Haar}(D)} {p_z}^\ell,
\end{equation}
from the definition of the relative error $\epsilon$. The Haar average of ${p_z}^{\ell}$ can be explicitly evaluated as
\begin{equation}
   \Eset{\Psi \sim \text{Haar}(D)} {p_z}^\ell = \frac{1}{\ell! D_\ell} \sum_{\pi \in S_\ell} {D_A}^{\#\text{cycles}(\pi)} = \frac{D_{A,\ell}}{D_\ell} = \frac{1}{D_B^\ell} \parens{1 + \bigO{\frac{\ell^2}{D_A}}}.
\end{equation}
Thus,
\begin{equation}
    \Eset{\Psi \sim \text{$2k$-design}} \norm{\rho_{\mathcal{E}}^{(k)} - \tilde{\rho}_\mathcal{E}^{(k)}}_1 = \bigO{\sqrt{\epsilon + \frac{k^2}{D_A}}}.
\end{equation}
Next, using Lemma~\ref{lemma:projens_relerrorgens} with the choice of $q_z = 1/D_B$, $M = \rho_{\text{Haar},A}^{(k)}$,
\begin{equation}
\begin{aligned}
    \Delta &\equiv \abs{\Eset{\Psi \sim \text{$2k$-design}}\norm{\tilde{\rho}_\mathcal{E}^{(k)} - \rho_{\text{Haar,A}}^{(k)}}_2^2 - \Eset{\Psi \sim \text{Haar}(D)}\norm{\tilde{\rho}_\mathcal{E}^{(k)} - \rho_{\text{Haar,A}}^{(k)}}_2^2} \\
    &\leq 2\epsilon \Eset{\Psi \sim \text{Haar}(D)} \Tr \tilde{\rho}_\mathcal{E}^{(k)} \rho_{\text{Haar,A}}^{(k)} + \epsilon D_B^{2k-2} \Eset{\Psi \sim \text{Haar}(D)} \parens{\sum_{z=1}^{D_B} p_z^k}^2.
\end{aligned}
\end{equation}
The first term can be evaluated using
\begin{equation}
    \Eset{\Psi \sim \text{Haar}(D)} \Tr \tilde{\rho}_\mathcal{E}^{(k)} \rho_{\text{Haar,A}}^{(k)} = \frac{1}{D_{A,k}} \Eset{\Psi \sim \text{Haar}(D)} \Tr \tilde{\rho}_{\mathcal{E}}^{(k)} = \frac{D_B^{k-1}}{D_{A,k}} \sum_{z=1}^{D_B} \Eset{\Psi \sim \text{Haar}(D)} p_z^k = \frac{1}{D_{A,k}} \parens{1 + \bigO{\frac{k^2}{D_A}}}. 
\end{equation}
The second term can be evaluated using
\begin{equation}
\begin{aligned}
    D_B^{2k-2} \Eset{\Psi \sim \text{Haar}(D)}\parens{\sum_{z=1}^{D_B} p_z^k}^2 &= D_B^{2k-2} \sum_{z,z^\prime = 1}^{D_B} (I_A^{\otimes 2k} \otimes \bra{z^k z^{\prime k}}) \rho_{\text{Haar}}^{(2k)} (I_A^{\otimes 2k} \otimes \ket{z^k z^{\prime k}}) \\
    &= \frac{D_B^{2k-2}}{D^{2k}}\parens{1 + \bigO{\frac{k^2}{D}}} \sum_{z,z^\prime=1}^{D_B} \sum_{\pi \in S_{2k}} {D_A}^{\#\text{cycles}(\pi)} \bra{z^k z^{\prime k}} \hat{\pi}_B \ket{z^k z^{\prime k}} \\
    &= \frac{1}{D_B^2 D_A^{2k}}\parens{1 + \bigO{\frac{k^2}{D}}} \sparens{\sum_{z=1}^{D_B} \sum_{\pi \in S_{2k}} {D_A}^{\#\text{cycles}(\pi)} + \sum_{\substack{z,z^\prime=1 \\ z\neq z^\prime}}^{D_B} \parens{\sum_{\pi \in S_k} {D_A}^{\#\text{cycles}(\pi)}}^2} \\
    &= \frac{1}{D_B^2 D_A^{2k}}\parens{1 + \bigO{\frac{k^2}{D}}} \parens{D_B D_A^{2k} + D_B(D_B-1) D_A^{2k}}\parens{1 + \bigO{\frac{k^2}{D_A}}}\\ 
    &= 1 + \bigO{\frac{k^2}{D_A}}.
\end{aligned}
\end{equation}
Therefore, $\Delta = \bigO{\epsilon}$, and
\begin{equation}
    \Eset{\Psi \sim \text{$2k$-design}}\norm{\tilde{\rho}_\mathcal{E}^{(k)} - \rho_{\text{Haar,A}}^{(k)}}_2^2 \leq \Eset{\Psi \sim \text{Haar}(D)}\norm{\tilde{\rho}_\mathcal{E}^{(k)} - \rho_{\text{Haar,A}}^{(k)}}_2^2 + \bigO{\epsilon}.
\end{equation}
Now,
\begin{equation}
\begin{aligned}
    \Eset{\Psi \sim \text{Haar}(D)}\norm{\tilde{\rho}_\mathcal{E}^{(k)} - \rho_{\text{Haar,A}}^{(k)}}_2^2 &= \Eset{\Psi \sim \text{Haar}(D)} \Tr \tilde{\rho}_{\mathcal{E}}^{(k) 2} - 2 \Eset{\Psi \sim \text{Haar}(D)} \Tr \tilde{\rho}_{\mathcal{E}}^{(k)} \rho_{\text{Haar},A}^{(k)} + \frac{1}{D_{A,k}} \\
    &= \Eset{\Psi \sim \text{Haar}(D)} \Tr \tilde{\rho}_{\mathcal{E}}^{(k) 2} - \frac{2}{D_{A,k}} \parens{1 + \bigO{\frac{k^2}{D_A}}} + \frac{1}{D_{A,k}},
\end{aligned}
\end{equation}
where
\begin{equation}
\begin{aligned}
    \Eset{\Psi \sim \text{Haar}(D)} \Tr \tilde{\rho}_{\mathcal{E}}^{(k) 2} &= D_B^{2k-2} \sum_{z,z^\prime=1}^{D_B} (I_A^{\otimes 2k} \otimes \bra{z^k z^{\prime k}})\rho_{\text{Haar}}^{(2k)} (I_A^{\otimes 2k} \otimes \ket{z^k z^{\prime k}}) \\
    &= \frac{1}{D_B^2 D_A^{2k}} \parens{1 + \bigO{\frac{k^2}{D}}} \sum_{z,z^\prime=1}^{D_B} \sum_{\pi \in S_{2k}} D_A^{\#\text{cycles}(\pi)} \bra{z^k z^{\prime k}}\hat{\pi}_B \ket{z^{\prime k}z^k} \\
    &= \frac{1}{D_B^2 D_A^{2k}} \parens{1 + \bigO{\frac{k^2}{D}}} \sparens{D_B D_A^{2k}\parens{1+\bigO{\frac{k^2}{D_A}}} + D_B (D_B-1) k! D_A^k \parens{1 + \bigO{\frac{k^2}{D_A}}}} \\
    &= \frac{1}{D_{A,k}}\parens{1+\bigO{\frac{k^2}{D_A}}} + \frac{1}{D_B}.
\end{aligned}
\end{equation}
Thus,
\begin{equation}
    \Eset{\Psi \sim \text{Haar}(D)}\norm{\tilde{\rho}_\mathcal{E}^{(k)} - \rho_{\text{Haar,A}}^{(k)}}_2^2 = \frac{1}{D_B} + \bigO{\frac{k^2}{D_A D_{A,k}}}
\end{equation}
and
\begin{equation}
    \Eset{\Psi \sim \text{$2k$-design}}\norm{\tilde{\rho}_\mathcal{E}^{(k)} - \rho_{\text{Haar,A}}^{(k)}}_2^2 \leq \frac{1}{D_B} + \bigO{\frac{k^2}{D_A D_{A,k}} + \epsilon}.
\end{equation}
This implies that
\begin{equation}
    \Eset{\Psi \sim \text{$2k$-design}}\norm{\tilde{\rho}_\mathcal{E}^{(k)} - \rho_{\text{Haar,A}}^{(k)}}_1 \leq \parens{D_{A,k} \Eset{\Psi \sim \text{$2k$-design}}\norm{\tilde{\rho}_\mathcal{E}^{(k)} - \rho_{\text{Haar,A}}^{(k)}}_2^2}^{1/2} \leq \sqrt{\frac{D_{A,k}}{D_B} + \bigO{\frac{k^2}{D_A} + D_{A,k} \epsilon}}.
\end{equation}
Finally, by the triangle inequality,
\begin{equation}
    \Eset{\Psi \sim \text{$2k$-design}}\norm{{\rho}_\mathcal{E}^{(k)} - \rho_{\text{Haar,A}}^{(k)}}_1 \leq \Eset{\Psi \sim \text{$2k$-design}}\norm{\tilde{\rho}_\mathcal{E}^{(k)} - \rho_{\text{Haar,A}}^{(k)}}_1 + \Eset{\Psi \sim \text{$2k$-design}}\norm{\tilde{\rho}_{\mathcal{E}}^{(k)} - \rho_{\mathcal{E}}^{(k)}}_1 \leq \sqrt{\frac{D_{A,k}}{D_B} + \bigO{\frac{k^2}{D_A} + D_{A,k} \epsilon}},
\end{equation}
which yields the desired result.
\end{proof}

\subsection{Late-time chaotic Hamiltonian dynamics}
Consider a generator state $\ket{\Psi}_{AB}$ obtained by evolving the initial state $\ket{\Psi_0}_{AB}$ under an ergodic Hamiltonian $H$, for a late-time $t$. Generically, it is reasonable to expect $\ket{\Psi}_{AB}$ to be modeled by the random phase ensemble
\begin{equation}
\label{eq:randomphase_def_app}
    \mathcal{E}_{\text{Random Phase}} = \bparens{\frac{d^D \varphi}{(2\pi)^D}, \sum_{j=1}^{D} |\braket{E_j|\Psi_0}|e^{i\varphi_j} \ket{E_j}},
\end{equation}
with the diagonal density matrix (in the energy basis $\{\ket{E_j}\}_j)$
\begin{equation}
    \sigma_{\text{diag}} = \sum_{j=1}^{D} |\braket{E_j|\Psi_0}|^2 \ket{E_j}\bra{E_j},
\label{eq:app_diagonalens}
\end{equation}
as described in the main text. If $H$ satisfies the $k$th no-resonance condition, Mark et al.~\cite{mark2024maximum} showed that the temporal ensemble
\begin{equation}
    \mathcal{E}_{\text{temp}} = \bparens{e^{-iHt}\ket{\Psi_0} \,\,|\, t \in (-\infty,\infty)}
\end{equation}
is close to the random phase ensemble, up to the $k$th moment. While we have shown that the random phase ensemble forms an approximate Scrooge $k$-design (Theorem~\ref{thm:global_scrooge} in the main text), Theorem~\ref{thm:2kgenerator} does not strictly apply in this scenario, since the random phase ensemble is not close to Scrooge$(\sigma)$ in relative error. Nonetheless, it is interesting to ask whether the projected ensemble generated by a state drawn from the random phase ensemble forms a generalized Scrooge ensemble, as in Theorem~\ref{thm:2kgenerator}. 
This problem is partially resolved with Theorem 4 of Ref.~\cite{mark2024maximum}, where they showed that the $k$th moment of the \textit{unnormalized} projected ensemble $\mathcal{E}(\Psi)$ is approximately that of the \textit{unnormalized} generalized Scrooge ensemble, i.e.,
\begin{equation}
    \tilde{\rho}_\mathcal{E}^{(k)} \approx \sum_{z=1}^{D_B} \braket{z|\sigma_B|z} \tilde{\rho}_{\text{Scrooge}}^{(k)}(\hat{\sigma}_{A|z}),
\end{equation}
with an error that is expected to be exponential small in the system size. In the above, $\sigma_A$ and $\sigma_B$ are the reduced density matrices of $\sigma_{\text{diag}}$ on $A$ and $B$, respectively. $\sigma_{A|z} \equiv (I_A \otimes \bra{z})\sigma(I_A \otimes \ket{z})$, and $\hat{\sigma}_{A|z} =  \sigma_{A|z}/\braket{z|\sigma_B|z}$ is the normalized conditional mixed state on $A$. Using the technical results developed in proving Theorem~\ref{thm:2kgenerator}, we close this conceptual gap, by showing that the Scrooge behavior does hold true for the actual \textit{normalized} projected ensemble, i.e.,
\begin{equation}
    {\rho}_\mathcal{E}^{(k)} \approx \sum_{z=1}^{D_B} \braket{z|\sigma_B|z} {\rho}_{\text{Scrooge}}^{(k)}(\hat{\sigma}_{A|z}).
\end{equation}
This is stated informally in Proposition~\ref{prop:scrooge_from_temporal} of the main text, which we reproduce here, stated more formally.
\begin{theorem}[Emergent Scrooge designs from late-time generator states]
Let $\ket{\Psi}_{AB}$ be the generator state drawn from the random phase ensemble~\eqref{eq:randomphase_def_app}. Denote the reduced state of $\sigma_\text{diag}$ on $A$ and $B$ by $\sigma_A$ and $\sigma_B$ respectively, where $\sigma_{\text{diag}}$ is the diagonal ensemble given in Eq.~\eqref{eq:app_diagonalens}. Consider the projected ensemble $\mathcal{E}(\Psi)$ obtained by applying projective measurements on $B$ in an arbitrary orthonormal basis $\{\ket{z}\}_{z=1}^{D_B}$. Then, assuming $k^2 \norm{\hat{\sigma}_{A|z}}_2 \ll 1$,
\begin{equation}
    \E_\Psi\norm{{\rho}_\mathcal{E}^{(k)} - \sum_{z=1}^{D_B} \braket{z|\sigma_B|z} {\rho}_{\text{Scrooge}}^{(k)}(\hat{\sigma}_{A|z})}_1 \leq \sum_{z=1}^{D_B}\braket{z|\sigma_B|z} \bigO{k \norm{\hat{\sigma}_{A|z}}_2} + \bigO{\Delta_\beta^{1/2}},
\end{equation}
where $\sigma_{A|z} \equiv (I_A \otimes \bra{z})\sigma(I_A \otimes \ket{z})$, and $\hat{\sigma}_{A|z} =  \sigma_{A|z}/\braket{z|\sigma_B|z}$ is the normalized conditional mixed state on $A$. Above,
\begin{equation}
    \Delta_\beta = \sum_{z=1}^{D_B} \frac{ \braket{z^{\otimes 2}|\Tr_A \parens{\sigma_{\text{diag}}^{(2)}} | z^{\otimes 2}}}{\braket{z|\sigma_B|z}},
\end{equation}
where
\begin{equation}
    \sigma_{\text{diag}}^{(2)} = \sum_j |\braket{E_j|\Psi_0}|^4 \ket{E_j^{\otimes 2}}\bra{E_j^{\otimes 2}}.
\end{equation}

\end{theorem}
\begin{proof}
We first bound the average trace distance between $\tilde{\rho}_{\mathcal{E}}^{(k)}$ and $\rho_{\mathcal{E}}^{(k)}$. Define $X = p_z / \E_\varphi p_z$, where $\E_\varphi$ indicates the averaging over the uniformly random phases in Eq.~\eqref{eq:app_randomphase}, equivalent to averaging over the generator states $\E_\Psi$. Averaging over the random phase ensemble,
\begin{equation}
\begin{aligned}
    \E_\varphi \norm{\tilde{\rho}_{\mathcal{E}}^{(k)} - \rho_{\mathcal{E}}^{(k)}}_1 &\leq \E_{\varphi} \sum_{z=1}^{D_B} p_z \abs{1 - X^{k-1}} \\
    &= \E_\varphi \sum_{z=1}^{D_B} \abs{(1-X)(1+X+\ldots + X^{k-2}} \\
    &\leq \sum_{z=1}^{D_B} \parens{\E_{\varphi} p_z^2}^{1/2} \parens{\E_\varphi(1-X)^2}^{1/2} \parens{\E_\varphi\parens{1+X+\ldots+X^{k-2}}^{2}}^{1/2},
\end{aligned}
\end{equation}
Using 
\begin{equation}
\begin{aligned}
    \E_\varphi X^k &\leq \braket{z|\sigma_B|z}^{\change{-k}} \sum_{\pi \in S_k} \Tr\sparens{\sigma^{\otimes k} (I_A \otimes \ket{z}\bra{z})^{\otimes k} \hat{\pi}} \\
    &= \braket{z|\sigma_B|z}^{\change{-k}} \sum_{\pi \in S_k} \Tr\parens{\sigma_{A|z}^{\otimes k} \hat{\pi}_A} \\
    &= \sum_{\pi \in S_k} \Tr\parens{\hat{\sigma}_{A|z}^{\otimes k} \hat{\pi}_A} \\
    &= 1 + \bigO{k^2 \norm{\hat{\sigma}_{A|z}}_2^2} \quad &\text{(Lemma~\ref{lemma:mom_bound})},
\end{aligned}
\end{equation}
we have
\begin{equation}
\begin{aligned}
    \E_\varphi \norm{\tilde{\rho}_{\mathcal{E}}^{(k)} - \rho_{\mathcal{E}}^{(k)}}_1 \leq \sum_{z=1}^{D_B}  \braket{z|\sigma_B|z} \bigO{k \norm{\hat{\sigma}_{A|z}}_2}.
\end{aligned}
\end{equation}
This implies that $\tilde{\rho}_\mathcal{E}^{(k)} \approx \rho_\mathcal{E}^{(k)}$ in the low-purity regime. Similarly, from Lemma~\ref{lemma:scrooge_approx}, we have
\begin{equation}
    \norm{\sum_{z=1}^{D_B} \braket{z|\sigma_B|z} \parens{\tilde{\rho}_\text{Scrooge}^{(k)}(\hat{\sigma}_{A|z}) - {\rho}_\text{Scrooge}^{(k)}(\hat{\sigma}_{A|z})}}_1 \leq \sum_{z=1}^{D_B}\braket{z|\sigma_B|z} \bigO{k \norm{\hat{\sigma}_{A|z}}_2}.
\end{equation}
Thus, in the low-purity regime, we have
\begin{equation}
    \E_\varphi\norm{{\rho}_\mathcal{E}^{(k)} - \sum_{z=1}^{D_B} \braket{z|\sigma_B|z} {\rho}_{\text{Scrooge}}^{(k)}(\hat{\sigma}_{A|z})}_1 \leq \bigO{\Delta_\beta^{1/2}} + \sum_{z=1}^{D_B}\braket{z|\sigma_B|z} \bigO{k \norm{\hat{\sigma}_{A|z}}_2},
\end{equation}
where $\Delta_\beta$ is the error term defined in Ref.~\cite{mark2024maximum}, which is argued to be exponentially small in system size in typical many-body systems. The subscript $\beta$ indicates that $\ket{\Psi}$ has an effective temperature $\beta^{-1}$.
\end{proof}

\section{Projected ensemble generated by measurements in a scrambled basis}
\label{app:projens_2kmeasbasis}
In this Appendix, we prove Theorem~\ref{thm:ScroogeByMeasBasis} in the main text. We first prove the following lemmas, which will be useful in proving Theorem~\ref{thm:ScroogeByMeasBasis}.
\begin{lemma}\label{lemma:unnormalized_projens}[$k$th moment of unnormalized projected ensemble]
Let $\ket{\Psi}_{AB}$ be an arbitrary bipartite state, with subsystem density operators denoted $\sigma_A$ and $\sigma_B$ respectively. Consider the projected ensemble $\mathcal{E}$ obtained by applying a Haar random unitary $U$ on subsystem $B$, followed by projective measurements on $B$ in an arbitrary orthonormal basis $\{\ket{z}\}_{z=1}^{D_B}$. The average $k$th moment of the unnormalized state $\ket{\tilde{\psi}_z} = (I_A \otimes \bra{z}U)\ket{\Psi}$ is given by
\begin{equation}
    \Eset{U \sim \text{Haar}(D_B)} \parens{\ket{\tilde{\psi}_z}\bra{\tilde{\psi}_z}}^{\otimes k} = \frac{D_{A,k}}{D_{B,k}} \sigma_A^{\otimes k} \rho_{\text{Haar},A}^{(k)}.
\end{equation}
\begin{proof}
Let us write the Schmidt decomposition of $\ket{\Psi}_{AB}$ as
\begin{equation}
    \ket{\Psi}_{AB} = \sum_i \sqrt{\lambda_i} \ket{i}_A \otimes \ket{i}_B,
\end{equation}
where $\lambda_i$ are the Schmidt coefficients, and $|i\rangle$ are the corresponding Schmidt vectors. Then, we have
\begin{equation}
\begin{aligned}
    \Eset{U \sim \text{Haar}(D_B)} \parens{\ket{\tilde{\psi}_z}\bra{\tilde{\psi}_z}}^{\otimes k} &= \sum_{\substack{i_1,\ldots,i_k \\ j_1,\ldots,j_k}} \parens{\prod_{\alpha=1}^{k} \lambda_{i_\alpha}\lambda_{j_\alpha}}^{1/2} \parens{\ket{i_1,\ldots,i_k}\bra{j_1,\ldots,j_k}}_A \braket{j_1,\ldots,j_k|\rho_{\text{Haar},B}^{(k)}|i_1,\ldots,i_k}_B \\
    &= \frac{1}{k! D_{B,k}} \sum_{\substack{i_1,\ldots,i_k \\ j_1,\ldots,j_k}} \sum_{\pi \in S_k} \parens{\prod_{\alpha=1}^{k} \lambda_{i_\alpha}\lambda_{j_\alpha}}^{1/2} \parens{\ket{i_1,\ldots,i_k}\bra{j_1,\ldots,j_k}}_A \braket{j_1,\ldots,j_k|\hat{\pi}_B|i_1,\ldots,i_k}_B.
\end{aligned}   
\end{equation}
By orthonormality of the Schmidt vectors, the term $\braket{j_1,\ldots,j_k|\hat{\pi}_B|i_1,\ldots,i_k}_B = \delta_{j_1,\pi(i_1)} \ldots \delta_{j_k,\pi(i_k)}$ enforces the constraints on $j_1,\ldots,j_k$. This simplifies the sum to
\begin{equation}
\begin{aligned}
    \Eset{U \sim \text{Haar}(D_B)} \parens{\ket{\tilde{\psi}_z}\bra{\tilde{\psi}_z}}^{\otimes k}
    &= \frac{1}{k! D_{B,k}} \sum_{\substack{i_1,\ldots,i_k}} \sum_{\pi \in S_k} \parens{\prod_{\alpha=1}^{k} \lambda_{i_\alpha}} \parens{\ket{i_1,\ldots,i_k}\bra{i_1,\ldots,i_k}}_A \hat{\pi}_A^\dag \\
    &= \frac{D_{A,k}}{D_{B,k}} \sigma_A^{\otimes k} \rho_{\text{Haar},A}^{(k)},
\end{aligned}  
\end{equation}
where we used $\sigma_A = \sum_i \lambda_i \parens{\change{\ket{i}\bra{i}}}_A$.
\end{proof}
\end{lemma}

\begin{lemma}[Average $k$-copy overlap between projected states]\label{lemma:unnormalized_overlaps}\normalfont Let $\ket{\Psi}_{AB}$ be an arbitrary bipartite state, with subsystem density operators denoted $\sigma_A$ and $\sigma_B$ respectively. Denote the unnormalized projected state as $\ket{\tilde{\psi}_z} = (I_A \otimes \bra{z}U)\ket{\Psi}$. Assume $k^2 \ll D_B$ and $k \norm{\sigma_A}_4 \ll \norm{\sigma_A}_2$. Then, for any pair of measurement outcomes $z \neq z^\prime$,
\begin{equation}
    \Eset{U \sim \text{Haar}(D_B)} \abs{\braket{\tilde{\psi}_z|\tilde{\psi}_{z^\prime}}}^{2k} \leq \frac{k! \norm{\sigma_A}_2^{2k}}{D_B^{2k}} \sparens{1 + \bigO{k^2 \frac{\norm{\sigma_A}_4^4}{\norm{\sigma_A}_2^4}}} + \bigO{\frac{k^{2k+2}}{D_B^{2k+1}}}.
\end{equation}   
\end{lemma}
\begin{proof}
Using Weingarten calculus, we obtain
\begin{equation}
\begin{aligned}
\Eset{U \sim \text{Haar}(D_B)}\abs{\braket{\tilde{\psi}_z|\tilde{\psi}_{z^\prime}}}^{2k} &= \Eset{U \sim \text{Haar}(D_B)} \Tr \parens{\sigma_B^{\otimes 2k} U_B^{\dag \otimes 2k} \ket{z^{\otimes k} z^{\prime \otimes k}}\bra{z^{\prime \otimes k}z^{\otimes k}} U_B^{\otimes 2k}} \\
&= \sum_{\pi,\chi \in S_{2k}} \text{Wg}(\chi^{-1}\pi,D_B) \Tr\parens{\sigma_B^{\otimes 2k} \hat{\pi}_B^\dag} \braket{z^{\prime \otimes k} z^{\otimes k}|\hat{\chi}_B|z^{\otimes k} z^{\prime \otimes k}}
\end{aligned}
\end{equation}
in terms of permutation operators $\hat{\pi}_B$ and $\hat{\chi}_B$, and the Weingarten function $\text{Wg}(\chi^{-1}\change{\pi},D_B)$. To proceed, observe that by orthogonality of $\ket{z}$ and $\ket{z^\prime}$, the non-vanishing contributions to the Weingarten sum are those with $\chi$ of the form $\chi = \chi_1 \chi_2 \tau$, where $\chi_1, \chi_2 \in S_k$ are arbitrary permutations acting on the first and second set of $k$-copy replicas, and $\tau \in S_{2k}$ is a fixed permutation which swaps between the first and second set of $k$-copy replicas. This will contract every $\bra{z}$ with $\ket{z}$, and every $\bra{z^\prime}$ with $\ket{z^\prime}$, giving $\braket{z^{\prime \otimes k} z^{\otimes k}|\hat{\chi}_B|z^{\otimes k} z^{\prime \otimes k}} = 1$. Thus, we rewrite the sum as
\begin{equation}
\Eset{U \sim \text{Haar}(D_B)}\abs{\braket{\tilde{\psi}_z|\tilde{\psi}_{z^\prime}}}^{2k} = \sum_{\substack{\pi \in S_{2k} \\ \chi_1,\chi_2 \in S_k}} \text{Wg}((\chi_1\chi_2\tau)^{-1}\pi,D_B) \Tr\parens{\sigma_B^{\otimes 2k} \hat{\pi}_B^\dag}.   
\end{equation}
Next, we split the sum over $\pi \in S_{2k}$ into terms where $\pi = \chi_1 \chi_2 \tau$ and $\pi \neq \chi_1 \chi_2 \tau$, respectively. This gives
\begin{equation}
\begin{aligned}
\Eset{U \sim \text{Haar}(D_B)}\abs{\braket{\tilde{\psi}_z|\tilde{\psi}_{z^\prime}}}^{2k} &= \text{Wg}(1_{2k},D_B) \sum_{\chi_1,\chi_2 \in S_k} \Tr \parens{{\sigma_B}^{\otimes 2k} \parens{\hat{\chi}_{1,B} \otimes \hat{\chi}_{2,B}} \hat{\tau}_B} + \sum_{\substack{\pi \in S_{2k} \\ \chi_1,\chi_2 \in S_k \\ \pi \neq \chi_1\chi_2\tau}}  \text{Wg}((\chi_1\chi_2\tau)^{-1}\pi,D_B) \Tr\parens{\sigma_B^{\otimes 2k} \hat{\pi}_B^\dag} \\
&\leq \frac{1}{D_B^{2k}}\sparens{1 + \bigO{\frac{k^{7/4}}{D_B^2}}} k! \sum_{\chi \in S_k} \Tr \parens{{\sigma_B^{2}}^{\otimes k} \hat{\chi}_{B}} + \sum_{\substack{\pi \in S_{2k} \\ \chi_1,\chi_\change{2} \in S_k \\ \pi \neq \chi_1\chi_2\tau}} \abs{\text{Wg}((\chi_1\chi_2\tau)^{-1}\pi,D_B)}
\end{aligned}
\label{eq:upperbound_overlap}
\end{equation}
where we used~\cite{collins2017weingarten}
\begin{equation}
    \text{Wg}(1_{2k},D_B) = \frac{1}{D_B^{2k}}\sparens{1 + \bigO{\frac{k^{7/4}}{D_B^2}}}
\end{equation}
valid for $4k^2 < D_B$.
Applying Lemma~\ref{lemma:lowpurity_bounds},
\begin{equation}
\sum_{\chi \in S_k}\Tr \parens{{\sigma_B^{2}}^{\otimes k} \hat{\chi}_{B}} = \Tr^{k}\parens{\sigma_B^2} \sum_{\chi \in S_k} \Tr\parens{\frac{{\sigma_B^2}^{\otimes k}}{\Tr^{k}\parens{\sigma_B^2}} \hat{\chi}_B} = \norm{\sigma_A}_2^{2k} \sparens{1 + \bigO{k^2 \frac{\norm{\sigma_A}_4^4}{\norm{\sigma_A}_2^4}}},
\end{equation}
using the fact that $\sigma_A$ and $\sigma_B$ share the same non-zero eigenvalues. To bound the last term in Eq.~\eqref{eq:upperbound_overlap}, we use
\begin{equation}
\begin{aligned}
\sum_{\substack{\pi \in S_{2k} \\ \chi_1,\chi_\change{2} \in S_k \\ \pi \neq \chi_1\chi_2\tau}} \abs{\text{Wg}((\chi_1\chi_2\tau)^{-1}\pi,D_B)} &= \sum_{\substack{\pi \in S_{2k} \\ \chi_1,\chi_2 \in S_k \\ \pi \neq 1_{2k}}} |\text{Wg}(\pi,D_B)|
\\ &= (k!)^2 \parens{\sum_{\pi \in S_{2k}}\abs{\text{Wg}(\pi,D_B)} - \abs{\text{Wg}(1_{2k},D_B)}} \\
&= (k!)^2 \bparens{\frac{(D_B - 2k)!}{D_B !} - \frac{1}{D_B^{2k}}\sparens{1 + \bigO{\frac{k^{7/4}}{D_B^2}}} } \\
&= (k!)^2 \bparens{\frac{1}{D_B^{2k}}\sparens{1 + \bigO{\frac{k^2}{D_B}}} - \frac{1}{D_B^{2k}}\sparens{1 + \bigO{\frac{k^{7/4}}{D_B^2}}}} \\
&= \bigO{\frac{k^{2k+2}}{D_B^{2k+1}}}\,.
\end{aligned}
\end{equation}
Substituting these into Eq.~\eqref{eq:upperbound_overlap} yields the desired result.
\end{proof}
Now, we are ready to prove Theorem~\ref{thm:ScroogeByMeasBasis} in the main text, which we reproduce here for convenience.

\begin{theorem}\normalfont\label{thm:2kdesign_meas}
Let $\ket{\Psi}_{AB}$ be an arbitrary bipartite state, with subsystem density operators denoted $\sigma_A$ and $\sigma_B$ respectively. Consider the projected ensemble $\mathcal{E}(\Psi)$ obtained by applying a unitary $U_B$, drawn from an approximate unitary $2k$-design with relative error $\epsilon$, on subsystem $B$, followed by projective measurements on $B$ in an arbitrary orthonormal basis $\{\ket{z}\}_{z=1}^{D_B}$. Then, assuming that $k \ll \norm{\sigma_A}_2/\norm{\sigma_A}_4$ and $1 \ll D_A \leq D_B$,
\begin{equation}
    \E_{U_B}\norm{\rho_\mathcal{E}^{(k)} - \rho_{\text{Scrooge}}^{(k)}(\sigma_A)}_1 \leq \sparens{(D_A \norm{\sigma_A}_2^2)^k \bigO{k^2 \frac{\norm{\sigma_A}_4^4}{\norm{\sigma_A}_2^4}} + \bigO{\frac{D_A^k k^{k+2}}{D_B}} \change{+ \bigO{D_{A,k} \epsilon}}}^{1/2}.
\end{equation}
\end{theorem}

\begin{proof}
The $k$th moment of the projected ensemble $\mathcal{E}$ generated by $\ket{\Psi_U}_{AB} = (I_A \otimes U_B)\ket{\Psi}_{AB}$ is denoted as
\begin{equation}
    \rho_{\mathcal{E}}^{(k)} = \sum_{z=1}^{D_B} p_z^{1-k} [(I_A \otimes \bra{z})\ket{\Psi_U}\bra{\Psi_U}(I_A \otimes \ket{z})]^{\otimes k},
\end{equation}
where $p_z = \braket{\Psi_U|(I_A \otimes \ket{z}\bra{z})|\Psi_U}$ is the probability of measuring the outcome $z$ on subsystem $B$. Since $U$ is sampled from an approximate $2k$-design with relative error $\epsilon$, the generator states $\ket{\Psi_U}$ satisfy
\begin{equation}
    (1-\epsilon)\Eset{U \sim \text{Haar}(D_B)} \parens{\ket{\Psi_U}\bra{\Psi_U}}^{\otimes \change{2k}} \preceq \Eset{U \sim \text{$2k$-design}} \parens{\ket{\Psi_U}\bra{\Psi_U}}^{\otimes \change{2k}} \preceq (1+\epsilon)\Eset{U \sim \text{Haar}(D_B)} \parens{\ket{\Psi_U}\bra{\Psi_U}}^{\otimes \change{2k}}.
\end{equation}
Now, using Lemma~\ref{lemma:projens_relerrorgens}, with the choice
\begin{equation}
    \tilde{\rho}_\mathcal{E}^{(k)} = {D_B}^{k-1} \sum_{z=1}^{D_B} [(I_A \otimes \bra{z})\ket{\Psi_U}\bra{\Psi_U}(I_A \otimes \ket{z})]^{\otimes k}
\end{equation}
and
\begin{equation}
    M = \tilde{\rho}_{\text{Scrooge}}^{(k)}(\sigma_A) = {D_A}^k \sigma_A^{\otimes k} \rho_{\text{Haar},A}^{(k)},
\end{equation}
we have
\begin{equation}
\begin{aligned}
    \Delta &= \abs{\Eset{U \sim \text{$2k$-design}} \norm{\tilde{\rho}_\mathcal{E}^{(k)} - \tilde{\rho}_{\text{Scrooge}}^{(k)}(\sigma_A)}_2^2 - \Eset{U \sim \text{Haar}} \norm{\tilde{\rho}_\mathcal{E}^{(k)} - \tilde{\rho}_{\text{Scrooge}}^{(k)}(\sigma_A)}_2^2} \\
    &\leq 2\epsilon \Eset{U \sim \text{Haar}(D_B)} \Tr \sparens{\tilde{\rho}_{\mathcal{E}}^{(k)} \tilde{\rho}_{\text{Scrooge}}^{(k)}(\sigma_A)} \change{+} \epsilon {D_B}^{2k-2} \Eset{U \sim \text{Haar}(D_B)} \parens{\sum_{z=1}^{D_B} p_z^k}^2.
\end{aligned}
\end{equation}
With
\begin{equation}
\begin{aligned}
    \Eset{U \sim \text{Haar}(D_B)} \Tr \sparens{\tilde{\rho}_{\mathcal{E}}^{(k)} \tilde{\rho}_{\text{Scrooge}}^{(k)}(\sigma_A)} &= \frac{{D_B}^k D_{A,k}}{D_{B,k}} \Tr\parens{\sigma_A^{\otimes k} \rho_{\text{Haar},A}^{(k)} \tilde{\rho}_{\text{Scrooge}}^{(k)}(\sigma_A)} \quad &\text{(Lemma~\ref{lemma:unnormalized_projens})} \\
    &= \frac{{D_B}^k}{D_{B,k}} \Tr \parens{\sigma_A^{\otimes k} \tilde{\rho}_{\text{Scrooge}}^{(k)}(\sigma_A)} \\
    &= \frac{{D_B}^k {D_A}^k}{D_{B,k}} \Tr\parens{{\sigma_A^2}^{\otimes k} \rho_{\text{Haar},A}^{(k)}} \\
    &= k!\parens{1 + \bigO{\frac{k^2}{D_A}}} \norm{\sigma_A}_2^{2k} \parens{1 + \bigO{k^2 \frac{\norm{\sigma_A}_4^4}{\norm{\sigma_A}_2^4}}} \quad &\text{(Lemma~\ref{lemma:lowpurity_bounds})} \\
    &= k! \norm{\sigma_A}_2^{2k} \parens{1 + \bigO{k^2 \frac{\norm{\sigma_A}_4^4}{\norm{\sigma_A}_2^4}}} \quad &\text{($\norm{\sigma_A}_4^4/\norm{\sigma_A}_2^4 \geq 1/D_A$)},
\end{aligned}
\end{equation}
and, \change{using the Cauchy-Schwarz inequality},
\change{
\begin{equation}
    \Eset{U \sim \text{Haar}} \parens{\sum_{z=1}^{D_B} p_z^k}^2 \leq \sum_{z=1}^{D_B} \sparens{\parens{\Eset{U \sim \text{Haar}} {p_z}^{2k}}^{1/2}}^{2} = \frac{1}{{D_B}^{2k-2}} \parens{1+\bigO{k^2\norm{\sigma_A}_2^2}}.
\end{equation}
}
Thus,
\begin{equation}
\begin{aligned}
    \Delta &\leq 2\epsilon k! \norm{\sigma_A}_2^{2k} \parens{1 + \bigO{k^2\frac{\norm{\sigma_A}_4^4}{\norm{\sigma_A}_2^4}}} + \epsilon \change{\parens{1 + \bigO{k^2\norm{\sigma_A}_2^2}} = \bigO{\epsilon}}.
\end{aligned}
\end{equation}
This implies that
\begin{equation}
    \Eset{U \sim \text{$2k$-design}} \norm{\tilde{\rho}_\mathcal{E}^{(k)} - \tilde{\rho}_{\text{Scrooge}}^{(k)}(\sigma_A)}_2^2 \leq \Eset{U \sim \text{Haar}(D_B)} \norm{\tilde{\rho}_\mathcal{E}^{(k)} - \tilde{\rho}_{\text{Scrooge}}^{(k)}(\sigma_A)}_2^2 + \change{\bigO{\epsilon}}.
\end{equation}
To proceed, we need to evaluate
\begin{equation}
\begin{aligned}
    \Eset{U \sim \text{Haar}(D_B)} \norm{\tilde{\rho}_\mathcal{E}^{(k)} - \tilde{\rho}_{\text{Scrooge}}^{(k)}(\sigma_A)}_2^2 &= \Eset{U \sim \text{Haar}(D_B)} \Tr \tilde{\rho}_{\mathcal{E}}^{(k) 2} - 2 \Eset{U \sim \text{Haar}(D_B)} \Tr \sparens{\tilde{\rho}_{\mathcal{E}}^{(k)} \tilde{\rho}_{\text{Scrooge}}^{(k)}(\sigma_A)} + \Tr \tilde{\rho}_{\text{Scrooge}}^{(k) 2} (\sigma_A) \\
    &= k! \norm{\change{\sigma_A}}_2^{2k} \bigO{k^2 \frac{\norm{\sigma_A}_4^4}{\norm{\sigma_A}_2^4}} + \bigO{\frac{k^{2k+2}}{D_B}}.
\end{aligned}
\end{equation}
To obtain this, we used
\begin{equation}
    \Eset{U \sim \text{Haar}(D_B)} \Tr \tilde{\rho}_{\mathcal{E}}^{(k) 2} = \frac{1}{D_B}\parens{1 + \bigO{k^2 \norm{\sigma_A}_2^2 + k^{2k+2}}} + k! \norm{\sigma_A}_2^{2k} \parens{1 + \bigO{k^2 \frac{\norm{\sigma_A}_4^4}{\norm{\sigma_A}_2^4}}}
\end{equation}
from Lemma~\ref{lemma:unnormalized_overlaps},
and
\begin{equation}
\begin{aligned}
    \Tr \tilde{\rho}_{\text{Scrooge}}^{(k) 2}(\sigma_A) &= D_A^{2k} \Tr\parens{\sigma_A^{\otimes k} \rho_{\text{Haar},A}^{(k)}\sigma_A^{\otimes k} \rho_{\text{Haar},A}^{(k)}} \\
    &= \frac{D_A^{2k}}{D_{A,k}} \Tr\parens{{\sigma_A^2}^{\otimes k} \rho_{\text{Haar},A}^{(k)}} \\
    &= k! \norm{\sigma_A}_2^{2k} \parens{1 + \bigO{k^2 \frac{\norm{\sigma_A}_4^4}{\norm{\sigma_A}_2^4}}} \quad &\text{(Lemma~\ref{lemma:lowpurity_bounds})}.
\end{aligned}
\end{equation}
Thus,
\begin{equation}
\begin{aligned}
    &\Eset{U \sim \change{\text{$2k$-design}}} \norm{\tilde{\rho}_\mathcal{E}^{(k)} - \tilde{\rho}_{\text{Scrooge}}^{(k)}(\sigma_A)}_2^2 \leq \change{\bigO{\frac{k^{2k+2}}{D_B}}} + k! \norm{\sigma_A}_2^{2k} \bigO{k^2 \frac{\norm{\sigma_A}_4^4}{\norm{\sigma_A}_2^4}} + \change{\bigO{\epsilon}}.
\end{aligned}
\end{equation}
This can be used to upper bound the trace distance, via
\begin{equation}
\begin{aligned}
    \parens{\Eset{U \sim \change{\text{$2k$-design}}} \norm{\tilde{\rho}_\mathcal{E}^{(k)} - \tilde{\rho}_{\text{Scrooge}}^{(k)}(\sigma_A)}_1}^2 &\leq D_{A,k} \Eset{U \sim \change{\text{$2k$-design}}} \norm{\tilde{\rho}_\mathcal{E}^{(k)} - \tilde{\rho}_{\text{Scrooge}}^{(k)}(\sigma_A)}_2^2 \\
    &= (D_A \norm{\sigma_A}_2^2)^k \bigO{k^2 \frac{\norm{\sigma_A}_4^4}{\norm{\sigma_A}_2^4}} + \bigO{\frac{D_A^k k^{k+2}}{D_B}} \change{+ \bigO{D_{A,k} \epsilon}}.
\end{aligned}
\end{equation}
The next step is to bound the average trace distance between $\rho_{\mathcal{E}}^{(k)}$ and $\tilde{\rho}_{\mathcal{E}}^{(k)}$, and the trace distance between $\rho_{\text{Scrooge}}^{(k)}(\sigma_A)$ and $\tilde{\rho}_{\text{Scrooge}}^{(k)}(\change{\sigma_A})$. We have, from Lemma~\ref{lemma:general_approx},
\begin{equation}
\begin{aligned}
    \Eset{U \sim \text{$2k$-design}} \norm{\rho_\mathcal{E}^{(k)} - \tilde{\rho}_{\mathcal{E}}^{(k)}}_1 \leq \sum_{z=1}^{D_B} \parens{\Eset{U \sim \text{$2k$-design}} p_z^2}^{1/2} \parens{1 - 2 D_B^{k-1} \Eset{U \sim \text{$2k$-design}} p_z^{k-1} + D_B^{2k-2} \Eset{U \sim \text{$2k$-design}} p_z^{2k-2}}^{1/2}.
\end{aligned}
\end{equation}
Using the fact that
\begin{equation}
\begin{aligned}
    \Eset{U \sim \text{Haar}(D_B)} p_z^k &= \Eset{U \sim \text{Haar}(D_B)} \braket{z|U \sigma_B U^\dag|z}^k \\
    &= \Tr \parens{\sigma_B^{\otimes k} \rho_{\text{Haar},B}^{(k)}} \\
    &= \frac{1}{k! D_{B,k}} \parens{1 + \bigO{k^2 \norm{\sigma_B}_2^2}} \quad &\text{(Lemma~\ref{lemma:lowpurity_bounds})} \\
    &= \frac{1}{D_B^k} \parens{1 + \bigO{\frac{k^2}{D_B}}}\parens{1 + \bigO{k^2 \norm{\sigma_A}_2^2}} \\
    &= \frac{1}{D_B^k} \parens{1 + \bigO{k^2 \norm{\sigma_A}_2^2}} \quad &\text{($\norm{\sigma_B}_2^2 = \norm{\sigma_A}_2^2 \geq 1/D_B$)},
\end{aligned}
\end{equation}
and
\begin{equation}
    (1-\epsilon)\Eset{U \sim \text{Haar}(D_B)} p_z^\ell \leq \Eset{U \sim \text{$2k$-design}} p_z^{\ell} \leq (1+\epsilon)\Eset{U \sim \text{Haar}(D_B)} p_z^\ell
\end{equation}
for all $1 \leq \ell \leq 2k$, we get
\begin{equation}
    \Eset{U \sim \text{$2k$-design}} \norm{\rho_\mathcal{E}^{(k)} - \tilde{\rho}_{\mathcal{E}}^{(k)}}_1 \leq \bigO{\sqrt{\epsilon + k^2 \norm{\sigma_A}_2^2}}.
\end{equation}
From Lemma~\ref{lemma:app_scrooge_approx},
\begin{equation}
    \norm{\rho_{\text{Scrooge}}^{(k)}(\sigma_A) - \tilde{\rho}_{\text{Scrooge}}^{(k)}(\sigma_A)}_1 \leq \bigO{k \norm{\sigma_A}_2}.
\end{equation}
Therefore, by the triangle inequality,
\begin{equation}
\begin{aligned}
    &\Eset{U \sim {\text{$2k$-design}}}\norm{\rho_\mathcal{E}^{(k)} - \rho_{\text{Scrooge}}^{(k)}(\sigma_A)}_1 \\ &\leq  \Eset{U \sim \text{$2k$-design}} \norm{\rho_\mathcal{E}^{(k)} - \tilde{\rho}_{\mathcal{E}}^{(k)}}_1 + \norm{\rho_{\text{Scrooge}}^{(k)}(\sigma_A) - \tilde{\rho}_{\text{Scrooge}}^{(k)}(\sigma_A)}_1 + \Eset{U \sim \text{$2k$-design}}\norm{\tilde{\rho}_{\mathcal{E}}^{(k)} - \tilde{\rho}_{\text{Scrooge}}^{(k)}(\sigma_A)}_1 \\
    &\leq \bigO{\sqrt{\epsilon + k^2 \norm{\sigma_A}_2^2}} + \bigO{k \norm{\change{\sigma_A}}_2} + \sparens{(D_A \norm{\sigma_A}_2^2)^k \bigO{k^2 \frac{\norm{\sigma_A}_4^4}{\norm{\sigma_A}_2^4}} + \bigO{\frac{D_A^k k^{k+2}}{D_B}} \change{+ \bigO{D_{A,k} \epsilon}}}^{1/2} \\
    &= \sparens{(D_A \norm{\sigma_A}_2^2)^k \bigO{k^2 \frac{\norm{\sigma_A}_4^4}{\norm{\sigma_A}_2^4}} + \bigO{\frac{D_A^k k^{k+2}}{D_B}} \change{+ \bigO{D_{A,k} \epsilon}}}^{1/2},
\end{aligned}
\end{equation}
with the dominant contribution to the error bound coming from the average trace distance between $\tilde{\rho}_{\mathcal{E}}^{(k)}$ and $\tilde{\rho}_{\text{Scrooge}}^{(k)}(\sigma_A)$.
\end{proof}

\subsection{Application: Local Hamiltonian at finite temperatures}

Theorem~\ref{thm:ScroogeByMeasBasis} requires the condition $k \ll \norm{\sigma_A}_2 / \norm{\sigma_A}_4$, for the projected ensemble to converge to Scrooge($\sigma_A$). As stated in the main text, we can define the effective dimension of $\sigma_A$ via
\begin{equation}
    D_{A,\text{eff}} = \parens{\frac{\norm{\sigma_A}_2}{\norm{\sigma_A}_4}}^4,
\end{equation}
If $\sigma_A = I_A/D_A$ is the maximally mixed state, then $D_{A,\text{eff}} = D_A$. Then, the above condition reads $k^4 \ll D_{A,\text{eff}}$. 

This condition is usually satisfied in many-body quantum systems. To analyze a concrete example, let us consider the case where $\sigma_A$ is the thermal Gibbs state 
\begin{equation}
    \sigma_A = \frac{e^{-\beta H_A}}{\Tr \parens{e^{-\beta H_A}}},
\end{equation}
where $\beta$ is the inverse temperature, and $H_A$ is the Hamiltonian restricted to subsystem $A$.

Let us denote the spectral density of the Hamiltonian $H_A$ by $f(E)$, which satisfies the normalization
\begin{equation}
    \int_{-\infty}^{\infty} dE \, f(E) = 1.
\end{equation}
We assume that $H_A$ has a Gaussian spectral density, i.e.,
\begin{equation}
    f(E) = \frac{1}{\sqrt{2\pi \Delta^2}} \exp\parens{-\frac{(E-\mu)^2}{2\Delta^2}},
\end{equation}
where
\begin{equation}
    \mu = \frac{1}{D_A} \sum_{i=1}^{D_A} E_i = \frac{\Tr H_A}{D_A},
\end{equation}
is the mean energy, and
\begin{equation}
    \Delta^2 = \frac{1}{D_A} \sum_{i=1}^{D_A} E_i^2 - \mu^2 = \frac{\Tr H_A^2}{D_A} - \frac{\Tr^2 H_A}{D_A^2}
\end{equation}
is the variance. To convert the discrete sum over energies to an integral over the spectral density, we use $D_A^{-1} \sum_{i=1}^{D_A} \to \int dE\, f(E)$. Thus, the partition function at inverse temperature $\beta$ is given by
\begin{equation}
    Z_\beta = \Tr e^{-\beta H_A} = \sum_{i=1}^{D_A} e^{-\beta E_i} = D_A \int_{-\infty}^{\infty} dE\, f(E) e^{-\beta E} = D_A \exp\parens{-\beta \mu + \frac{1}{2}\beta^2 \Delta^2}.
\end{equation}
Note that this gives the partition function for any $H$ (regardless of spectral density), to quadratic order in $\beta$, i.e., in the high-temperature regime.

The generalized purities read
\begin{equation}
    \Tr \parens{{\sigma_A}^q} = \frac{Z_{q\beta}}{{Z_\beta}^q} = \frac{D_A \exp\parens{-q\beta \mu + \frac{1}{2} q^2 \beta^2 \Delta^2}}{D_A^q \exp\parens{-q \beta \mu + \frac{1}{2} q \beta^2 \Delta^2}} = D_A^{1-q} \exp\parens{\frac{1}{2}q(q-1)\beta^2 \Delta^2}.
\end{equation}
For the above equation to be self-consistent, we need to impose the constraint $\Tr({\sigma_A}^q) \leq 1$ which implies the constraint $\beta^2 \Delta^2 \leq (2/q) N_A \ln 2$. Outside of this range (e.g., at low temperatures), the expression becomes unphysical. This reflects the fact that the finiteness of the spectrum is relevant in the low-temperature regime, leading to the breakdown of the Gaussian approximation. In what follows, we will work in the self-consistent regime.

Let us expand $H_A$ in the Pauli basis, i.e.,
\begin{equation}
    H_A = \sum_{m=1}^{M} c_m P_m,
\end{equation}
where $P_m \neq I$ is a Pauli string on $N$ qubits, and $|c_m| \leq 1$ are real coefficients, as explained in Eq.~\eqref{eq:ham_A_pauliexpansion} of the main text. Then,
\begin{equation}
    \mu = \frac{\Tr H_A}{\change{D_A}} = 0,
\end{equation}
and
\begin{equation}
    \Delta^2 = \frac{\Tr H_A^2}{D_A} = \frac{1}{D_A} \sum_{m,n=1}^{M} c_m c_n \Tr\parens{P_m P_n} = \sum_{m=1}^{M} c_m^2 \leq M.
\end{equation}
For a geometrically local $H_A$, we expect $M \propto N_A$, thus $\Delta^2 = O(N_A)$. For simplicity, we set $M = \alpha N_A$, where $\alpha$ is some constant. Note that we are using dimensionless units here, with the characteristic energy scale given by the typical magnitude of $c_m$. Using this, we have
\begin{equation}
    \norm{\sigma_A}_2 = \frac{1}{D_A^{1/2}} e^{\beta^2 \Delta^2/2},
\end{equation}
and
\begin{equation}
    \norm{\sigma_A}_4 = \frac{1}{D_A^{3/4}} e^{3\beta^2 \Delta^2 /2},
\end{equation}
which yields
\begin{equation}
    \frac{\norm{\sigma_A}_4}{\norm{\sigma_A}_2} = \frac{1}{D_A^{1/4}} e^{\beta^2 \Delta^2} = \exp\parens{\beta^2 \Delta^2 - \frac{N_A}{4} \ln 2}.
\end{equation}
For $\beta < \beta_c$, where
\begin{equation}
    \beta_c = \sqrt{\frac{N_A}{4 \Delta^2} \ln 2} \geq \sqrt{\frac{N_A}{4 M} \ln 2} = \sqrt{\frac{\ln 2}{4\alpha}},
\end{equation}
the ratio $\norm{\sigma_A}_4/\norm{\sigma_A}_2$ is exponentially suppressed in $N_A$, and therefore $k \ll \norm{\sigma_A}_2/\norm{\sigma_A}_4$ is satisfied for large $N_A$ and any sub-exponential $k$. Finally, we remark that the value of $\beta_c$ lies within the range $\beta^2 \Delta^2 \leq N_A \ln 2 / 2$, imposed by self-consistency up to $q = 4$, as explained above.

\section{Uniform random phase states}\label{sec:HS}

Here, we present additional numerical results on the random phase states given by
\begin{align}\label{eq:randphase_sup}
\ket{\Psi}_{AB} & =  \frac{1}{2^{N/2}}\sum_{j}e^{-i\varphi_j}\ket{j}\,.
\end{align}
\begin{figure}[htbp]
	\centering	
\subfigimg[width=0.3\textwidth]{a}{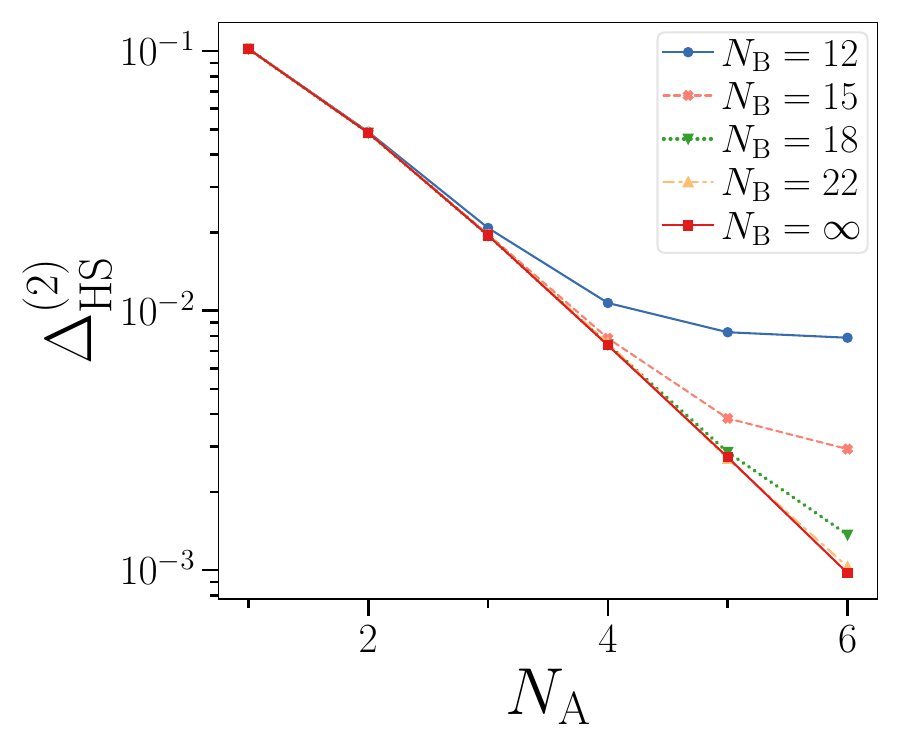}
        \subfigimg[width=0.3\textwidth]{b}{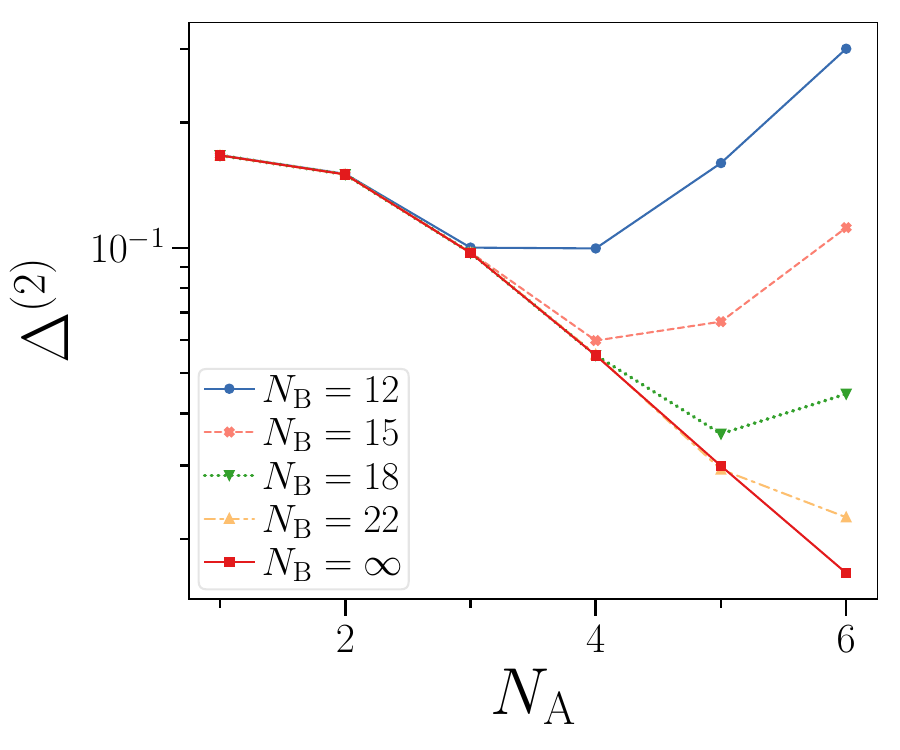}
    \caption{Hilbert-Schmidt distance to Haar $2$-design $\Delta_\text{HS}^{(2)}$ of the projected ensemble generated from uniform random phase states~\eqref{eq:randphase_sup}. We measure $N_\text{B}$ qubits in the computational basis ($\theta=0$), and show \idg{a} $\Delta_\text{HS}^{(2)}$ and \idg{b} $\Delta^{(2)}$ against $N_\text{A}$. 
	}
	\label{fig:randphase_HS}
\end{figure}
While in the main text we consider the trace distance $\Delta^{(k)}$ for our numerical studies, now we also regard a weaker notion of distance, namely the Hilbert-Schmidt distance. It is given by
\begin{equation}
    \Delta_\text{HS}^{(k)}(\sigma)=\frac{1}{2}\norm{\rho_{\mathcal{E}}^{(k)} - \rho_{\text{Scrooge}}^{(k)}(\sigma)}_2\,.
\end{equation}
This is a weaker notion of statistical closeness between the ensembles $\mathcal{E}$ and Scrooge($\sigma$) because
\begin{equation}
    \Delta_{\text{HS}}^{(k)} \leq \Delta^{(k)} \leq \sqrt{D_{A,k}} \Delta_{\text{HS}}^{(k)},
\end{equation}
thus a small $\Delta_{\text{HS}}^{(k)}$ does not guarantee a small $\Delta^{(k)}$.

We study random phase states with computational basis measurements ($\theta=0$). We plot the Hilbert-Schmidt distance $\Delta_{\text{HS}}^{(k)}$ in Fig.~\ref{fig:randphase_HS}a, and reproduce the trace distance from the main text as reference in Fig.~\ref{fig:randphase_HS}b. We find similar behavior in both cases, although the difference between finite $N_\text{B}$ and  $N_\text{B}\rightarrow \infty$ is less pronounced for $\Delta_\text{HS}^{(k)}$ compared to $\Delta^{(k)}$.

\section{T-doped Clifford circuits}\label{sec:Cliffordmagicdepth}
In this section, we study emergent Scrooge ensembles via rotating the measurement basis using Clifford circuits doped with T-gates.
In particular, we prepare the  entangled state 
\begin{equation}\label{eq:state_ent_sup}
    \ket{\Psi(\chi)}=\ket{\psi_\text{ent}(\chi)}^{\otimes N_\text{A}}\ket{0_{B_2}}^{\otimes N_\text{B}-N_\text{A}}
\end{equation}
with $\ket{\psi_\text{ent}(\chi)}=\cos(\chi/2)\ket{0_A 0_{B_1}}+\sin(\chi/2)\ket{1_A 1_{B_1}}$
between $A$ and $B_1$. We then apply a Clifford unitary of depth $d$ and $N_\text{T}$ gates on $B=B_1\cup B_2$ only, in the setting described in the main text. Now, in Fig.~\ref{fig:cliffdepth_B_sup}, we study the behavior as a function of depth $d$ and $N_\text{T}$ in more detail. 
First in Fig.~\ref{fig:cliffdepth_B_sup}a, we plot $\Delta^{(2)}$ against $d$ for different $N_\text{T}$. We find that $\Delta^{(2)}$ decreases with $d$, and converges to a limiting value for large $d$, indicating the need for scrambling via entangling gates to generate projected Scrooge ensembles. 
However, we find that the minimal $\Delta^{(2)}$ decreases with increasing number of T-gates $N_\text{T}$, demonstrating that magic is also necessary for Scrooge ensembles. Notably, there is a number of T-gates beyond which $\Delta^{(2)}$ no longer improves.%

In Fig.~\ref{fig:cliffdepth_B_sup}b, we show $\Delta^{(2)}$ against $d$ for different total qubit numbers $N$, where we choose $N_\text{T}=3N$, i.e. in the limit where magic is large such that we can converge to minimal $\Delta^{(2)}$. 
We find that for all $N$, the  decrease in $d$ is the same, until beyond a certain $d$, where $\Delta^{(2)}$ does not decrease further and plateaus.
The plateau value of $\Delta^{(2)}$ arises from finite-size effects, and decreases exponentially with $N$.
We refer to the depth $d$ where $\Delta^{(2)}$ stops decreasing as $d_0$. We find that $d_0$ increases with $N$, which we find approximately to be $d_0\approx \frac{2}{3} N$ from our numerical study. 
 
In Fig.~\ref{fig:cliffdepth_B_sup}c, we show $\Delta^{(2)}$ against T-gate density $N_\text{T}/N$ for deep circuits of depth $d=30$. We find that $\Delta^{(2)}$ decreases with $N_\text{T}/N$, converging to a $N$-dependent minimum. We observe that this minimum is reached approximately around $N_\text{T}\approx 2.5$, indicating a transition in magic when the system becomes fully random. This mirrors the saturation transition in magic observed in Ref.~\cite{haug2025probing}.

\begin{figure}[htbp]
	\centering	
        \subfigimg[width=0.3\textwidth]{a}{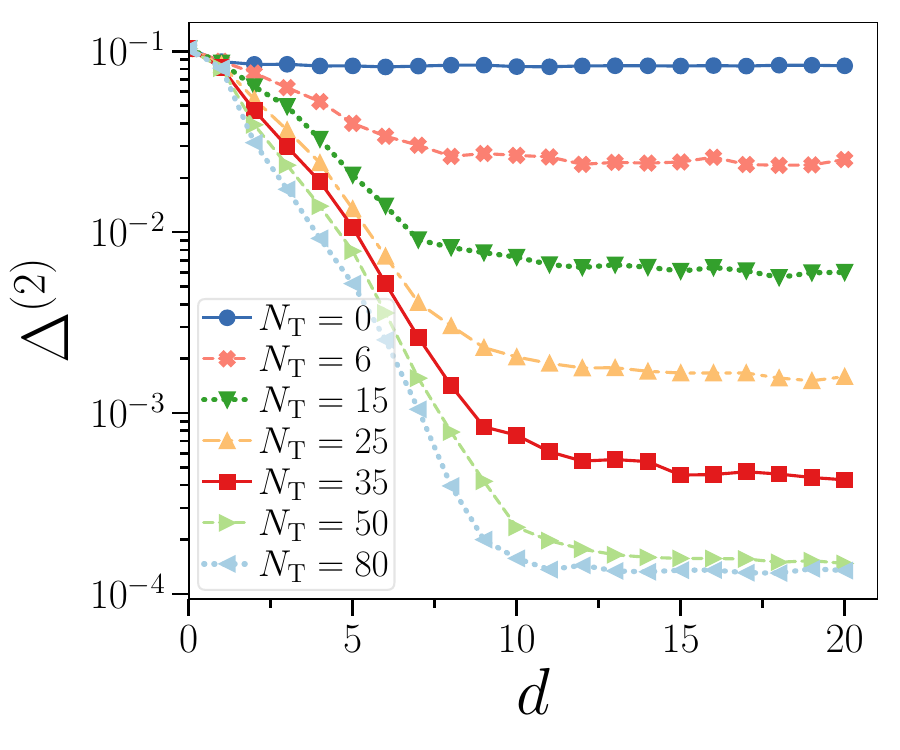}
        \subfigimg[width=0.3\textwidth]{b}{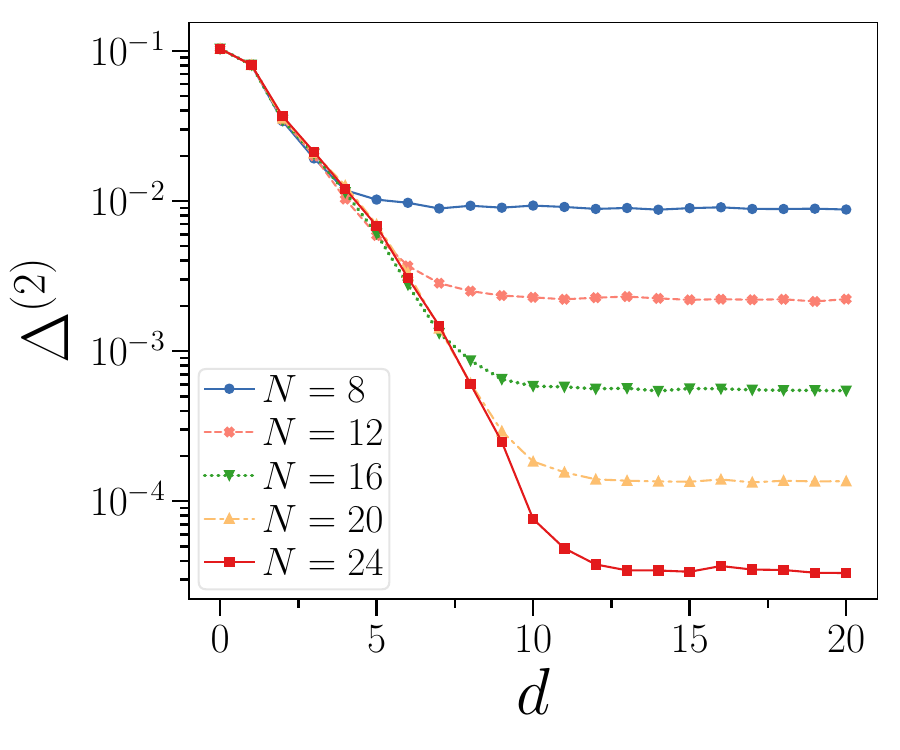}
        \subfigimg[width=0.3\textwidth]{c}{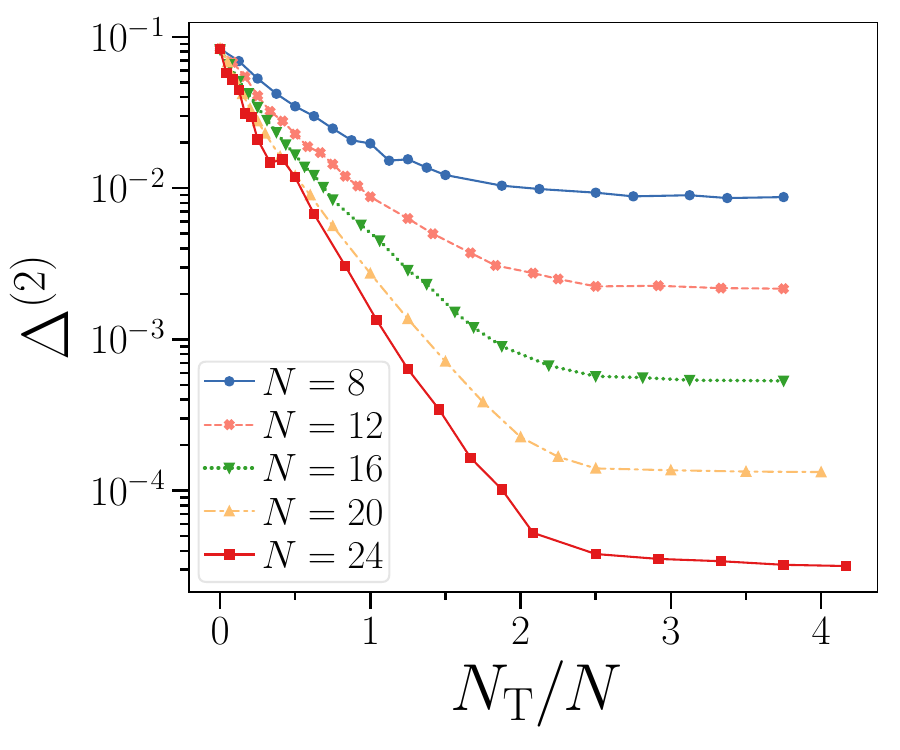}
        
    \caption{Trace distance to Scrooge $2$-design $\Delta^{(2)}$ of projected ensemble generated from $\ket{\Psi(\chi)}$ as defined in~\eqref{eq:state_ent_sup} where we choose $\chi=\pi/6$. We then measure on $B$  in a rotated basis via unitaries composed of $d$ layers of local Clifford gates doped with $N_\text{T}$ T-gates.
    \idg{a} $\Delta^{(2)}$ against $d$ for different $N_\text{T}$.    We have $N_A=1$, $N_B=19$, and average over 500 random realizations of the circuit.
    \idg{b} $\Delta^{(2)}$ against $d$ for different total qubit numbers $N$ and fixed $N_\text{T}=3N$.
    \idg{c} $\Delta^{(2)}$ against T-gate density $N_{\text{T}}/N$ for different total qubit numbers $N$ and fixed high depth $d=30$.
	}
	\label{fig:cliffdepth_B_sup}
\end{figure}

Now, we study higher moments $k$ of the Scrooge ensemble. In Fig.~\ref{fig:cliff_k}a, we plot $\Delta^{(k)}$ against $N_\text{T}$ for large $d$ and different $k$. We find a qualitatively similar decay and convergence behavior for all $k$.
Similarly, in Fig.~\ref{fig:cliff_k}a, we plot $\Delta^{(k)}$ against $d$ for large $N_\text{T}$.

\begin{figure}[htbp]
	\centering	
        \subfigimg[width=0.3\textwidth]{a}{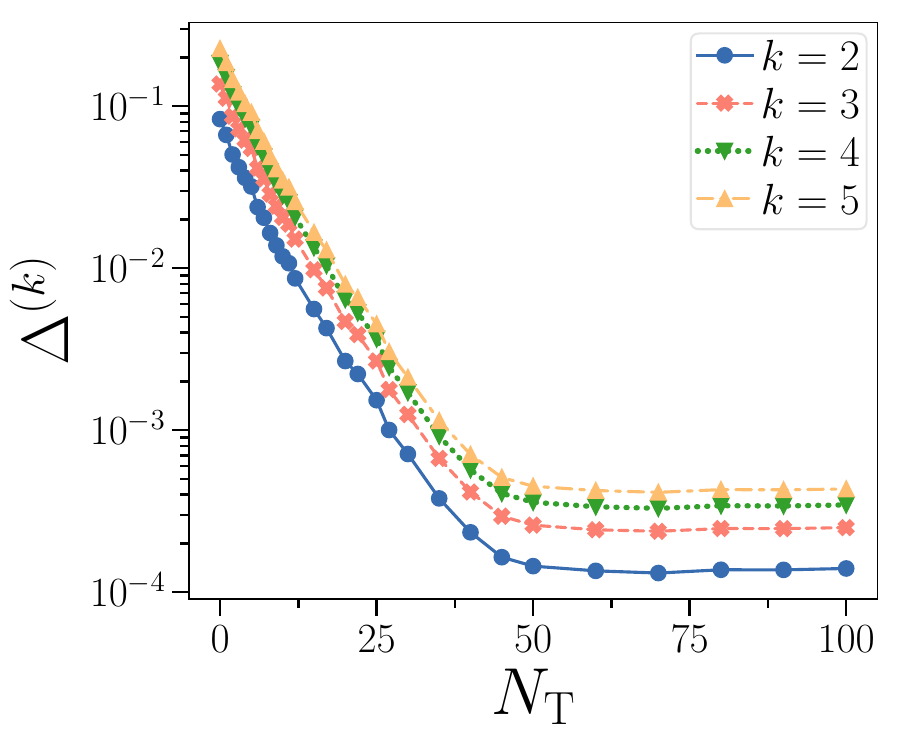}
        \subfigimg[width=0.3\textwidth]{b}{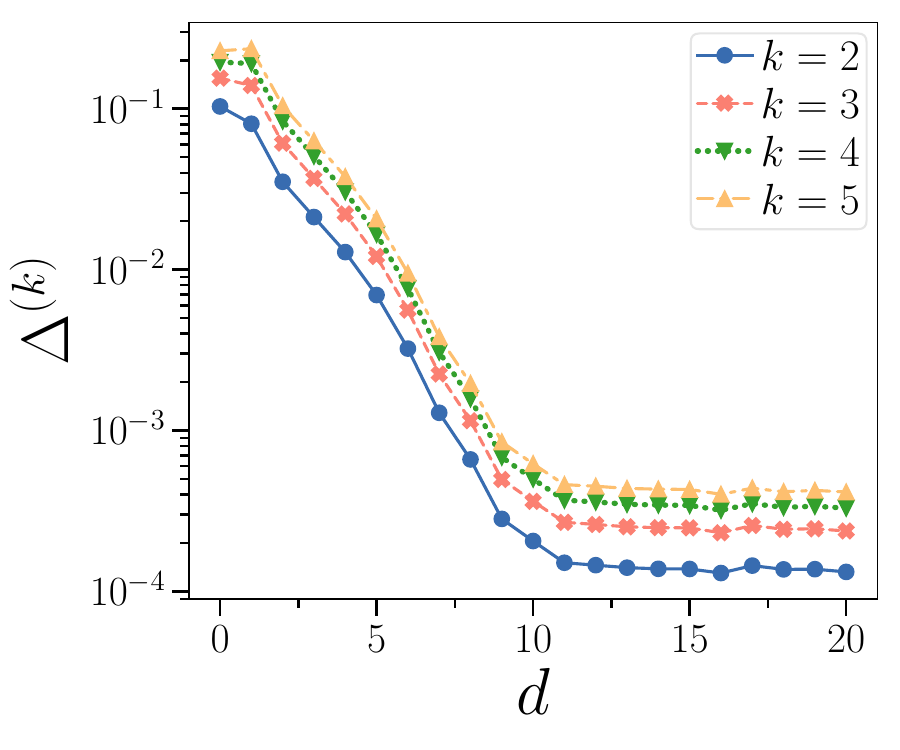}
    \caption{Trace distance to Scrooge $k$-design $\Delta^{(k)}$ of projected ensemble generated from $\ket{\Psi(\chi)}$ as defined in~\eqref{eq:state_ent_sup} where we choose $\chi=\pi/6$. We then measure in a transformed basis on $B$ with $d$ layers of local Clifford gates doped with $N_\text{T}$ T-gates.
    \idg{a} $\Delta^{(k)}$ against $N_\text{T}$ for $d=30$ and different $k$.
    \idg{b} $\Delta^{(k)}$ against $d$ for $N_\text{T}=3N$.
    We have $N_A=1$, $N_B=19$, and average over 100 random realizations of the circuit.
	}
	\label{fig:cliff_k}
\end{figure}

Finally, we study different angles $\chi$ for the generator state $\ket{\psi_\text{ent}(\chi)}$. The choice of $\chi$ affects the entanglement between subsystem $A$ and $B$, and thus the reduced density matrix $\sigma_A$ of the corresponding Scrooge ensemble.
In Fig.~\ref{fig:theta2D}, we show $\chi=\pi/2$ in Fig.~\ref{fig:theta2D}a, which corresponds to the case where $\sigma_A = I/D_A$ is the maximally mixed state. Then, we show $\chi=\pi/3$ in Fig.~\ref{fig:theta2D}b and $\chi=\pi/6$ in Fig.~\ref{fig:theta2D}c. We find similar behavior as function of $N_\text{T}$ and $d$ for all $\chi$.

\begin{figure}[htbp]
	\centering	
        \subfigimg[width=0.3\textwidth]{a}{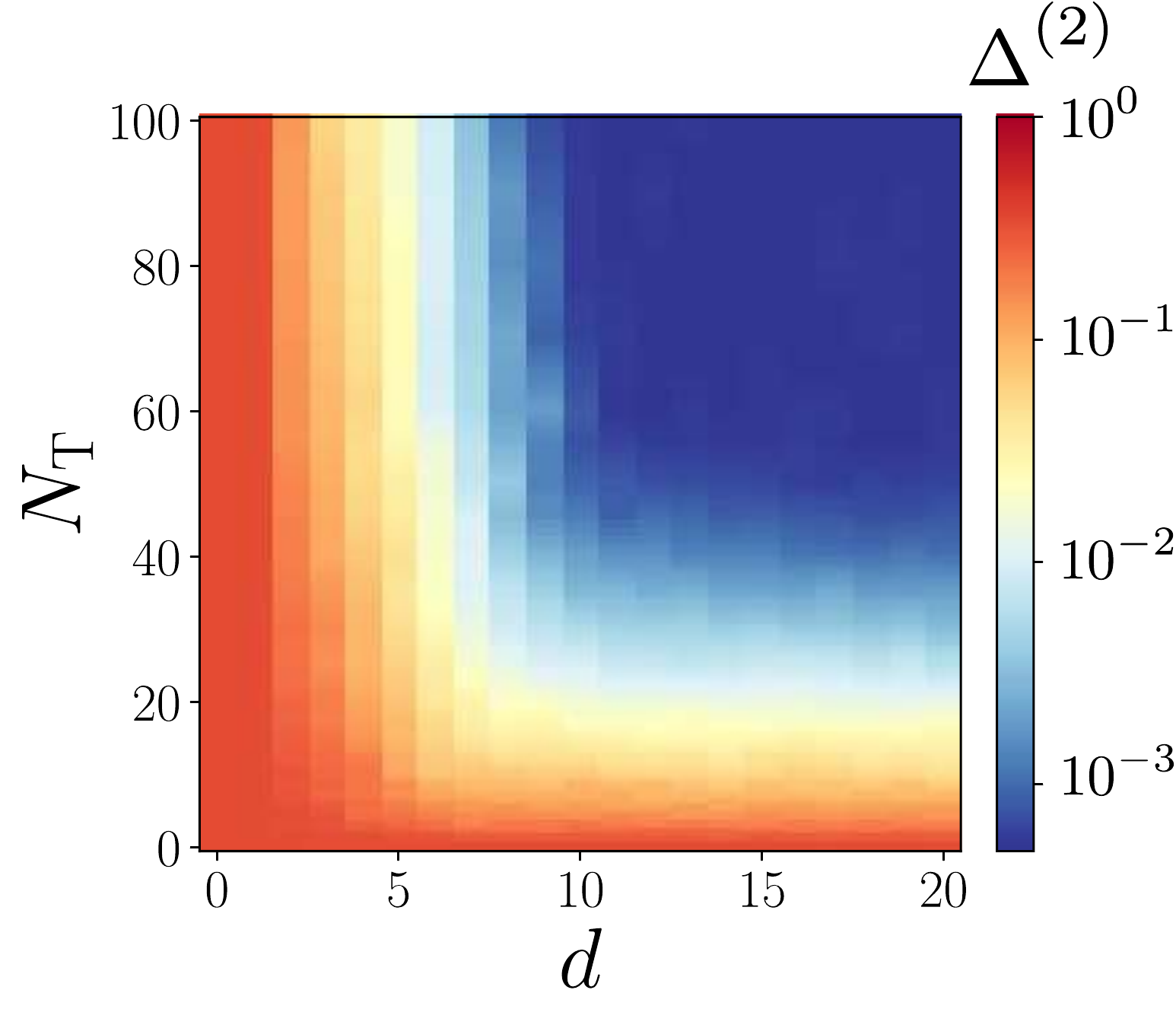}
        \subfigimg[width=0.3\textwidth]{b}{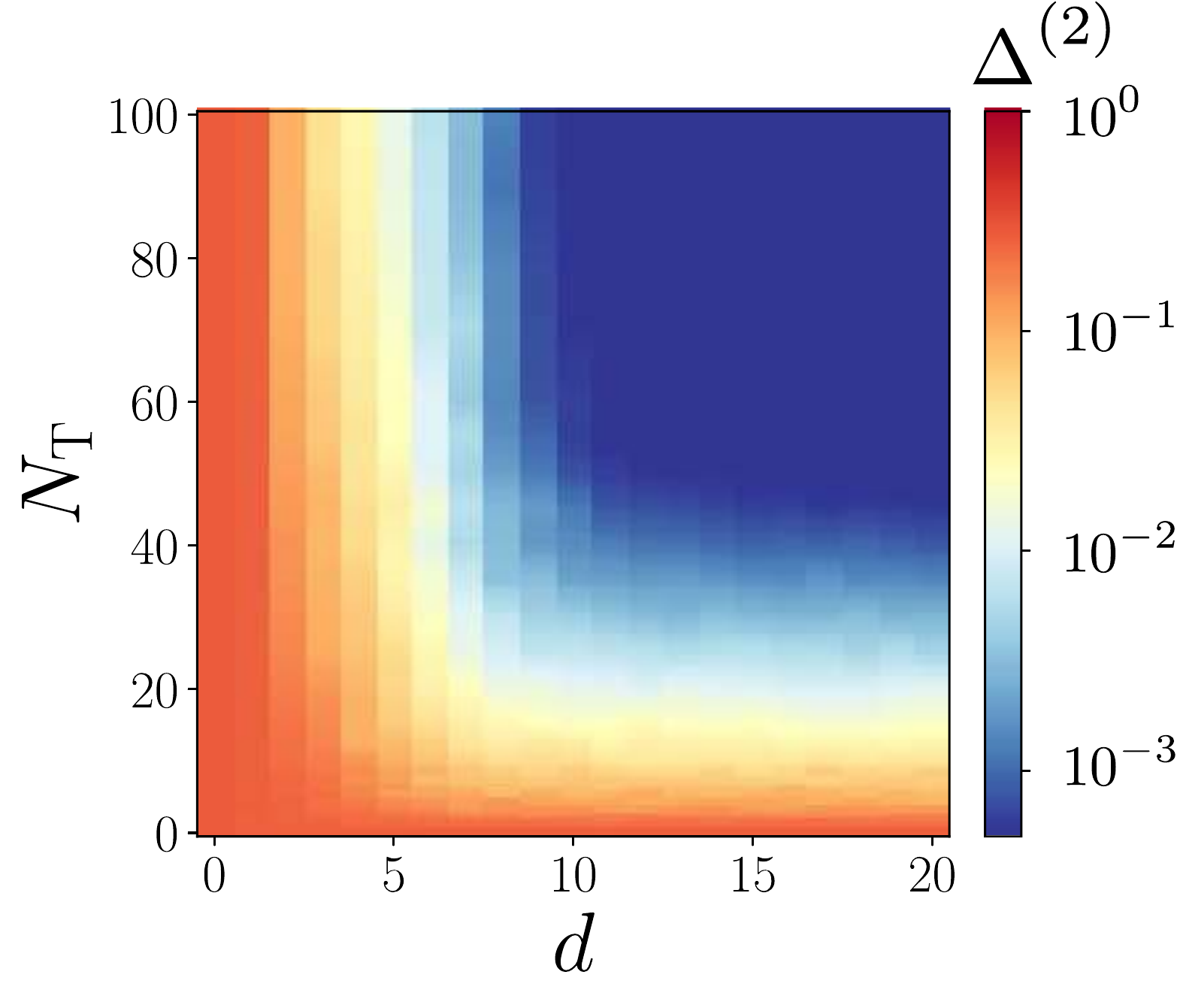}
        \subfigimg[width=0.3\textwidth]{c}{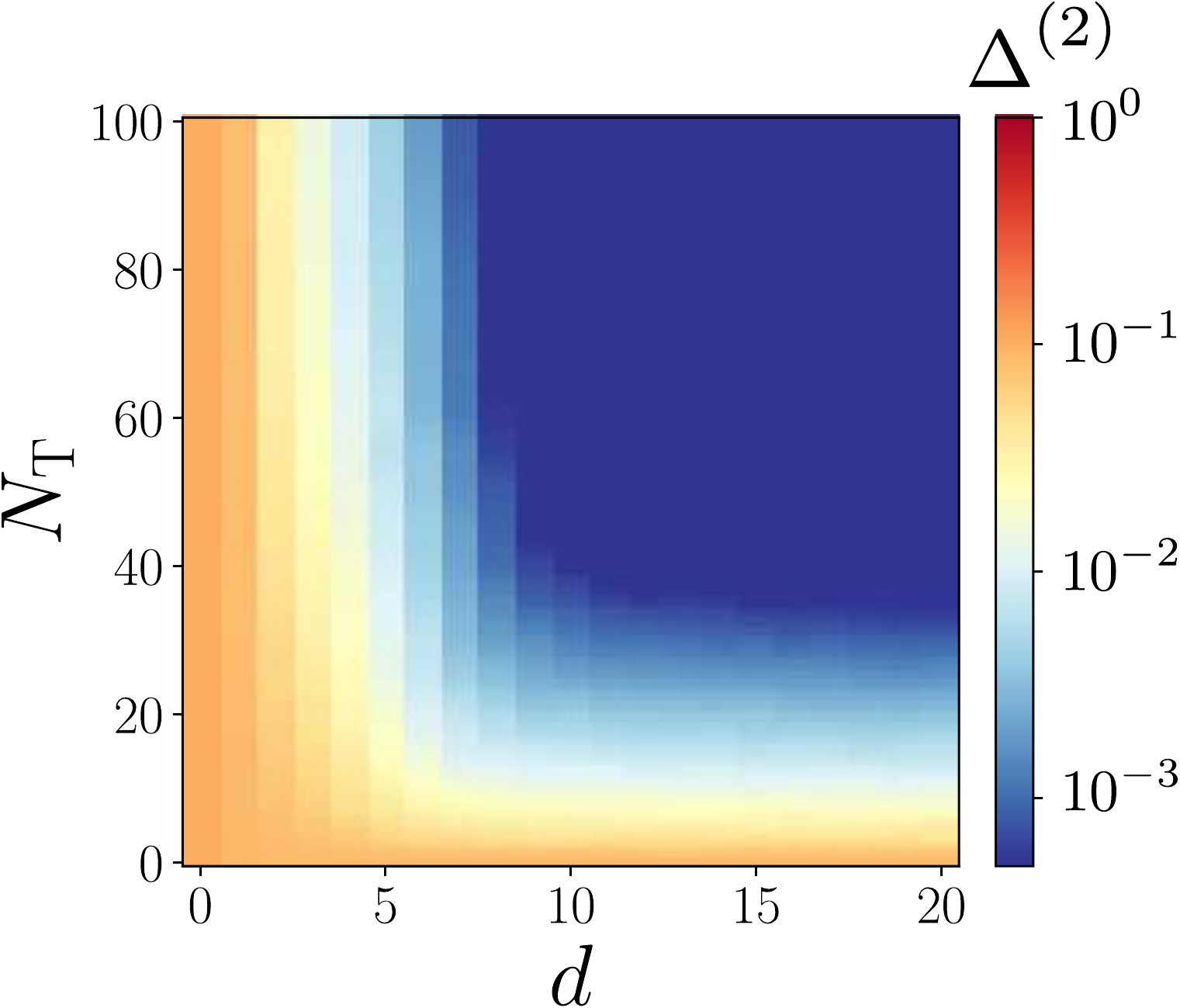}
    \caption{Comparison of different degree of entanglement between $A$ and $B$, where we show different angle $\chi$ for projected ensemble generated from $\ket{\Psi(\chi)}$ as defined in~\eqref{eq:state_ent_sup}.
     We measure in a transformed basis on $B$ with $d$ layers of local Clifford gates doped with $N_\text{T}$ T-gates.
    We plot $d$ against $N_\text{T}$ where we show trace distance to Scrooge $2$-design $\log_{10}(\Delta^{(2)})$ in color. We have 
    \idg{a} $\chi=\pi/2$, \idg{b} $\chi=\pi/3$ and \idg{c} $\chi=\pi/6$.
    We fix $N_A=1$, $N_B=N-N_\text{A}$, and $N=20$.
    We have $N_A=1$, $N_B=19$, and average over 100 random realizations of the circuit.
	}
	\label{fig:theta2D}
\end{figure}

\section{T-doped Clifford generator states}\label{sec:cliffTgenerator}

Next, we consider T-doped Clifford circuits, where, in contrast to the main text, we apply the circuit on both subsystems $A$ and $B$: We prepare a $\ket{0}^{\otimes N}$ state, and then apply the doped Clifford+T circuit on all $N$ qubits (instead of just subsystem $B$). 
Here, in Fig.~\ref{fig:cliffdepth_sup}a we study circuits of $d$ layers, composed of random single-qubit Clifford gates, CNOT gates arranged in a 1D nearest-neighbor configuration, which are doped with in total $N_\text{T}$ T-gates placed randomly in the circuit. While the scrambling via the Clifford circuit increases with $d$, the magic increases with $N_\text{T}$. 
In Fig.~\ref{fig:cliffdepth_sup}a, we show the circuit and setup in detail. Then, in Fig.~\ref{fig:cliffdepth_sup}b we show a 2D heat map for $\Delta^{(2)}$, with varying $N_\text{T}$ and $d$. We find that for low $N_\text{T}$ or low $d$, the additive error between the projected ensemble and the corresponding Scrooge ensemble is large. Both $N_\text{T}$ and $d$ need to be sufficiently large to yield a small $\Delta^{(2)}$.
In Fig.~\ref{fig:cliffdepth_sup}c, we study $\Delta^{(2)}$ against $N_\text{T}$ for different $N$, where we fix the circuit depth to be $d=30$. We find that $\Delta^{(2)}$ decays exponentially with $N_\text{T}$, and saturates to a plateau value at large $N_\text{T}$. The plateau value is primarily due to finite-size effects, and decays exponentially with $N$. 
In Fig.~\ref{fig:cliffdepth_sup}d, we plot $\Delta^{(2)}$ against $d$ for different $N_\text{T}$. The behavior mirrors the one observed in Fig.~\ref{fig:cliffdepth_B_sup}a.
In Fig.~\ref{fig:cliffdepth_sup}e, we study $\Delta^{(2)}$ against $d$ for different $N$, where we choose $N_\text{T}=3N$, i.e. in the limit where magic is large such that we can converge to minimal distance. Notably, we find similar behavior as in the case of Fig.~\ref{fig:cliffdepth_B_sup}b.

\begin{figure}[htbp]
	\centering	
        \subfigimg[width=0.3\textwidth]{a}{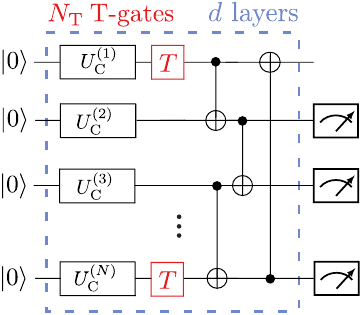}
        \subfigimg[width=0.3\textwidth]{b}{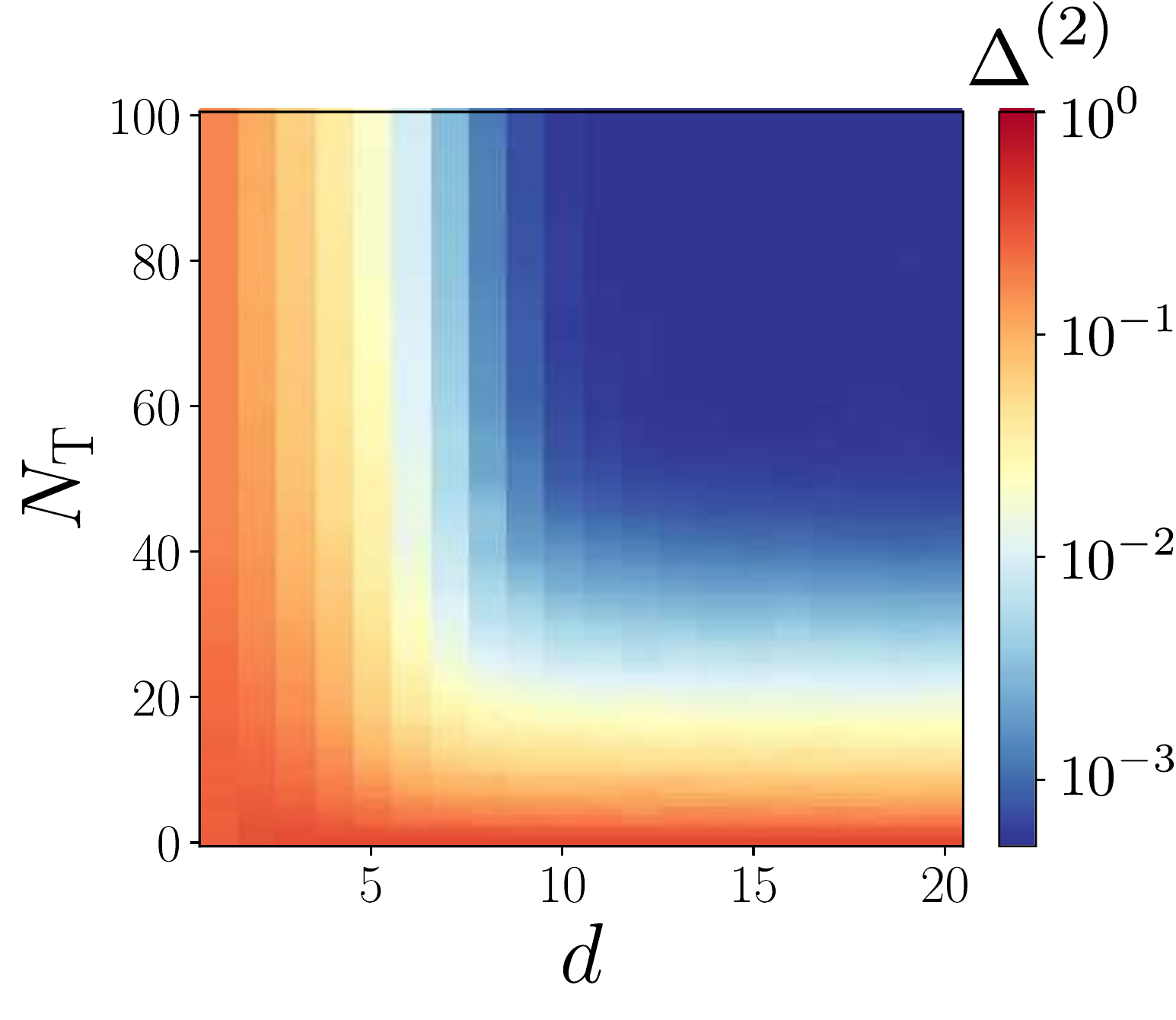}
        	\subfigimg[width=0.3\textwidth]{c}{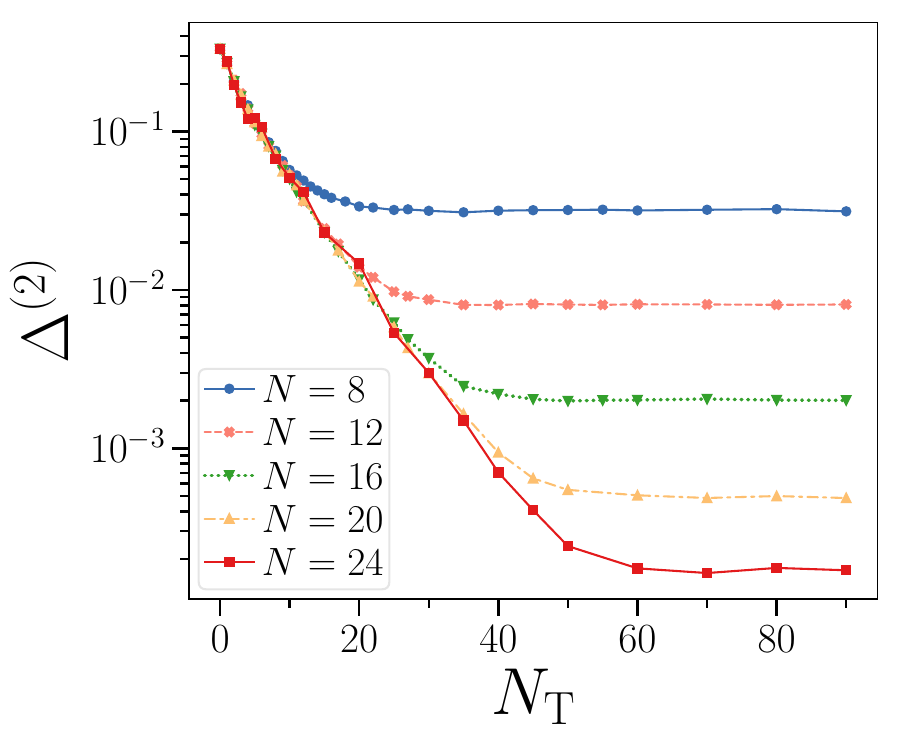}
    	\subfigimg[width=0.3\textwidth]{d}{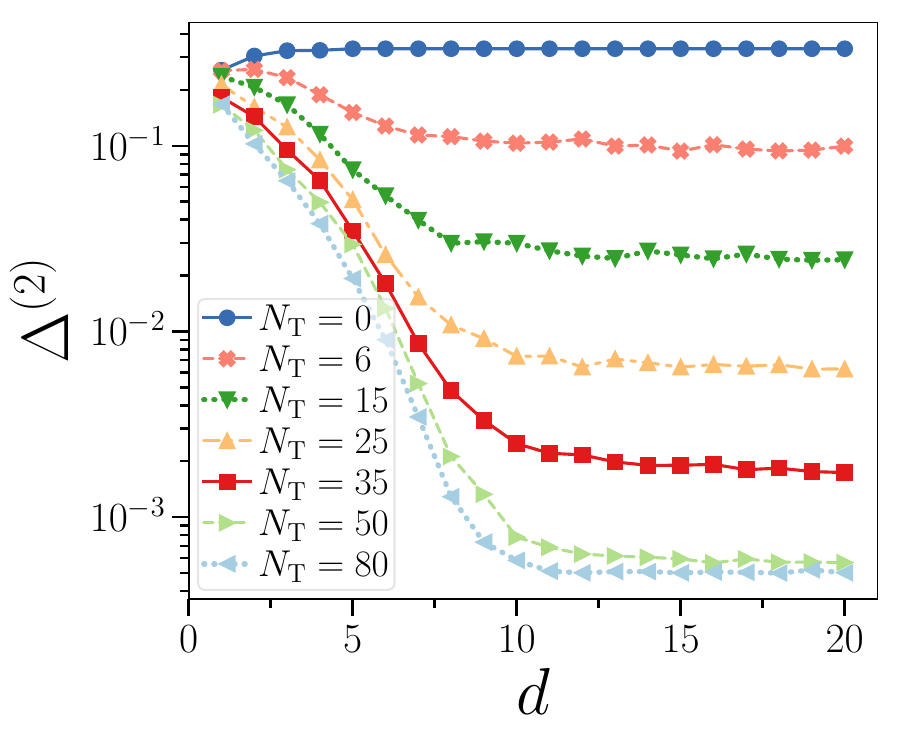}
        \subfigimg[width=0.3\textwidth]{e}{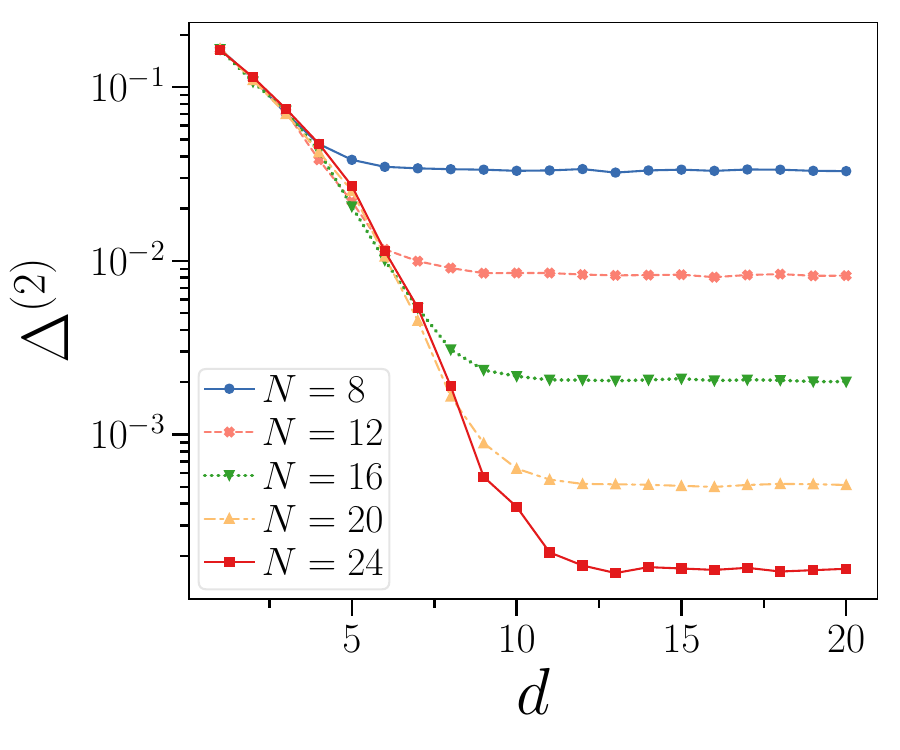}
    \caption{Distance to Scrooge $2$-design $\Delta^{(2)}$ of projected ensemble generated from $U\ket{0}^{\otimes N}$, where $N$-qubit unitary $U$ consists of $d$ layers of random single-qubit Clifford unitaries and 1D layer of nearest-neighbor CNOT gates with periodic boundary conditions, where the circuit is doped with $N_\text{T}$ T-gates at random positions. 
    \idg{a} Sketch of doped Clifford circuit with $N_\text{A}=1$ and $N_\text{B}=3$.
    \idg{b} We plot $d$ (starting from $d\geq1$) against $N_\text{T}$ where we show $\Delta^{(2)}$ (in logarithmic scale) as color scale. We choose $N_\text{B}=19$ and $N_\text{A}=1$.
    \idg{c} $\Delta^{(2)}$ against $N_\text{T}$ for different total qubit numbers $N$ and fixed $d=30$.
    We have $N_A=1$, $N_B=N-N_\text{A}$, and average over 500 random realizations of the circuit.
    \idg{d} $\Delta^{(2)}$ against $d$ for different $N_\text{T}$.
    \idg{e} $\Delta^{(2)}$ against $d$ for different total qubit numbers $N$ and fixed $N_\text{T}=3N$.
    We average over 500 random realizations of the circuit.
	}
	\label{fig:cliffdepth_sup}
\end{figure}

\section{Transverse-field Ising model}\label{sec:ising}
In this section, we study projected ensembles generated from the ground state of the 1D transverse-field Ising model with periodic boundary conditions, as defined in the main text, in more detail. 
We apply a random Clifford unitary on subsystem $B$, then measure the qubits in $B$ in the computational basis, and finally construct the projected ensemble over the remaining $N_\text{A}$ qubits. 
In Fig.~\ref{fig:ising_sup}a, we study the distance to Scrooge $2$-design $\Delta^{(2)}$ against $N_\text{B}$ for different $h$. Notably, we find exponential decay $\Delta^{(2)}\sim 2^{-\alpha(h) N_\text{B}}$ for all $h\neq0$, where we fit $\alpha(h)$. In Fig.~\ref{fig:ising_sup}b, we plot $\alpha(h)$ against $h$, finding a pronounced maximum close to the critical point $h=1$, indicating that the distance decays fastest at the critical point.

Next, we study in Fig.~\ref{fig:ising_sup}c the field $h$ against $y=\log_2(\Delta^{(2)})/N_\text{B}$, which is the logarithm of the trace distance normalized by $N_\text{B}$. We define $y$ as it converges  to a non-zero constant for large $N_\text{B}$, as we will show below.
Notably, we find that for sufficiently large $N_\text{B}$, there is a pronounced dip for $h\approx 1$, matching the well-known critical point of the Ising model~\cite{osterloh2002scaling}. 
We characterize now the behavior around the critical point.
First, we define the minimal distance $y_0(N_\text{B})=\text{min}_h y(h,N_\text{B})$, where we perform the minimization at a small region around the critical point. In Fig.~\ref{fig:ising_sup}d, we plot $y_0$ against $N_\text{B}$, finding that $y_0$ increases with $N_\text{B}$. Following the approach of Ref.~\cite{haug2023quantifying}, we fit the curve with $y_0=a N_\text{B}^\gamma+y_\text{c}$, finding good agreement. The fit allows us to extract the asymptotic value $y_0(N_\text{B}\rightarrow\infty)$. 
In Fig.~\ref{fig:ising_sup}e, we plot the minimal field $h_0(N_\text{B})=\text{argmin}_h y(h,N_\text{B})$ against $N_\text{B}$. We perform a similar polynomial fit with $h_0=a N_\text{B}^\gamma+h_\text{c}$, finding the asymptotic field $h_\text{c}\equiv h_0(N_\text{B}\rightarrow\infty)\approx1.006(22)$ which closely matches the critical point $h=1$. 
Our scaling analysis shows that within our numerical study, the dip in $\Delta^{(2)}$ indeed converges to $h=1$ for $N_\text{B}\rightarrow\infty$ (see \SM{}~\ref{sec:ising}). Thus, emergent Scrooge designs can provide a method to determine the critical point~\cite{osterloh2002scaling,haug2023quantifying}.

Next, in Fig.~\ref{fig:ising_sup}f, we rescale $y$ and $h$ with fitted $y_0$ and $h_0$ for different $N_\text{B}$. We find that curves for different $N_\text{B}$ collapse onto a single curve when rescaling $h-h_0$ with $N_\text{B}^{1/\nu}$, where we have $\nu=1$ as expected for the Ising universality class. We also show a third-order polynomial fit as a dashed line, allowing us to predict the behavior close to the critical point for all $N_\text{B}$.  

\begin{figure*}[htbp]
	\centering	
        \subfigimg[width=0.3\textwidth]{a}{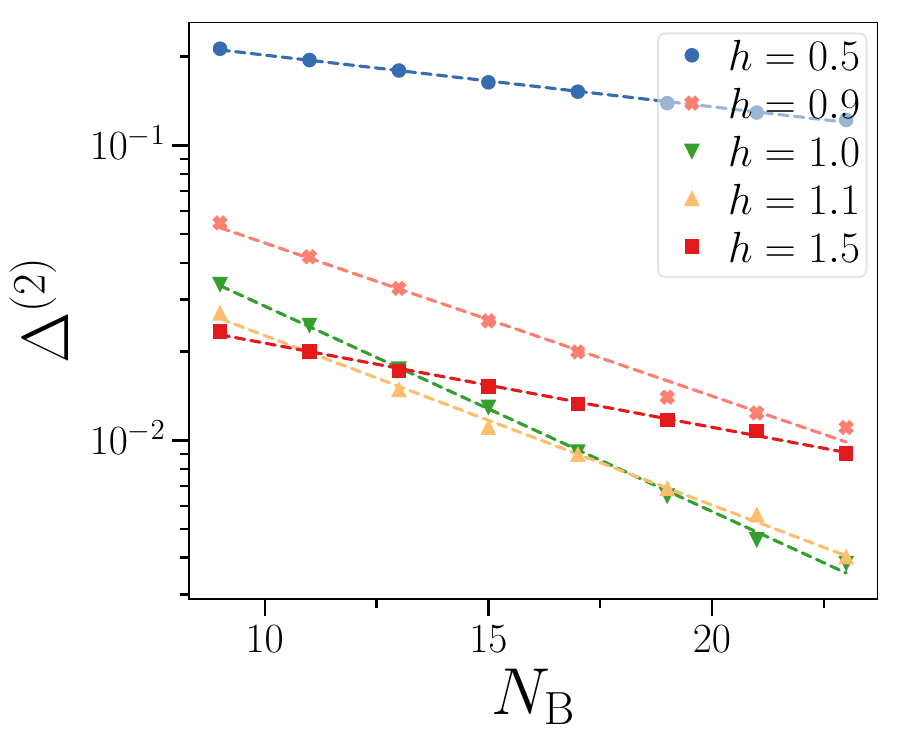}
       \subfigimg[width=0.3\textwidth]{b}{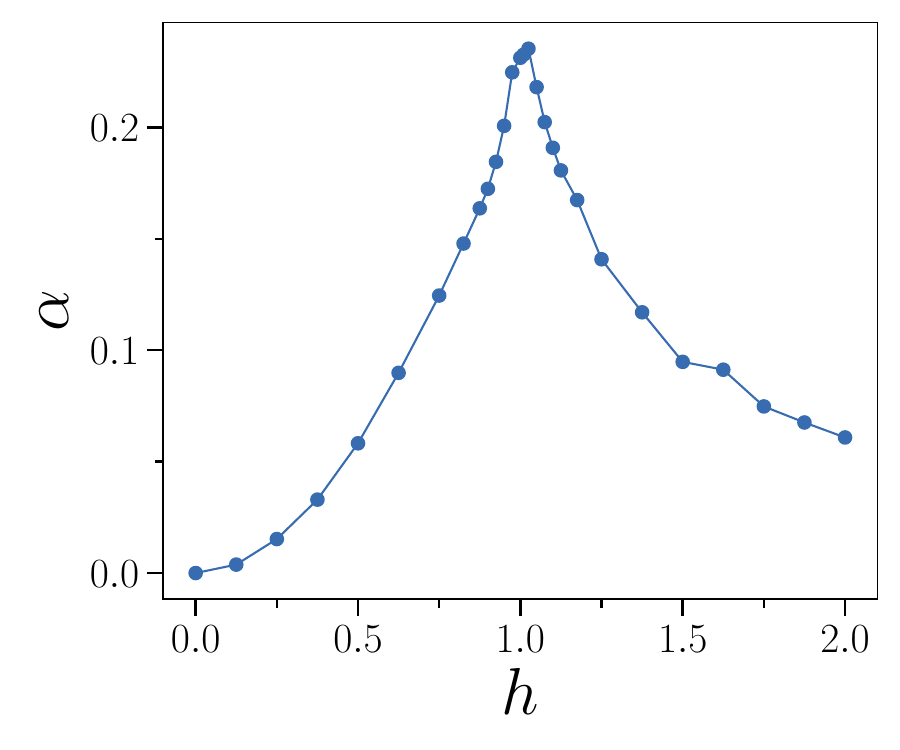}
    	\subfigimg[width=0.3\textwidth]{c}{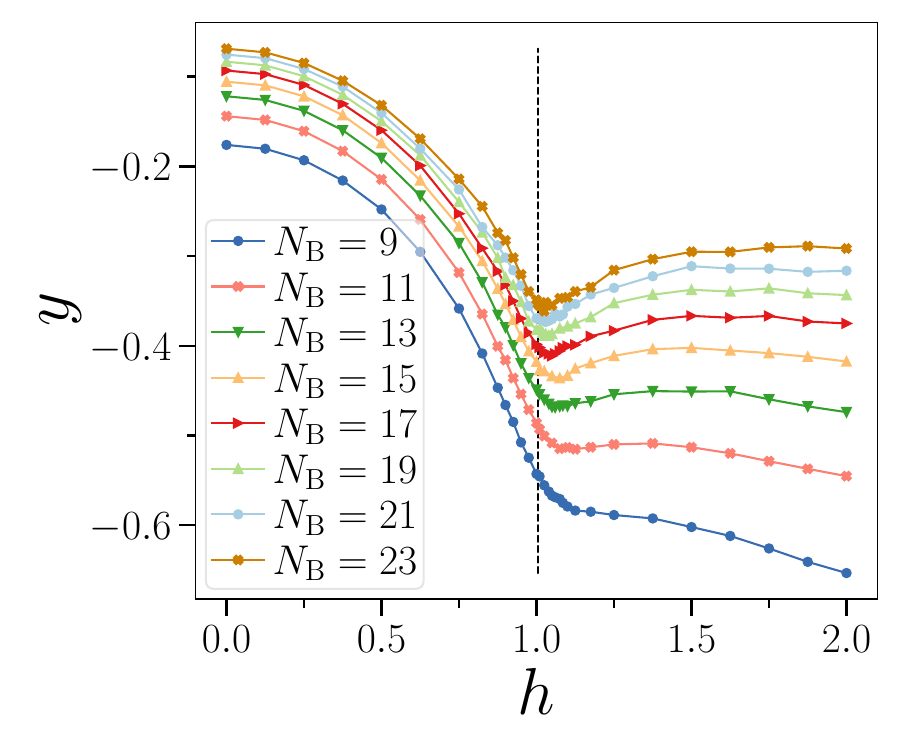}
            \subfigimg[width=0.3\textwidth]{d}{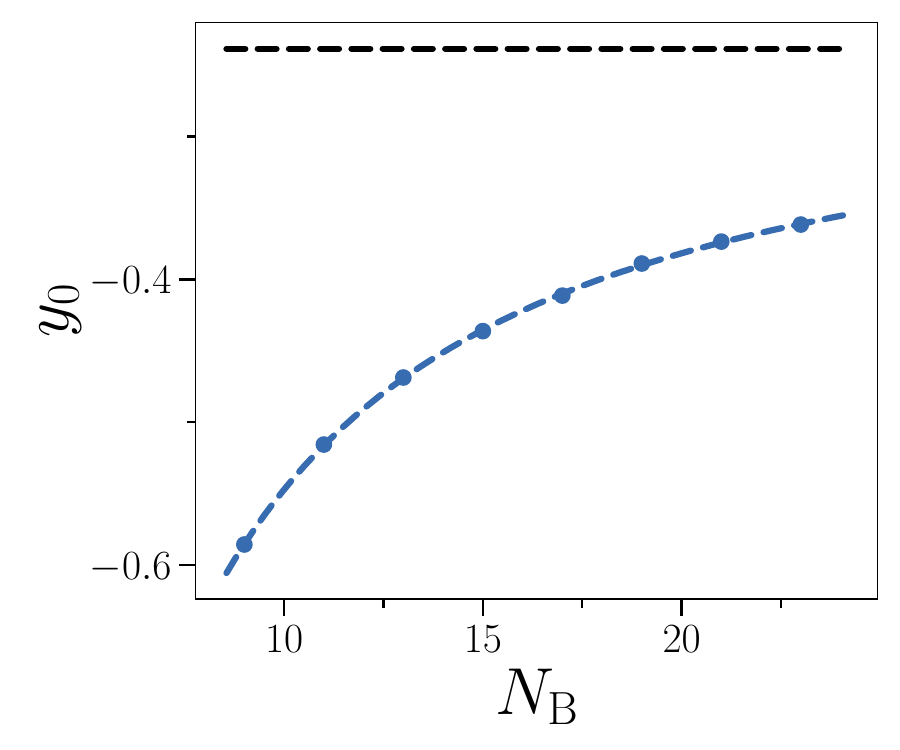}
        \subfigimg[width=0.3\textwidth]{e}{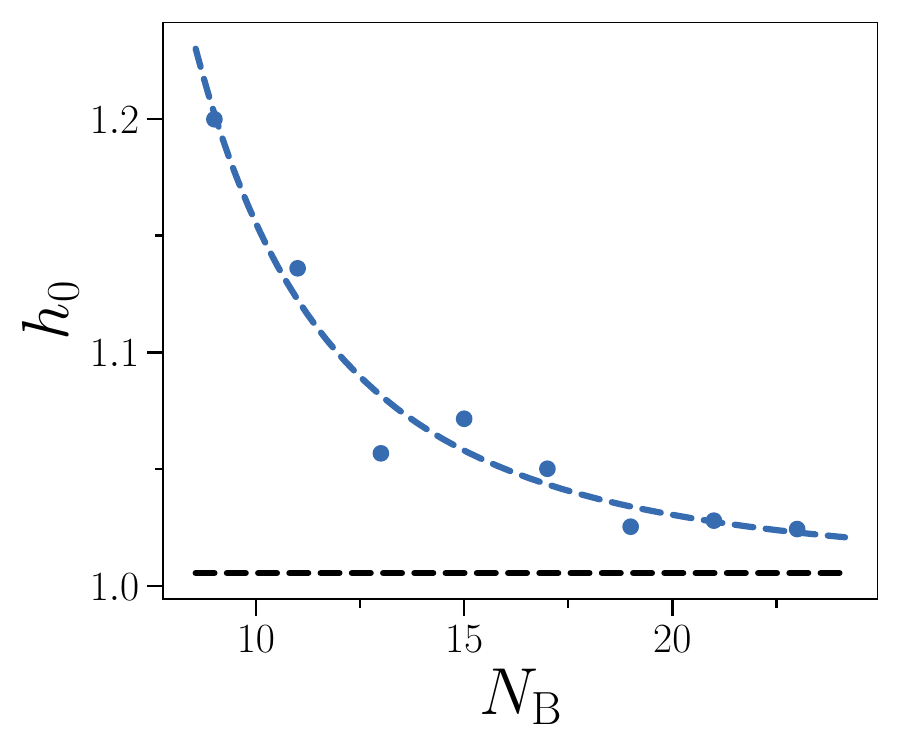}
        \subfigimg[width=0.3\textwidth]{f}{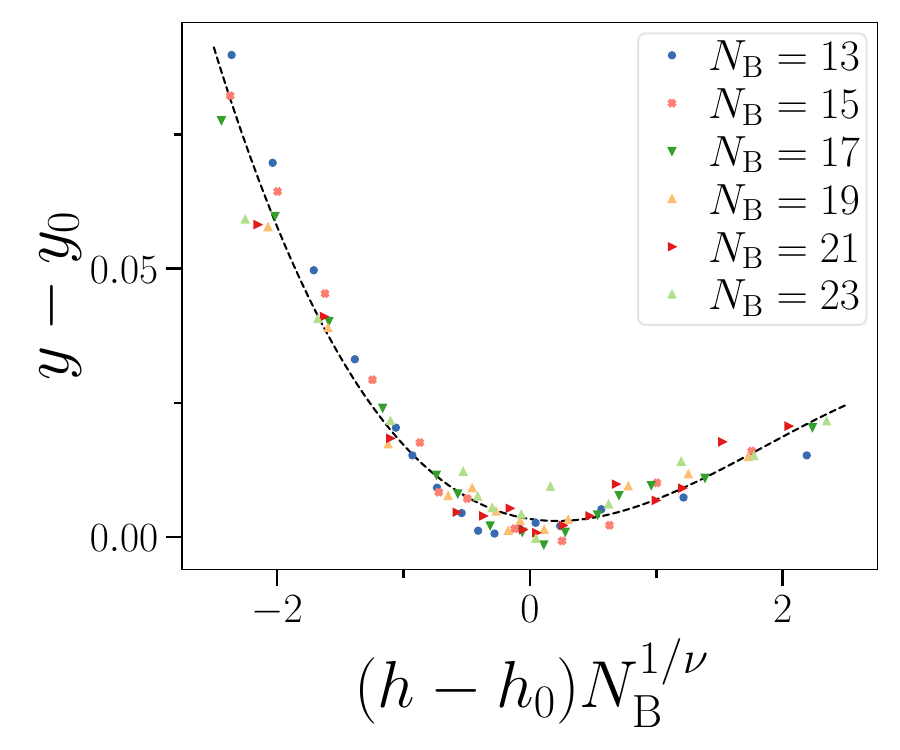}
    \caption{Distance to Scrooge $2$-design $\Delta^{(2)}$ for projected ensemble generated from ground state of Ising model~\eqref{eq:ising}. We apply unitary $U_\text{B}$ on $N_\text{B}$ qubits, measure $N_\text{B}$ qubits, and gain projected ensemble of $N_\text{A}=1$ qubits.
    \idg{a} We show $\Delta^{(2)}$ against $N_\text{B}$ for different $h$ for the case of random Clifford unitaries. Dashed lines are fit $\Delta^{(2)}\sim 2^{-\alpha(h) N_\text{B}}$
    \idg{b} We plot fitted exponential decay rates $\alpha(h)$ against $h$. 
    \idg{c} We define $y=\log_2(\Delta^{(2)})/N_\text{B}$ and plot against $h$. Dashed vertical line is critical point $h_\text{c}=1$ of Ising model which we now proceed to fit using $y$. 
    \idg{d} We plot minimal distance $y_0(N_\text{B})=\text{min}_h y(h,N_\text{B})$ against $N_\text{B}$, where the minimization over $h$ is performed around a small neighborhood around $h=1$. We fit with $y_0=a N_\text{B}^\gamma+y_\text{c}$, where we find as asymptotic value $y_\text{c}\equiv y_0(N_\text{B}\rightarrow\infty)\approx -0.238(7)$. 
    \idg{e} We plot field with minimal distance $h_0(N_\text{B})=\text{argmin}_h y(h,N_\text{B})$ against $N_\text{B}$. We fit with $h_0=a N_\text{B}^\gamma+h_\text{c}$, where we find fitted critical field $h_\text{c}\equiv h_0(N_\text{B}\rightarrow\infty)\approx1.006(22)$.
    \idg{f} We rescale $y$ and $h$ with fitted $y_0$ and $h_0$. We find that curves for different $N_\text{B}$ collapse onto a single curve when rescaling $h-h_0$ with $N_\text{B}^{1/\nu}$, where we have $\nu=1$ as expected for the Ising universality class. The dashed line is a third-order polynomial fit.
	}
	\label{fig:ising_sup}
\end{figure*}

Next, we study the behavior of projected ensemble generated from the  ground state  of the Ising model for large fields $h$. We transform bipartition $B$ of the ground state with random Clifford unitaries $U_\text{B}$. 
In Fig.~\ref{fig:isinglargeh}a, we plot $\Delta^{(2)}$ against $h$ for different types of unitaries applied on $N_\text{B}$. Via fitting, we find that for all unitaries, we have a decay $\Delta^{(2)}\sim h^{-\alpha}$, with $\alpha\approx 2$. The decay results from the fact that for $h\rightarrow\infty$, the ground state of the Ising model is a product state. As product states are trivially exact Scrooge ensembles, $\Delta^{(2)}$ must decay to zero as $h\rightarrow\infty$. 

In Fig.~\ref{fig:isinglargeh}b, we plot $\Delta^{(2)}$ against $h$ for different $N_\text{B}$. For large $h$, we find that all $N_\text{B}$ follow the same polynomial decay as $\Delta^{(2)}\sim h^{-\alpha}$, where again we find $\alpha\approx 2$.
\begin{figure}[htbp]
	\centering	
    \subfigimg[width=0.3\textwidth]{a}{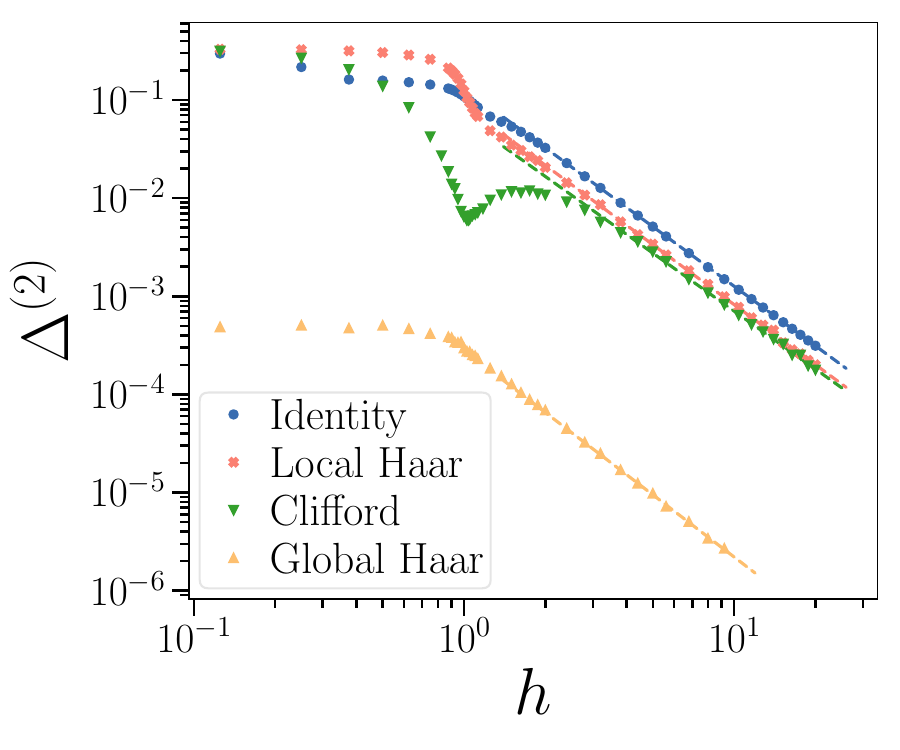}
    	\subfigimg[width=0.3\textwidth]{b}{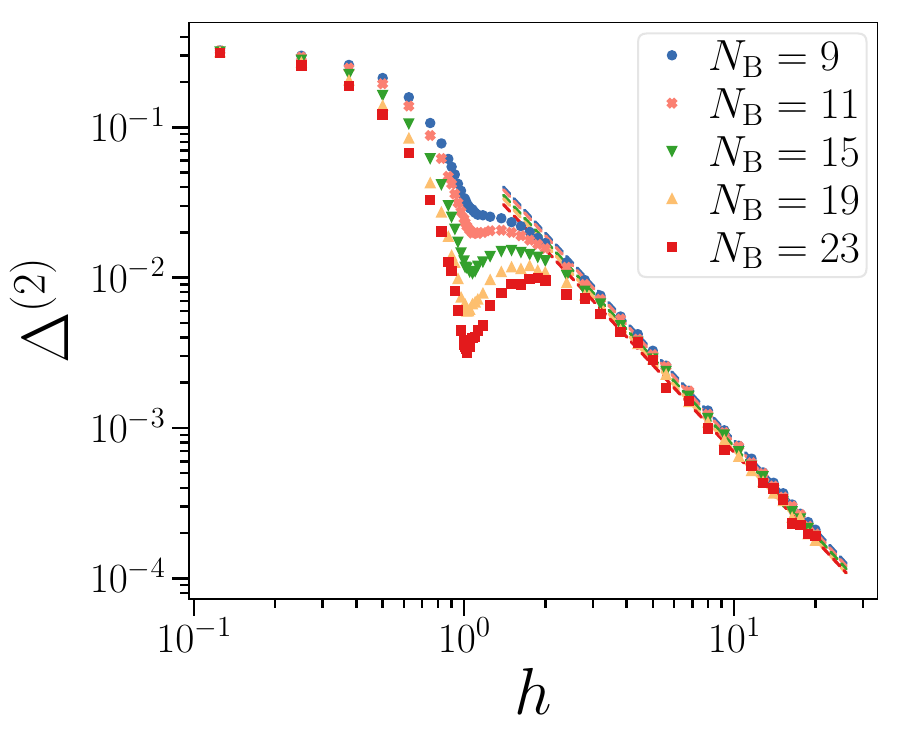}
    \caption{Distance to Scrooge $2$-design $\Delta^{(2)}$ of the projected ensemble generated from ground state of Ising model by applying unitary $U_\text{B}$ on $B$, then measuring $B$ in the computational basis. We have $N_\text{A}=1$ qubits for the projected ensemble.
    \idg{a} We plot $\Delta^{(2)}$ against $h$ for different types of unitaries $U_\text{B}$ on $B$ which has $N_\text{B}=19$ qubits. For large $h\gg1$, we fit as dashed line a polynomial fit $\Delta^{(2)}\sim h^{-\alpha}$, where we find $\alpha\approx 2$. 
    \idg{b} We plot $\Delta^{(2)}$ against $h$ for different $N$, where we choose random Clifford unitaries $U_\text{B}$.
	}
	\label{fig:isinglargeh}
\end{figure}

Finally, we study the error of Scrooge $k$-design $\Delta^{(k)}$ beyond $k=2$. In Fig.~\ref{fig:isingk}, we plot $\Delta^{(k)}$ against $h$ for different measurement basis (via unitary $U_\text{B}$ and $k=2,\dots,5$). 
Here, we have the identity applied on $B$ in Fig.~\ref{fig:isingk}a, tensor product of single-qubit Haar random unitaries in Fig.~\ref{fig:isingk}b, random Clifford unitaries in Fig.~\ref{fig:isingk}c and unitaries drawn from the Haar measure in Fig.~\ref{fig:isingk}d.
We find similar behavior for all shown $k$, indicating that the qualitative behavior for $k=2$ holds similarly for higher $k$.

\begin{figure}[htbp]
	\centering	
    \subfigimg[width=0.24\textwidth]{a}{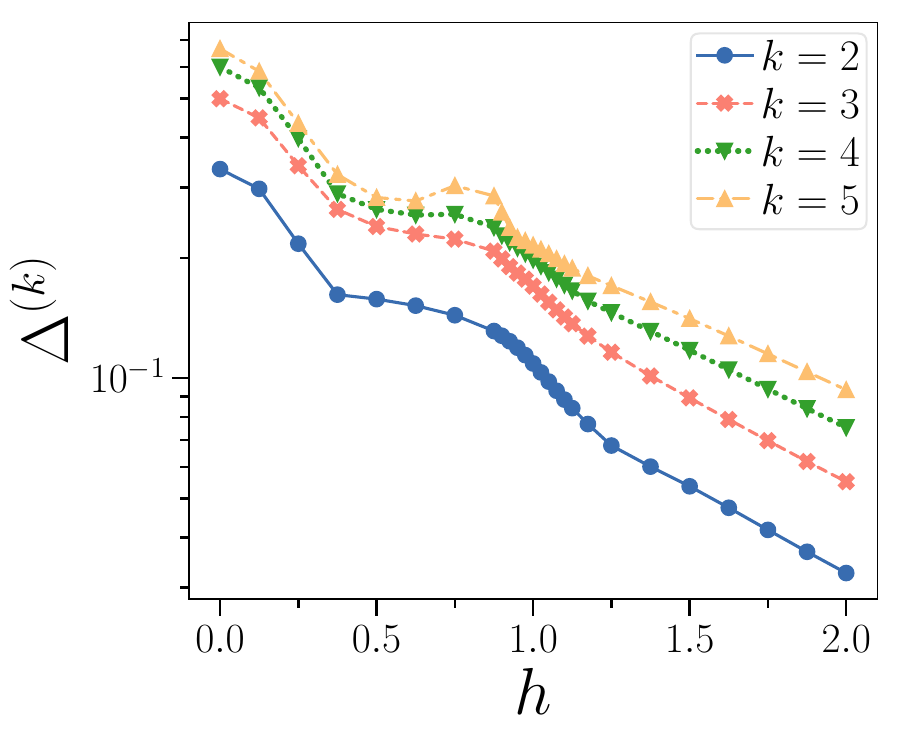}
    	\subfigimg[width=0.24\textwidth]{b}{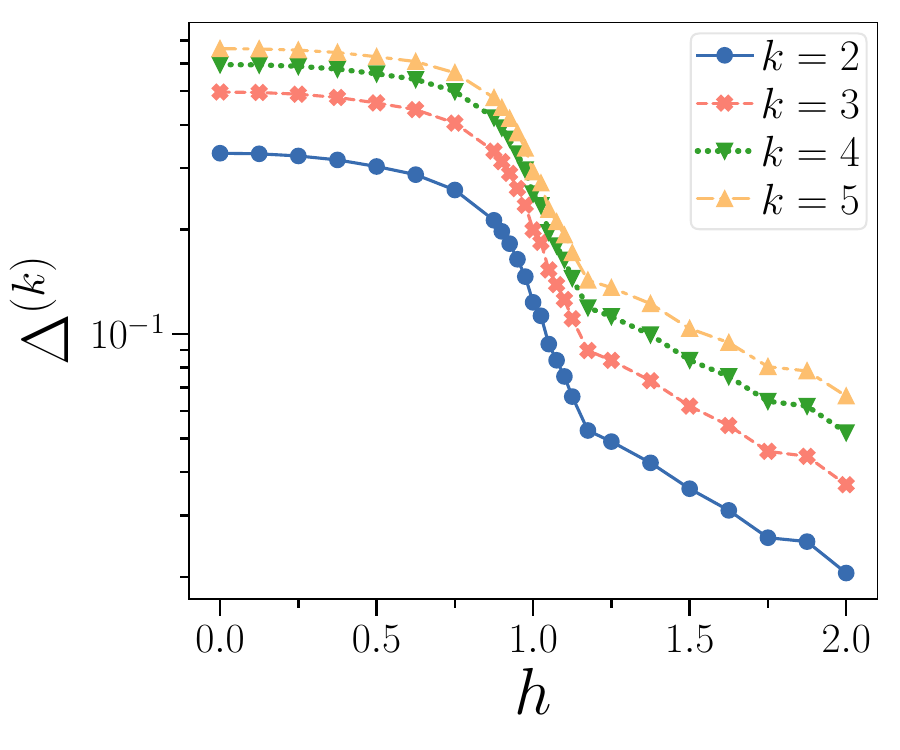}
        \subfigimg[width=0.24\textwidth]{c}{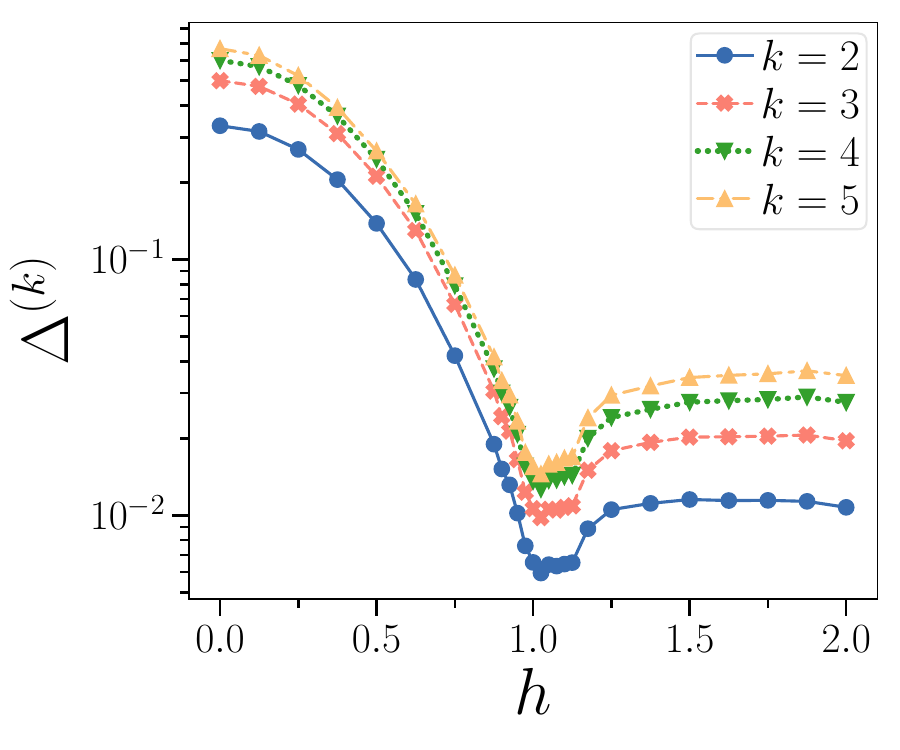}
    	\subfigimg[width=0.24\textwidth]{d}{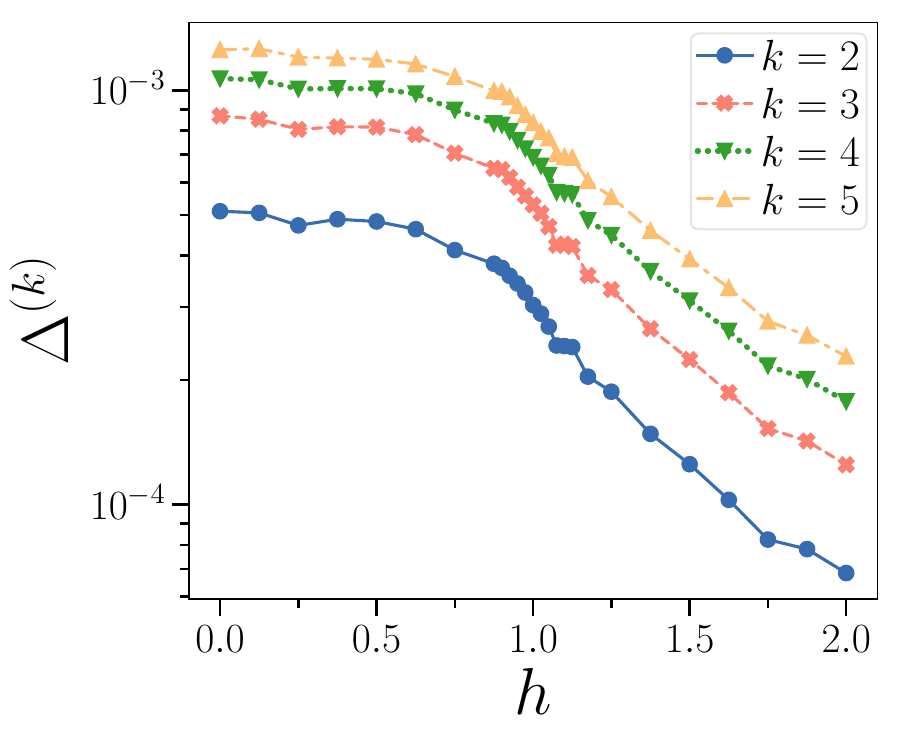}
    \caption{Distance to Scrooge $k$-design $\Delta^{(k)}$ of the projected ensemble generated from the ground state of Ising model by applying unitary $U_\text{B}$ on $B$, then measuring it in the computational basis. We show $\Delta^{(k)}$ against $h$ for different types of unitaries $U_\text{B}$ on $B$ with $N_\text{B}=19$, $N_\text{A}=1$.
    We regard $U_\text{B}$ as \idg{a} identity, \idg{b} single-qubit Haar random unitaries, \idg{c} random Clifford unitaries and \idg{d} unitaries drawn from the Haar measure on $N_\text{B}$ qubits.
	}
	\label{fig:isingk}
\end{figure}

\section{Heisenberg model}\label{sec:heisenberg}

Next, we study the Heisenberg model with anisotropy $h$
\begin{equation}\label{eq:xxz}
    H_{\text{XXZ}}=\sum_{j=1}^N( -X_jX_{j+1}-Y_j Y_{j+1} -hZ_j Z_{j+1})\,,
\end{equation}
where $X_j$, $Y_j$ and $Z_j$ are the respective Pauli x, y and z operators acting on the $j$th qubit.
We study the emergent Scrooge designs generated from the ground state in Fig.~\ref{fig:heisenberg}. We apply unitary $U_\text{B}$ on $B$ and regard the projected ensemble of the remaining $N_\text{A}$ qubits. We regard $U_\text{B}$ as being the identity, single-qubit Haar random unitaries, or random unitaries drawn from the Haar measure on $N_\text{B}$ qubits. We find in Fig.~\ref{fig:heisenberg}a that $\Delta^{(2)}$ changes slightly with anisotropy $h$. Notably, we find an increase in $\Delta^{(2)}$ for $h\approx1$, which is most pronounced when we apply Clifford unitaries on $B$.
In Fig.~\ref{fig:heisenberg}b, we study $\Delta^{(2)}$ against $N_\text{B}$ for different $U_\text{B}$ for $h=1$. We find an exponential decay for Clifford and Haar random unitaries on $B$, while local unitaries (i.e. single-qubit Haar or identity) yield a large $\Delta^{(2)}$ for any $N_\text{B}$.

\begin{figure}[htbp]
	\centering	
	\subfigimg[width=0.3\textwidth]{a}{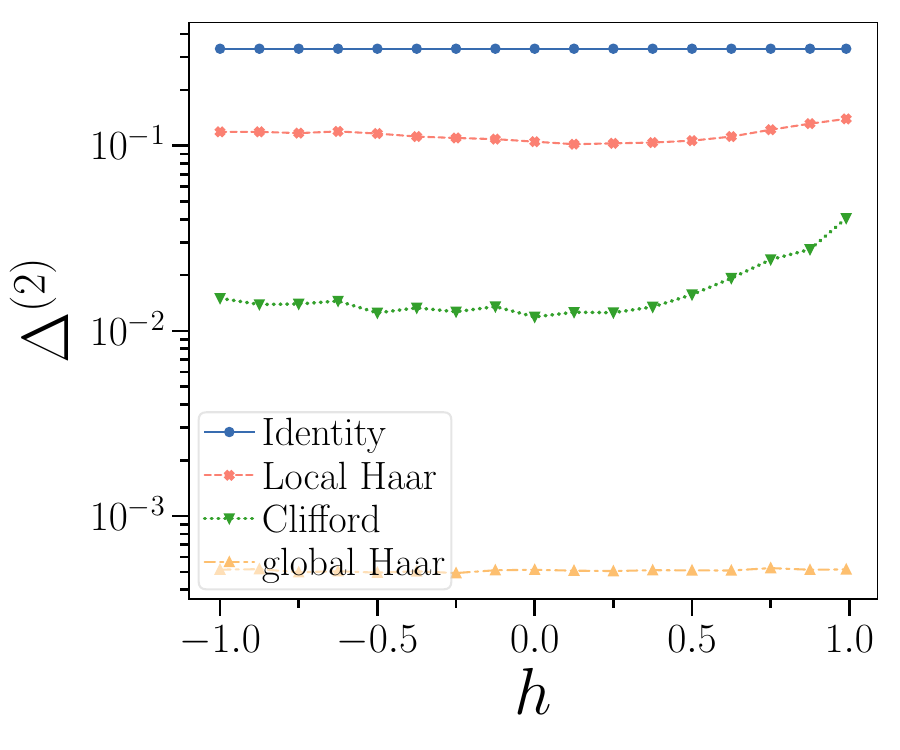}    \subfigimg[width=0.3\textwidth]{b}{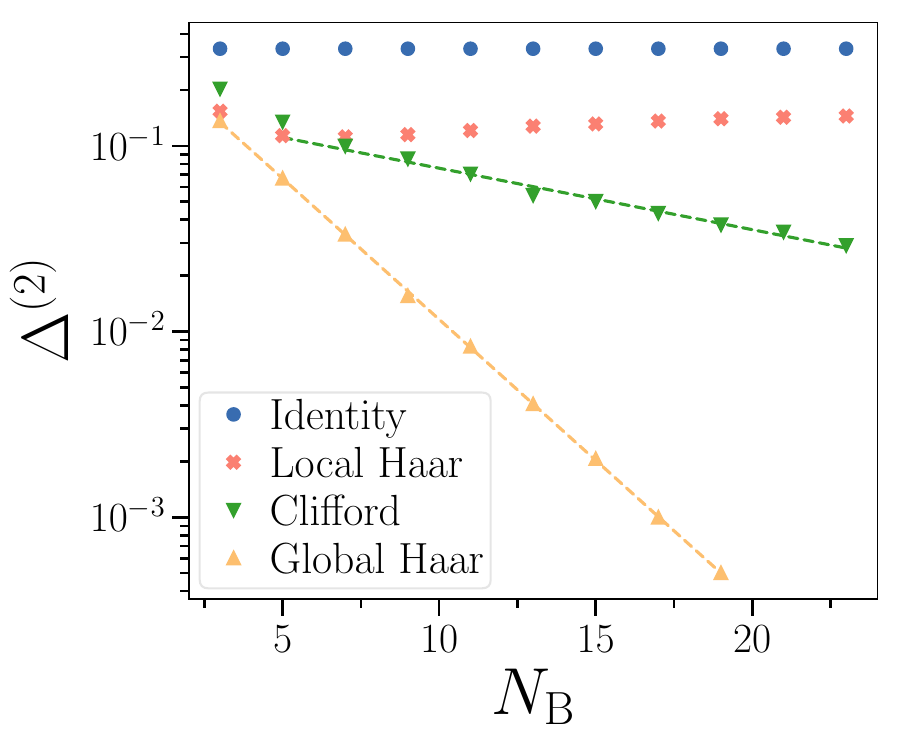}
    \caption{Trace distance to Scrooge $2$-design $\Delta^{(2)}$ for projected ensemble generated from ground state of Heisenberg model~\eqref{eq:xxz}. We apply unitary $U_\text{B}$ on $N_\text{B}$ qubits, measure $N_\text{B}$ qubits, and gain projected ensemble of $N_\text{A}=1$ qubits.
    \idg{a} $\Delta^{(2)}$ against anisotropy $h$ for $U_\text{B}$ being identity, single-qubit Haar random unitaries, \change{random Clifford unitaries}, or random unitaries drawn from the Haar measure on $N_\text{B}$ qubits. 
    \idg{b} $\Delta^{(2)}$ against $N_\text{B}$ for different $U_\text{B}$. Dashed line is fit with $\Delta^{(2)}\sim 2^{-\gamma N_\text{B}}$, where we find $\gamma_\text{Clifford}\approx0.11$ and $\gamma_\text{Haar}\approx0.5$.
	}
	\label{fig:heisenberg}
\end{figure}

\section{Stabilizer states with different basis measurements}\label{sec:stabmeasbasis}
In this section, we study the projected ensemble generated from stabilizer states. 
We prepare a random $N$-qubit stabilizer state, apply  unitaries $U_\text{B}$ on the bipartition $B$, then proceed to measure $B$ in the computational basis and study the emergent Scrooge ensemble on $A$. \change{For $U_\text{B}$, we consider random Clifford unitaries, tensor product of single-qubit Haar random unitaries, tensor product of T-gates followed by tensor product of Hadamard gates to measure in $X$-basis, and $N_\text{B}$-qubit Haar random unitaries.}
In Fig.~\ref{fig:stabilizer}, we show the distance to Haar $2$-design $\Delta^{(2)}$ against $N_\text{B}$ for different classes of unitaries $U_\text{B}$. Notably, we find that when $U_\text{B}$ is Clifford, $\Delta^{(2)}$ is large, which follows from the fact that the state after basis rotation is still a stabilizer state.
In contrast, for random single-qubit Haar unitaries, rotation into the T-basis $T=\text{diag}(1,e^{-i\pi/4})$ or Haar random unitaries over $N_\text{B}$, we find exponential decay with $N_\text{B}$. Thus, by injecting magic into $B$ via a unitary (which does not commute with the computational measurement basis), we can generate emergent Scrooge ensembles with low error. 
Notably, as the initial stabilizer state is already highly coherent and has been scrambled via Cliffords, we find that injection of magic via measurements in a local magical basis (such as single-qubit Haar random unitaries) are then sufficient to produce good Scrooge designs. 
This contrasts the case of ground states of local Hamiltonians (such as Ising or Heisenberg model), which are weakly scrambled and thus local (magical) measurements are not sufficient to yield Scrooge designs.

\begin{figure}[htbp]
	\centering	
	\subfigimg[width=0.3\textwidth]{a}{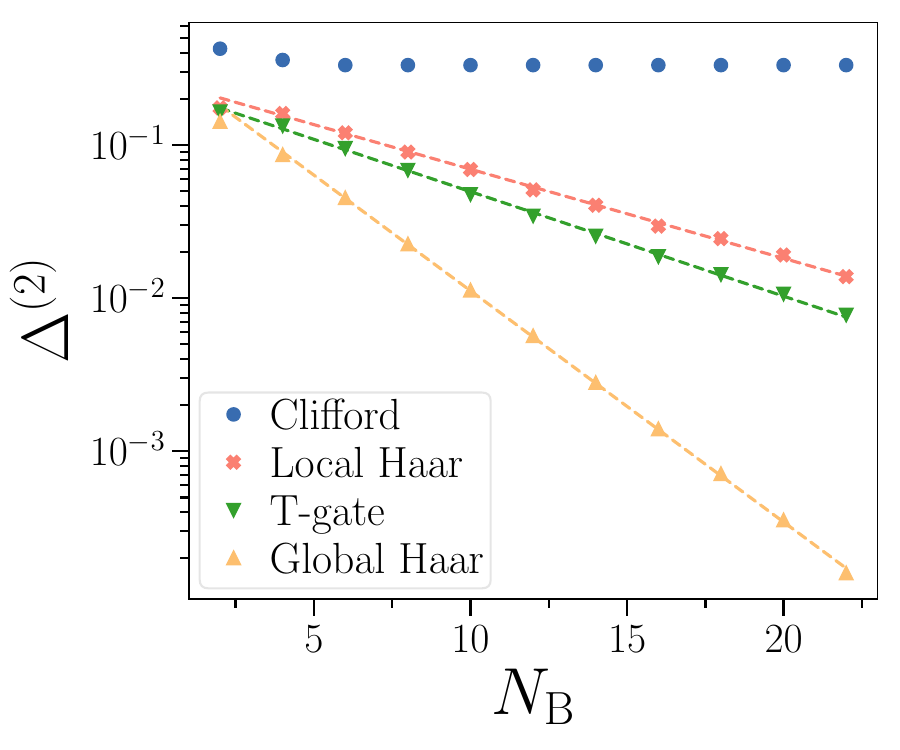}\hfill
    \caption{Distance to Haar $2$-design $\Delta^{(2)}$ for projected ensemble of $N_\text{A}=1$ qubits, generated from initial stabilizer state over $N=N_\text{A}+N_\text{B}$ qubits. We apply on $B$ either random Clifford unitaries, single-qubit Haar random unitaries, tensor product of T-gates \change{(followed by tensor product of Hadamard gates)} or random unitaries drawn from the Haar measure, and measure $B$ in the computational basis. We plot $\Delta^{(2)}$ against $N_\text{B}$. The dashed line is a fit with $\Delta^{(2)}\sim 2^{-\gamma N_\text{B}}$, where we find $\gamma_\text{1-Haar}\approx0.19$, $\gamma_\text{T-gate}\approx0.23$ and $\gamma_\text{Haar}\approx0.5$. We average over 100 random initializations.
	}
	\label{fig:stabilizer}
\end{figure}

\end{document}